\documentclass{article}
\usepackage[utf8]{inputenc}
\usepackage{amsmath}
\usepackage{amsfonts}
\usepackage{amssymb,amsthm,xcolor,enumitem}
\usepackage{amscd}
\usepackage{tocvsec2}
\usepackage{fullpage}

\usepackage[pdftex,plainpages=false,hypertexnames=false,pdfpagelabels]{hyperref}
\newcommand{\arxiv}[1]{\href{http://arxiv.org/abs/#1}{\tt arXiv:\nolinkurl{#1}}}
\newcommand{\arXiv}[1]{\href{http://arxiv.org/abs/#1}{\tt arXiv:\nolinkurl{#1}}}

\newcommand{\googlebooks}[1]{(preview at \href{http://books.google.com/books?id=#1}{google books})}

\usepackage{xcolor}
\definecolor{dark-red}{rgb}{0.7,0.25,0.25}
\definecolor{dark-blue}{rgb}{0.15,0.15,0.55}
\definecolor{medium-blue}{rgb}{0,0,.8}
\definecolor{DarkGreen}{RGB}{0,150,0}
\definecolor{rho}{named}{red}
\hypersetup{
   colorlinks, linkcolor={purple},
   citecolor={medium-blue}, urlcolor={medium-blue}
}

\newcommand\numberthis{\addtocounter{equation}{1}\tag{\theequation}}

\newcommand{\rhobar}{\overline{\rho}}
\newcommand{\B}{\mathcal{B}}
\newcommand{\F}{\mathcal{F}}
\newcommand{\M}{\mathcal{M}}
\renewcommand{\H}{\mathcal{H}}

\newcommand{\C}{\mathbb{C}}
\newcommand{\Id}{\mathbf{1}}

\newcommand{\U}{\mathcal{U}}
\newcommand{\CC}{\mathcal{C}}

\newcommand{\A}{\mathbb{A}}

\newcommand{\PV}{\operatorname{\mathbf P.V.}~}

\newcommand{\R}{\mathbb{R}}
\newcommand{\cR}{\mathcal{R}}
\newcommand{\Z}{\mathbb{Z}}

\newcommand{\cl}{\operatorname{cl}}

\newcommand{\interior}[1]{\mathring{#1}}
\newcommand{\abs}[1]{\left|#1\right|}
\newcommand{\norm}[1]{\left\|#1\right\|}
\newcommand{\ip}[1]{\left\langle#1\right\rangle}
\newcommand{\D}{\mathbb{D}}
\renewcommand{\O}{\mathcal{O}}
\renewcommand{\P}{\mathcal{P}}
\newcommand{\bbP}{\mathbb{P}}
\newcommand{\Diff}{\operatorname{Diff}}
\newcommand{\tr}{\operatorname{tr}}

\newcommand{\grotimes}{\hat \otimes}
\newcommand{\CAR}{\operatorname{CAR}}
\newcommand{\Rot}{\operatorname{Rot}}
\newcommand{\Aut}{\operatorname{Aut}}

\newcommand{\End}{\operatorname{End}}

\newtheorem{thmalpha}{Theorem}

\newtheorem{Theorem}{Theorem}[section]
\newtheorem*{Theorem*}{Theorem}
\newtheorem{Lemma}[Theorem]{Lemma}
\newtheorem{Proposition}[Theorem]{Proposition}
\newtheorem{Corollary}[Theorem]{Corollary}
\theoremstyle{definition}
\newtheorem{Notation}[Theorem]{Notation}
\newtheorem{Remark}[Theorem]{Remark}
\newtheorem{Definition}[Theorem]{Definition}
\newtheorem{Example}[Theorem]{Example}

\newtheorem{Caution}[Theorem]{Caution}

\numberwithin{equation}{section}

\title{Construction of the unitary free fermion Segal CFT}
\author{James E. Tener}

\begin{document}

\maketitle

\abstract{In this article, we provide a detailed construction and analysis of the mathematical conformal field theory of the free fermion, defined in the sense of Graeme Segal.
We verify directly that the operators assigned to disks with two disks removed correspond to vertex operators, and use this to deduce analytic properties of the vertex operators.
One of the main tools used in the construction is the Cauchy transform for Riemann surfaces, for which we establish several properties analogous to those of the classical Cauchy transform in the complex plane.
}

\bigskip\bigskip

\tableofcontents

\newpage

\settocdepth{section}
\section{Introduction}
In \cite{SegalDef}, Graeme Segal proposed a new mathematical definition of conformal field theory.
Under Segal's definition, a conformal field theory is a projective, monoidal functor from the cobordism category of closed $1$-manifolds and Riemann surfaces to the category of Hilbert spaces and trace class maps, subject to certain additional axioms.
We call conformal field theories in this spirit ``Segal CFTs.''

In \cite[\S 8]{SegalDef}, Segal describes the first examples of Segal CFTs, the charged chiral fermion theories (often called $b$-$c$ systems in physics).
In particular, there is one \emph{unitary} charged chiral fermion theory, which we will simply call \emph{the free fermion Segal CFT.}

Many authors have explored mathematical aspects of the free fermion Segal CFT.
The most detailed study is \cite{Kriz03}, in which Kriz studies the projective anomaly and partition functions of a class of conformal field theories which includes the free fermion.
While results concerning the analytic aspects of the construction have appeared (e.g. \cite{PrSe86,Po03}), to our knowledge there has never been a complete, rigorous analysis of the trace class operators assigned to surfaces with boundary.
The first purpose of this paper is to provide such a treatment (Sections \ref{secRiggedSpinRiemannSurfaces} and \ref{secSegalCFT}).

The second purpose of this paper is to establish concretely the connection between the free fermion vertex operator algebra and the free fermion Segal CFT. 
It has been understood for some time that the operators assigned by Segal CFTs to spheres with three holes should correspond to vertex operators, after slight modification.
This connection was used by Huang \cite{Huang03} to construct Segal CFTs in genus zero from a general class of vertex operator algebras, but in the context of topological vector spaces as opposed to Hilbert spaces.
In Section \ref{secVertexOperators}, we provide an explicit formula in terms of vertex operators for the operators assigned to a disk with two disks removed by the free fermion Segal CFT.
As a consequence of this formula, we are able to deduce analytic properties of the vertex operators (see Theorem \ref{thmIntroVertexOperators}).

We now summarize the main results.

Let $\Sigma$ be a compact Riemann surface with boundary, with no closed components.
One slight complication of the free fermion Segal CFT is that it is a \emph{spin conformal field theory}, so we must assume that $\Sigma$ is equipped with a spin structure.
That is, we assume we have a holomorphic line bundle $L \to \Sigma$, and an isomorphism $\Phi: L \otimes L \to K_\Sigma$, where $K_\Sigma$ is the holomorphic cotangent bundle.
We also assume that the boundary of $\Sigma$ comes with a family of parametrizations $\beta$ from the two standard spin structures on the unit circle $S^1$.
The collection of data $X=(\Sigma, L, \Phi, \beta)$ is called a spin Riemann surface with parametrized boundary, and we use $\cR$ to denote the collection of all such $X$.

We assign to each boundary component of $\Sigma$ the fermionic Fock space $\F$ assigned to the unit circle $S^1 \subset \C$ and the disk $\D$ that it bounds.
That is, if $H = L^2(S^1)$ and $H^2(\D)$ is the classical Hardy subspace, we define $\F$ to be the exterior Hilbert space 
$$
\F = \Lambda (\overline{H^2(\D)} \oplus H^2(\D)^\perp),
$$
which is a super Hilbert space.
Fermionic Fock space comes equipped with a representation of $\CAR(H)$, the $C^*$ algebra generated by annihilation and creation operators $a(f)$ and $a(f)^*$, for $f \in H$.

Now let $X \in \cR$. The boundary $\Gamma$ of $\Sigma$ is partitioned into incoming boundary components, $\Gamma^0$, on which the parametrizing map $\beta$ is orientation reversing, and outgoing boundary components, $\Gamma^1$, on which $\beta$ is orientation preserving.
We define the Hardy space 
$$
H^2(X) \subseteq \left(\bigoplus_{j \in \pi_0(\Gamma^1)} L^2(S^1) \right) \oplus \left(\bigoplus_{j \in \pi_0(\Gamma^0)} L^2(S^1) \right)
$$
to be the closure of holomorphic sections of $X$, pulled back to $L^2(S^1)$ by the boundary parametrizations $\beta$.

The free fermion Segal CFT assigns to $X$ the second quantization of the Hardy space $H^2(X)$. That is, it assigns the space $E(X)$ of trace class maps $T_X : \bigotimes_{j \in \pi_0(\Gamma^0)} \F \to \bigotimes_{j \in \pi_0(\Gamma^1)} \F$ which satisfy the \emph{$H^2(X)$ commutation relations} with the annihilation and creation operators:
\begin{equation}\label{eqnIntroCommRels1}
a(f^1)T_X = T_Xa(f^0), \qquad \mbox{ for all } (f^1, f^0) \in H^2(X)
\end{equation}
and
\begin{equation}\label{eqnIntroCommRels2}
a(g^1)^*T_X = -T_Xa(g^0)^*, \qquad \mbox{ for all } (g^1, g^0) \in H^2(X)^\perp.
\end{equation}
If $\Sigma$ has both incoming and outgoing boundary, the $H^2(X)$ commutation relations are equivalent to $T_X$ implementing the unbounded operator whose graph is $H^2(X)$ as a Bogoliubov-like endomorphism of $\CAR(H)$.

The basic properties of the assignment $X \mapsto E(X)$ are summarized in Theorem \ref{thmIntroCFTProperties} below, which is stated more precisely as Theorem \ref{thmCFTProperties} in the body of the paper.
\begin{thmalpha}
\label{thmIntroCFTProperties}
Let $X \in \cR$. The maps $E(X)$ assigned by the free fermion Segal CFT satisfy the following properties:
\begin{enumerate}
\item  (Existence) $E(X)$ is one-dimensional, and its elements are homogeneous and trace class.
\item (Non-degeneracy) If every connected component of $\Sigma$ has an outgoing boundary component, then non-zero elements of $E(X)$ are injective. If every connected component of $\Sigma$ has an incoming boundary component, then non-zero elements of $E(X)$ have dense image.
\item (Monoidal) If $Y \in \cR$, then $E(X \sqcup Y) = E(X) \grotimes E(Y)$, where $X \sqcup Y$ is the disjoint union and $\grotimes$ is the graded tensor product. 
\item (Sewing) If $\hat X \in \cR$ is obtained by sewing two boundary components of $X$ along the parametrizations, then the partial supertrace induces an isomorphism $\tr^s:E(X) \to E(\hat X)$. In particular, composition of cobordisms corresponds to composition of maps.
\item (Reparametrization) The Fock space $\F$ comes equipped with unitary representations of the automorphism groups of the standard spin structures on the circle, and the assignment $X \mapsto E(X)$ is covariant with respect to reparametrization of boundary components of $X$.
\item (Unitarity) $E(\overline{X}) = E(X)^*$, where $\overline{X}$ is the complex conjugate spin Riemann surface, and $E(X)^*$ denotes taking the adjoint elementwise.
\end{enumerate}
\end{thmalpha}
The proof of Theorem \ref{thmIntroCFTProperties} requires a careful study of the Hardy spaces $H^2(X)$. 
Our main tool for this is the Cauchy transform for Riemann surfaces, which we study in Section \ref{secCauchyTransform}.
In particular, we obtain analogs of the Plemelj formula and Kerzman-Stein formula.

The explicit description of $E(X)$ in terms of commutation relations \eqref{eqnIntroCommRels1} and \eqref{eqnIntroCommRels2} is useful for computing operators assigned to particular surfaces.
As a demonstration, we compute the operator assign to a disk with two disks removed, and identify the result with free fermion vertex operators, which we now describe in more detail.

The action of rotation on $S^1$ induces a one-parameter group of unitary operators acting on $\F$, which can be written as $e^{2 \pi i L_0 \theta}$ for a diagonalizable, positive operator $L_0$ with eigenvalues in $\tfrac12 \Z_{\ge 0}$ and finite-dimensional eigenspaces.
We let $\F^0$ denote the algebraic span of eigenvectors of $L^0$, which are called finite energy vectors.
The free fermion vertex operator algebra (often called the charged chiral fermion vertex operator algebra) provides a `state-field correspondence.'
That is, for every $\xi \in \F^0$, we have a formal power series
$$
Y(\xi, z) = \sum_{n \in \Z} \xi_n z^{-n-1},
$$
where $\xi_n \in \End(\F^0)$.

In general, the $\xi_n$ are closable operators on $\F$, but do not extend to bounded operators.
If one tries evaluating $Y(\xi, z)\eta$ with $\eta \in \F^0$ and $z$ a complex number instead of a formal variable, the resulting series will converge in $\F$ in general only when $\abs{z} < 1$. 
Even then, $Y(\xi, z)$ is not generally a closeable operator on $\F$; in fact, its adjoint may be defined only on the vector $0$.
However, we show in Theorem \ref{thmPantsAreVertexOperators} that the trace class operators assigned by the free fermion Segal CFT to disks with two disks removed are closely related to vertex operators.

\begin{thmalpha}
\label{thmIntroVertexOperators}
Let $\bbP_{w,r_1,r_2}$ be the Riemann surface obtained by removing from the closed unit disk the open disk of radius $r_1$ centered at $w$ and the open disk of radius $r_2$ centered at $0$.
Give $\bbP_{w,r_1,r_2}$ the spin structure obtained by its embedding into $\C$, and parametrize the boundary components via dilation and translation of the unit circle.
Then $E(\bbP_{w,r_1,r_2})$ is spanned by the map given on $\xi \otimes \eta \in \F^0 \otimes \F^0$ by
\begin{equation}
\xi \otimes \eta \mapsto Y(r_1^{L_0} \xi, w)r_2^{L_0}\eta = \sum_{n \in \Z} (r_1^{L_0} \xi)_n w^{-n-1} r_2^{L_0}\eta.
\end{equation}
The operators $(r_1^{L_0}\xi)_n$ extend to trace class operators on $\F$, and the sum
$$
\sum_{n \in \Z} (r_1^{L_0} \xi)_n r_2^{L_0}w^{-n-1}
$$
converges absolutely in operator norm, uniformly in $r_1, r_2,$ and $w$ on compact subsets of the configuration space of pairs of pants $\bbP_{w,r_1,r_2}$.
\end{thmalpha}

Most of the content of this paper is adapted from the author's Ph.D. thesis \cite{Ten14}.

\subsection{Acknowledgements}
We would like to thank Vaughan Jones for supporting this project while the author was a graduate student, and after.
We would also like to thank Antony Wassermann for his guidance throughout the project, especially in indicating the role of the Cauchy transform, and the relation between vertex operators and Segal CFT.
We would like to thank Andr\'e Henriques for many enlightening conversations on geometric conformal field theory.
We would also like to thank Dietmar Bisch, Terry Gannon, Yasuyuki Kawahigashi, Roberto Longo, Peter Teichner, and Feng Xu for helpful conversations.

We gratefully acknowledge the support and hospitality of the Max Planck Institute for Mathematics, Bonn.
This work was also supported in part by NSF DMS grant 0856316.

\settocdepth{subsection}
\section{Background}
\label{secBackground}

\subsection{Representations of $\CAR(H)$ and $\Diff(S^1)$}
\label{secFABackground}

\subsubsection{(Super) Hilbert spaces}

Let $H$ and $K$ be complex Hilbert spaces. We write $\B(H,K)$ for the Banach space of bounded linear maps $x: H \to K$, equipped with the operator norm. 
We write $\B_p(H,K)$ for the ideal of $\B(H,K)$ consisting of $x \in \B(H,K)$ which satisfy
$$
\norm{x}_p := \tr((x^*x)^{p/2})^{1/p} < \infty.
$$ 
Elements of $\B_1(H,K)$ are called trace class maps, and elements of $\B_2(H,K)$ are called Hilbert-Schmidt maps. 
The inner product $\ip{x,y} = \tr(y^*x)$ makes $\B_2(H,K)$ into a Hilbert space. 

When $H = K$ we simply write $\B(H)$ and $\B_p(H)$. 
In this case we define $\P(H)$ and $\U(H)$ to be the set of projections ($p^*=p^2=p$) and the group of unitary operators ($u^*=u^{-1}$) on $H$.

Trace class maps have a partial trace operation $\tr_L:\B_1(H \otimes L, K \otimes L) \to \B_1(H, K)$. 
The partial trace is continuous for the trace norms on $\B_1(H \otimes L, K \otimes L)$ and $B_1(H, K)$, and it is characterized by the property that if $x_1 \in \B_1(H,K)$ and $x_2 \in \B_1(L)$ then 
$$
\tr_L(x_1 \otimes x_2) = x_1 \tr(x_2).
$$ 
From this characterization one can deduce the tracial property 
$$
\tr_L((\Id_K \otimes x_2)y) = \tr_L(y(\Id_H \otimes x_2))
$$
for all $y \in \B_1(H \otimes L, K \otimes L)$.

A {\it super Hilbert space} is a Hilbert space $H$ with a $\Z/2$-grading, i.e. a decomposition $H = H^0 \oplus H^1$. Elements of $H^0$ (resp. $H^1$) are called even (resp. odd) homogeneous elements. A super Hilbert space comes with a grading involution $d_H$ which acts by $\Id$ on $H^0$ and by $-\Id$ on $H^1$.

The tensor product of super Hilbert spaces $H \otimes K$ is again a super Hilbert space, with 
$$
(H \otimes K)^0 := (H^0 \otimes K^0) \oplus (H^1 \otimes K^1), \quad (H \otimes K)^1 := (H^0 \otimes K^1) \oplus (H^1 \otimes K^0).
$$

Super Hilbert spaces have a symmetric braiding $H \otimes K \overset{\beta_{H,K}}{\to} K \otimes H$, given on homogeneous elements by
\begin{equation}\label{eqSuperGrading}
\beta_{H,K}(\xi \otimes \eta) = (-1)^{p(\xi)p(\eta)} \eta \otimes \xi,
\end{equation}
where $p(\xi),p(\eta) \in \{0,1\}$ are the parities. Since $\beta$ is symmetric, for every permutation $\sigma \in S_n$ we have unitary isomorphisms
$$
H_1 \otimes \cdots \otimes H_n \overset{\beta(\sigma)}{\to} H_{\sigma(1)} \otimes \cdots \otimes H_{\sigma(n)}
$$
compatible with composition in $S_n$.

The symmetric braiding allows us to talk about the unordered tensor product of super Hilbert spaces $\bigotimes_{i \in I} H_i$, over a finite index set $I$. 
A map of unordered tensor products 
$$
x: \bigotimes_{i \in I} H_i \to \bigotimes_{j \in J} H^\prime_j
$$
is defined to be a family of maps between every ordered tensor product of the $\{H_i\}$ and $\{H_j^\prime\}$, compatible with the braiding. 
That is, for every pair of bijections $\alpha:\{1, \ldots, \abs{I}\} \to I$ and $\alpha^\prime:\{1, \ldots, \abs{J}\} \to J$, we have a linear map
$$
x_{\alpha,\alpha^\prime} : H_{\alpha(1)} \otimes \cdots H_{\alpha(\abs{I})} \to H^\prime_{\alpha^\prime(1)} \otimes \cdots \otimes H^\prime_{\alpha^\prime(1)} \otimes  \cdots \otimes H^\prime_{\alpha^\prime(\abs{J})},
$$
and these maps should satisfy
$$
x_{\alpha_2, \alpha_2^\prime} = \beta((\alpha_1^\prime)^{-1} \circ \alpha_2^\prime) x_{\alpha_1, \alpha_1^\prime} \beta((\alpha_2)^{-1} \circ \alpha_1).
$$
for all bijections $\alpha_i:\{1, \ldots, \abs{I}\} \to I$ and all $\alpha_i^\prime: \{1, \ldots, \abs{J}\} \to J$.
There are obvious notions of sum, composition and tensor product of maps of unordred tensor products obtained by applying the operations to compatible representatives.

Note that every $x:H_1 \otimes \cdots \otimes H_n \to K_1 \otimes \cdots \otimes K_m$ is a representative of some map of unordered tensor products, corresponding to the family $\beta(\sigma^\prime)x\beta(\sigma)$, where $\sigma \in S_n$ and $\sigma^\prime \in S_m$. 
We refer to this as the map of unoriented tensor products associated to $x$, and will denote it again by $x$ when there is no risk of confusion.

If $H$ and $K$ are super Hilbert spaces, then $\B(H,K)$ has a $\Z/2$-grading corresponding to the involution $x \mapsto d_K x d_H$. We identify $\B(H_1 \otimes H_2, K_1 \otimes K_2)$ with the graded tensor product of algebras $\B(H_1,K_1) \grotimes \B(H_2, K_2)$ as follows. 

If $x_i \in \B(H_i,K_i)$, define 
$$
x_1 \grotimes x_2 := x_1 d_{H_1}^{p(x_2)} \otimes x_2 \in \B(H_1 \otimes H_2, K_1 \otimes K_2)
$$
if the $x_i$ are homogeneous, and by extending linearly otherwise. If $y_i \in \B(K_i, L_i)$ we have
$$
(y_1 \grotimes y_2)(x_1 \grotimes x_2) = (-1)^{p(y_2)p(x_1)}(y_1x_1 \grotimes y_2x_2).
$$

We denote by $H^*$ the continuous dual of $H$, and write $\xi \mapsto \xi^*$ for the canonical conjugate linear isomorphism. 

There is a natural isomorphism $\mu_{H,K}:K \otimes H^* \to B_2(H,K)$ given by
\begin{equation}\label{eqHilbertSchmidtIso}
\psi \otimes \eta^* \mapsto\ip{\;\cdot\;, \eta} \psi.
\end{equation}
Observe that we have adopted the convention that inner products are linear in the first variable.

There is a natural $\B(K)- \B(H^*)^{op}$ bimodule structure on $K \otimes H^*$, and a natural $\B(K)-\B(H)$ bimodule structure on $\B_2(H,K)$. 
We pause to observe an intertwining relation between these structures.

For $x \in \B(H,K)$, let $\overline{x} \in \B(H^*,K^*)$ be given by $\overline{x}\xi^* = (x\xi)^*$.
\begin{Proposition}\label{propMuBimodularity}
If $\xi \in K \otimes H^*$, $x \in \B(H)$ and $y \in \B(K)$, then 
$$
\mu_{H,K}((\Id \grotimes \overline{x})\xi) = d_K^{p(x)}\mu_{H,K}(\xi)x^*, \quad
\mu_{H,K}((y \grotimes \Id)\xi) = y \; \mu_{H,K}(\xi).
$$
\end{Proposition}
\begin{proof}
It suffices to check the relations when $\xi = \psi \otimes \eta^*$, when $\psi \in K$ and $\eta^* \in H^*$ are homogeneous vectors. We then have
\begin{align*}
\mu_{H,K}((\Id \grotimes \overline{x})(\psi \otimes \eta^*)) &= (-1)^{p(\psi)p(x)}\mu_{H,K} (\psi \otimes (x\eta)^*)\\
&= (-1)^{p(\psi)p(x)}\ip{ \; \cdot \;, x\eta}\psi\\
&= d_K^{p(x)}\ip{ x^* \; \cdot \;, \eta}\psi\\
&= d_K^{p(x)}\mu_{H,K}(\psi \otimes \eta^* )x^*
\end{align*}
which establishes the first relation. The second is calculated similarly:
\begin{align*}
\mu_{H,K}((y \grotimes \Id)(\psi \otimes \eta^*)) &= \mu_{H,K}(y\psi \otimes \eta^*)\\
&= \ip{\; \cdot \;, \eta}y\psi\\
&= y \; \mu_{H,K}(\psi \otimes \eta^*).
\end{align*}

\end{proof}

Define the {\it supertrace} $\tr^s :\B_1(H) \to \C$ by $\tr^s(x) = \tr(xd_H)$. Similarly, the partial supertrace 
$$
\tr^s_L : \B_1(H \otimes L, K \otimes L) \to \B_1(H,K)
$$
is defined by $\tr^s_L(x) = \tr_L(x(\Id \grotimes d_L))$. 

More generally, if 
$$
x \in \B_1(H_1 \otimes \cdots \otimes H_m, K_1 \otimes \cdots K_n)
$$
is a map of (ordered) tensor products and $H_{i^0} = K_{j^1} =: L$, then we define $\tr^s_{j^1i^0}(x)$ by using the braiding to move $H_{i^0}$ and $K_{j^1}$ all the way to the right, and then applying the definition of $\tr^s_L$ above. Specifically, let
$$
\beta: H_1 \otimes \cdots \otimes H_m \to H_1 \otimes \cdots \otimes H_{i^0 -1} \otimes H_{i^0 + 1} \otimes \cdots \otimes H_{m} \otimes H_{i^0}
$$
be the braiding, and similarly let $\beta^\prime$ be the braiding
$$
\beta^\prime: K_1 \otimes \cdots \otimes K_n \to K_1 \otimes \cdots \otimes K_{j^1 -1} \otimes K_{j^1 + 1} \otimes \cdots \otimes K_{n} \otimes K_{j^1}.
$$
Then we define
\begin{equation}\label{eqPartialSupertraceNotLast}
\tr^s_{j^1i^0}(x) := \tr^s_L(\beta^\prime x \beta^{-1} ).
\end{equation}

Now let $x:\bigotimes_{i \in I} H_i \to \bigotimes_{j \in J} K_j$ be a trace class map of \emph{unordered} tensor products, and fix $i^0 \in I$ and $j^1 \in J$ with $H_{i^0} = K_{j^1}=:L$. Then we can define a partial supertrace $\tr^s_{j^1i^0}(x)$ as a map of unordered tensor products
$$
\tr^s_{j^1i^0}(x): \bigotimes_{i \in I \setminus \{i^0\}} H_i \to \bigotimes_{j \in J \setminus \{j^1\}} K_j
$$
as follows. Given bijections $\alpha:\{1, \ldots, \abs{I}-1\} \to I \setminus \{i^0\}$ and $\alpha^\prime:\{1, \ldots, \abs{J}-1\} \to I \setminus \{j^1\}$, extend them to orderings $\tilde \alpha$ and $\tilde \alpha^\prime$ of $I$ and $J$, respectively, by putting $i^0$ and $j^1$ last.  
Now set
$$
\tr^s_{j^1i^0}(x)_{\alpha,\alpha^\prime} := \tr^s_{L}(x_{\tilde \alpha, \tilde \alpha^\prime}).
$$
It is straightforward to check that $\tr^s_{j^1i^0}(x)$ is a map of unordered tensor products, i.e., the maps $\tr_{j^1i^0}(x)_{\alpha,\alpha^\prime}$ satisfy the appropriate compatibility with the braiding.

Straightforward computation yields the following basic properties of the partial supertrace.
\begin{Proposition}\label{propSupertraceProperties} 
Let $x \in \B_1(H \otimes L, K \otimes L)$.
\begin{enumerate}
\item \label{itmSupertraceBimodular} If $y_1 \in \B(M, H)$ and $y_2 \in \B(K, M)$, then 
$$
\tr^s_L(x)y_1 = \tr^s_L(x(y_1 \grotimes \Id)), \quad \mbox{ and } \quad y_2\tr^s_L(x) = \tr^s_L((y_2 \grotimes \Id)x).
$$
\item \label{itmSupertracial} If $z \in \B(L)$, then 
$$
\tr^s_L((1 \grotimes z)x) = (-1)^{p(x)} \tr^s_L(x(1 \grotimes z))= (-1)^{p(x)p(z)} \tr^s_L(x(1 \grotimes z)).
$$
\end{enumerate}
\end{Proposition}

The partial supertrace also enjoys the expected associativity property.
\begin{Proposition}\label{propIteratedPartialTrace}
Let $x \in \B_1 (H \otimes L_1 \otimes L_2, K \otimes L_1 \otimes L_2)$. Then $\tr^s_{L_1 \otimes L_2}(x) = \tr^s_{L_1} \tr^s_{L_2}(x)$.
\end{Proposition}

Finally, we observe that the partial supertrace implements composition of maps of unordered tensor products.
\begin{Proposition}\label{propSupertraceComposition}
Let $x_1 \in \B_1(H, K \otimes L)$ and $x_2 \in \B_1(L \otimes M, N)$. We then have the identity of maps of unordered tensor products $\tr^s_L(x_2 \grotimes x_1) = (x_2 \grotimes \Id_K) \circ (\Id_M \grotimes x_1)$.
\end{Proposition}
\begin{proof}
Note that $x_2 \grotimes x_1 \in \B_1(L \otimes M \otimes H, N \otimes K \otimes L)$, and so the partial super trace $\tr^s_L(x_2 \grotimes x_1)$ is defined by precomposing with a braiding as in \eqref{eqPartialSupertraceNotLast}. 
That is, 
$$
\tr^s_L(x_2 \grotimes x_1) =\tr^s_L( (x_2 \grotimes x_1)\beta) = \tr_L( (x_2 \grotimes x_1)\beta(\Id_{M \otimes H} \otimes d_L)),
$$
where $\beta: M \otimes H \otimes L \to L \otimes M \otimes H$ is the braiding.

Also observe that $x_2 \grotimes \Id_K \in \B(L \otimes M \otimes K, N \otimes K)$ and $\Id_M \grotimes x_1 \in \B(M \otimes H, M \otimes K \otimes L)$. Thus a representative of the composition of maps of unordered tensor products $(x_2 \grotimes \Id_K) \circ (\Id_M \grotimes x_1)$ is given by $(x_2 \grotimes \Id_K)\beta^\prime(\Id_M \grotimes x_1)$, where
$$
\beta^\prime: M \otimes K \otimes L \to L \otimes M \otimes K
$$
is the braiding.

In light of the preceding discussion, we must prove that
\begin{equation}\label{eqSupertraceCompositionLiteral}
\tr_L( (x_2 \grotimes x_1)\beta(\Id_{M \otimes H} \otimes d_L)) = (x_2 \grotimes \Id_K)\beta^\prime(\Id_M \grotimes x_1)
\end{equation}
for all $x_1 \in \B_1(H, K \otimes L)$ and $x_2 \in \B_1(L \otimes M, N)$.

By the continuity of the partial trace, it suffices to check \eqref{eqSupertraceCompositionLiteral} when $x_1$ is given by 
$$
x_1(\eta) = y_1(\eta) \otimes \lambda_1
$$
for some homogeneous $y_1 \in \B_1(H,K)$ and a homogeneous $\lambda_1 \in L$. 
Similarly, we assume 
$$
x_2(\lambda \otimes \mu) = \ip{\lambda, \lambda_0} y_2(\mu)
$$
for a homogeneous $y_2 \in \B_1(M, N)$ and a homogeneous $\lambda_0 \in L$. 

For $\mu \in M$, $\eta \in H$ and $\lambda \in L$, we have
\begin{align*}
(x_2 \grotimes x_1)\beta(\Id_{M \otimes H} \otimes d_L)(\mu \otimes \eta \otimes \lambda) &= (x_2 d_{L \otimes M}^{p(x_1)} \otimes x_1)\beta(\Id_{M \otimes H} \otimes d_L)(\mu \otimes \eta \otimes \lambda)\\
 &=(-1)^{p(\lambda)(1+p(\eta)+p(\mu)) + p(x_1)(p(\lambda)+p(\mu))}(x_2 \otimes x_1)(\lambda \otimes \mu \otimes \eta)\\
&=(-1)^{p(\lambda)(1+p(\eta)+p(\mu)) + p(x_1)(p(\lambda)+p(\mu))} \ip{\lambda,\lambda_0}y_2(\mu)\otimes y_1(\eta) \otimes \lambda_1.
\end{align*}
Hence 
\begin{equation}\label{eqSupertraceCompositionComputation1}
\tr_L( (x_2 \grotimes x_1)\beta(\Id_{M \otimes H} \otimes d_L))(\mu \otimes \eta) =(-1)^{p(\lambda_1)(1+p(\eta)+p(\mu)) + p(x_1)(p(\lambda_1)+p(\mu))} \ip{\lambda_1, \lambda_0} y_2(\mu) \otimes y_1(\eta).
\end{equation}

On the other hand, 
\begin{align}\label{eqSupertraceCompositionComputation2}
(x_2 \grotimes \Id_K)\beta^\prime(\Id_M \grotimes x_1)(\mu \otimes \eta) &= (x_2 \otimes \Id_K)\beta^\prime(d_M^{p(x_1)} \otimes x_1)(\mu \otimes \eta)\nonumber\\
&=(-1)^{p(x_1)p(\mu)} (x_2 \otimes \Id_K)\beta^\prime(\mu \otimes y_1(\eta) \otimes \lambda_1)\nonumber\\
&=(-1)^{p(x_1)p(\mu)+p(\lambda_1)(p(y_1\eta)+p(\mu))} \ip{\lambda_1, \lambda_0} y_2(\mu) \otimes y_1(\eta).
\end{align}

It is clear that \eqref{eqSupertraceCompositionComputation1} and \eqref{eqSupertraceCompositionComputation2} agree up to sign.  

We can simplify the sign in \eqref{eqSupertraceCompositionComputation1} by working mod $2$, and we get
\begin{align}\label{eqSupertraceCompositionSign1}
p(x_1)p(\mu)+p(\lambda_1)(p(y_1\eta)+p(\mu)) &= (p(y_1)+p(\lambda_1))p(\mu) + p(\lambda_1)(p(y_1) + p(\eta) + p(\mu))\nonumber\\
&= p(\lambda_1)(p(y_1)+p(\eta)) + p(y_1)p(\mu).
\end{align}
On the other hand, simplifying the sign in \eqref{eqSupertraceCompositionComputation2} yields
\begin{align}\label{eqSupertraceCompositionSign2}
p(\lambda_1)(1+p(\eta)+p(\mu)) + p(x_1)(p(\lambda_1)+p(\mu)) &= p(\lambda_1)(1+p(\eta)+p(\mu)) + (p(\lambda_1)+p(y_1))(p(\lambda_1)+p(\mu))\nonumber\\
&= p(\lambda_1)(p(y_1)+p(\eta)) + p(y_1)p(\mu).
\end{align}
Since \eqref{eqSupertraceCompositionSign1} and \eqref{eqSupertraceCompositionSign2} agree, the signs in \eqref{eqSupertraceCompositionComputation1} and \eqref{eqSupertraceCompositionComputation2} agree, and thus we have established \eqref{eqSupertraceCompositionLiteral}, as desired.
\end{proof}

\subsubsection{Fermionic Fock space}\label{subsecFockSpace}

Given a complex Hilbert space $H$, the $*$-algebra $\CAR(H)$ is the universal unital $C^*$-algebra with generators $a(f)$ for $f \in H$ which are linear in $f$ and satisfy the canonical anticommutation relations
\begin{align*}
a(f)a(g) + a(g)a(f) &= 0,\\
a(f)a(g)^* + a(g)^* a(f) &= \ip{f,g} \Id.
\end{align*}

\begin{Remark}
The reader is welcome to replace $C^*$-algebra with $*$-algebra in the above definition with no loss of information, since the algebraic version has a unique $C^*$-norm.
\end{Remark}

There is an irreducible, faithful representation of $\CAR(H)$ on the Hilbert space
$$
\Lambda H = \bigoplus_{k =0}^\infty \Lambda^k H
$$
densely defined by $a(f)\zeta = f \wedge \zeta$. These operators are bounded, and $\norm{a(f)} = \norm{f}$. The exterior Hilbert space $\Lambda H$ is naturally a super Hilbert space, with $\Z/2$-grading inherited from the number grading. That is,
\begin{equation}\label{eqFockGrading}
(\Lambda H)^i = \bigoplus_{k =0}^\infty \Lambda^{2k+i} H.
\end{equation}
The subspace $\Lambda^0 H$ is spanned by a distinguished unit vector $\Omega$ which satisfies $a(f)^*\Omega = 0$ for all $f \in H$. 

There is a family of irreducible, faithful representations of $\CAR(H)$ indexed by $p \in \mathcal{P}(H)$ given as follows. 
Let $H_p = (pH)^* \oplus (\Id-p)H$, and define the representation $\pi_p:\CAR(H) \to \B(\Lambda H_p)$ by
$$
\pi_p(a(f)) = a((pf)^*)^* + a((\Id-p)f).
$$
We call $\Lambda H_p$ \textit{fermionic Fock space}, and denote it by $\mathcal{F}_{H,p}$, or simply $\mathcal{F}_p$ or $\F$ when the decorations are clear from context. Note that $\pi_p(a(f))$ is an odd operator on $\mathcal{F}_{H,p}$.

The distinguished unit vector $\Omega_p \in \Lambda^0 H_p$ is characterized, up to scalar multiples, by the equations
\begin{align}
\pi_p(a(f))\Omega_p &= 0 & \mbox{ for } & f \in pH,\label{eqnVacEqn1}\\
\pi_p(a(g))^*\Omega_p &= 0 & \mbox{ for } & g \in (\Id-p)H \label{eqnVacEqn2} .
\end{align}
In fact, the representation $(\mathcal{F}_p, \pi_p)$ is characterized up to unitary equivalence by the existence of a cyclic vector satisfying these equations (via the GNS construction).
The relations \eqref{eqnVacEqn1} and \eqref{eqnVacEqn2} are called ``vacuum equations.'
\begin{Definition}\label{defVacuumEquations}
Let $(\mathcal{K},\pi)$ be a representation of $\CAR(H)$, and let $q$ be a projection on $H$. A vector $\tilde \Omega_q \in \mathcal{K}$ is said to satisfy the $q$-vacuum equations if
\begin{align*}
\pi(a(f))\tilde \Omega_q &= 0 & \mbox{ for } & f \in qH,\\
\pi(a(g))^*\tilde \Omega_q &= 0 & \mbox{ for } & g \in (\Id-q)H .
\end{align*}
\end{Definition}

A crucial property of the Fock space construction is that it takes (unordered) direct sums to (unordered) tensor products.
\begin{Proposition}\label{propFockSumToTensor}
As super Hilbert spaces, we have natural isomorphisms $$\mathcal{F}_{H \oplus K, p \oplus q} \cong \mathcal{F}_{H,p} \otimes \mathcal{F}_{K,q}.$$ 
The isomorphism $\mathcal{F}_{H,p} \otimes \mathcal{F}_{K,q} \to \mathcal{F}_{K,q} \otimes \mathcal{F}_{H,p}$ induced by $H \oplus K \cong K \oplus H$ is the braiding of super Hilbert spaces.
The induced action of $\CAR(H \oplus K)$ on $\mathcal{F}_{H,p} \otimes \mathcal{F}_{K,q}$ is
\begin{equation}\label{eqTensorCARReps}
a(h+k) \mapsto \pi_p(a(h)) \grotimes \Id + \Id \grotimes \pi_q(a(k)).
\end{equation}
\end{Proposition}

\begin{Remark}\label{rmkFockUnorderedTensorProduct}
The naturality of the isomorphisms from Proposition \ref{propFockSumToTensor} make $\F_{H \oplus K, p \oplus q}$ a model for the unordered tensor product. That is, maps to and from $\F_{H \oplus K, p \oplus q}$ are equivalent to maps to and from the unordered tensor product $\bigotimes \left\{\F_{H,p}, \F_{K,q}\right\}$. As a result, we will not distinguish between $\F_{H \oplus K, p \oplus q}$ and $\bigotimes \left\{\F_{H,p}, \F_{K,q}\right\}$.  We will freely identify $\pi_{p \oplus q}$ and the representation given in equation \eqref{eqTensorCARReps}.
\end{Remark}

Since $H_{\Id-p} = H_p^*$, we have a natural unitary $\Phi: \F_{\Id-p} \to \F_p^*$ given by
$$
\Phi(\xi_1^* \wedge \cdots \wedge \xi_n^*) = (\xi_n \wedge \cdots \wedge \xi_1)^*
$$
for $\xi_i \in H_p$.
\begin{Proposition}\label{propPhiConjugation}
For all $f \in H$ we have 
$$\Phi \pi_{1-p}(a(f))\Phi^* = \overline{\pi_p(a((2p-\Id)f))^*}d_{\F_{p}^*}$$
and
$$\Phi \pi_{1-p}(a(f))^*\Phi^* = -\overline{\pi_p(a((2p-\Id)f))}d_{\F_p^*}.$$
\end{Proposition}
\begin{proof}
The two identities are clearly equivalent for every fixed $f \in H$. We prove the first for $f \in pH$ and the second for $f \in (\Id-p)H$. 

If $f \in pH$, then the first identity reads 
$$
\Phi \pi_{1-p}(a(f))\Phi^* = \overline{\pi_p(a(f))^*}d_{\F_{p}^*}.
$$ 
Applying the left-hand side to $\omega^* \in (\Lambda^n H_p)^*$ yields $(\omega \wedge f^*)^*$, and applying the right-hand side yields $(-1)^n (f^* \wedge \omega)^*$. The proof of the second identity when $f \in (\Id-p)H$ is similar.
\end{proof}

The natural question of when $\pi_p$ and $\pi_q$ are unitarily equivalent is answered by the following theorem.
\begin{Theorem}\label{thmSegalEq}
The following are equivalent:
\begin{enumerate}[label=(\roman*)]
\item \label{itmSegeq1}$(\F_{H,p}, \pi_p)$ and $(\F_{H,q},\pi_q)$ are unitarily equivalent representations of $\CAR(H)$.
\item \label{itmSegeq2}There exists a unit vector $\tilde \Omega_q \in \F_{H,p}$, which will be unique up to phase, satisfying the $q$-vacuum equations.
\item \label{itmSegeq3}$p - q$ is a Hilbert-Schmidt operator on $H$.
\end{enumerate}
\end{Theorem}
This result is often called the Shale-Stinespring equivalence condition, and there are many proofs in the literature. A simple version of the argument may be found in the textbook \cite[Thm. 10.7]{Th92}. A more concise version of the constructive proof that \ref{itmSegeq3} implies \ref{itmSegeq2} and \ref{itmSegeq1} is in \cite[\S3]{Wa98}, and an abstract proof using von Neumann algebra techniques is given in \cite[Thm.  8.23]{dlHJ}.

If $u \in \mathcal{U}(H)$,  the \textit{Bogoliubov automorphism} $\alpha_u$ of $\CAR(H)$ is characterized by $\alpha_u(a(f)) = a(uf)$. We say that an automorphism $\alpha$ of a $C^*$-algebra $A$ is \textit{implemented} in a representation $\pi: A \to \mathcal{B}(\H_\pi)$ if there is a unitary $U \in \mathcal{U}(\H_\pi)$ such that $\operatorname{Ad} U \circ \pi = \pi \circ \alpha$. 
If $\pi$ is irreducible then an implementing unitary $U$ will be unique up to phase.

\begin{Corollary}\label{corImplementationOfBogoluibov}
The Bogoliubov automorphism $\alpha_u$ is implemented in $\pi_p$ if and only if $[u,p]$ is Hilbert-Schmidt. If $\alpha_u$ is implemented by $U$, then $U\Omega_p = \tilde \Omega_q$ for $q = upu^*$. In particular, $\Omega_p$ is an eigenvector for $U$ if and only if $[u,p] = 0$. 
\end{Corollary}

\begin{Definition}\label{defUres}
Define the \textit{restricted general linear group} 
$$
GL_{res}(H,p) = \{x \in GL(H) : [x,p] \in \B_2(H) \}.
$$
and the \textit{restricted unitary group} 
$$
\U_{res}(H,p) = GL_{res}(H,p) \cap \U(H).
$$
\end{Definition}

We give $\U_{res}$ the topology generated by the strong operator topology, along with the pseudometric
$
\norm{[u-v,p]}_{2}.
$
With this topology, $\U_{res}$ is a topological group, but we will not need this fact.

In light of Corollary \ref{corImplementationOfBogoluibov}, there is a natural projective representation of $\U_{res}(H,p)$ on $\F_{H,p}$ called the \textit{basic representation}, which we will write $u \mapsto U$. The basic representation is characterized by
\begin{equation}\label{eqnBasicRep}
U\pi_p(a(f))U^* = \pi_p(a(uf))
\end{equation}
for all $f \in H$. The basic representation restricts to an honest representation on the subgroup of unitary operators $u$ commuting with $p$.  On this subgroup, a lift to $\U(\F_{H,p})$ is given by choosing $U$ so that $U\Omega = \Omega$.

\begin{Theorem}\label{thmBasicRepContinuous}
The basic representation is strongly continuous (i.e. continuous as a map into the projective unitary group $\mathcal{PU}(\F_{H,p})$ given the quotient topology of the strong operator topology).
\end{Theorem}
A proof of this theorem is given in \cite[\S 3]{Wa98}.

Note that the grading operator $d_{\F_{H,p}}$ for the $\Z/2$ grading on $\F_{H,p}$ given by \eqref{eqFockGrading} implements the Bogoliubov automorphism $\alpha_{-\Id}$. We will simply write $d$ for this grading operator when the Fock space that it acts on is clear.

\begin{Proposition}\label{propVacuumHomogeneous}
The vectors $\tilde \Omega_q$ from Theorem \ref{thmSegalEq} are homogeneous. The parity of $\tilde \Omega_q$ is the parity of $\dim \left(pH \cap (1-q)H\right) + \dim \left((1-p)H \cap qH\right)$.
\end{Proposition}
\begin{proof}
The homogeneity of $\tilde \Omega_q$ follows immediately from the fact that $d \tilde \Omega_q$ again satisfies the $q$-vacuum equations, and thus $\tilde \Omega_q$ is an eigenvector for the grading operator. The parity can be read off from an explicit formula for $\tilde \Omega_q$ (see e.g. \cite[\S 3]{Wa98} or \cite[Thm 10.6]{Th92}).
\end{proof}

The following proposition is an immediate corollary.

\begin{Proposition}\label{propBasicRepHomogeneous}
If $U$ implements the Bogoliubov automorphism $\alpha_u$ in $\F_{H,p}$, then $U$ is homogeneous. The parity of $U$ is the same as the parity of $\tilde \Omega_q$, where $q = upu^*$.
\end{Proposition}

\subsubsection{Representations of $\Diff(S^1)$}\label{subsecDiffReps}

We will use fermionic Fock space $\F_{H,p}$ primarily in the case where $H=L^2(S^1)$ and $pH$ is the Hardy space $H^2(\D)$. Here $S^1$ is the unit circle in $\C$ and
$$
H^2(\D) = \cl(\operatorname{span} \{ z^n : n \ge 0\}).
$$

Let $\Diff(S^1)$ be the group of diffeomorphisms of the circle, and let $\Diff_+(S^1)$ and $\Diff_-(S^1)$ be the orientation preserving and orientation reversing diffeomorphisms, respectively. If $\gamma \in \Diff(S^1)$, define $\epsilon(\gamma) = \pm 1$ if $\gamma \in \Diff_\pm(S^1)$. 

If $f:S^1 \to \C$ is a smooth function, then we define the complex derivative $f^\prime: S^1 \to \C$ by
$$
f'(z) := \frac{1}{iz} \left(\left.\frac{d}{d\theta} f(e^{i\theta})\right|_{e^{i\theta}=z}\right).
$$
Of course, if $f$ extends to a holomorphic function in a neighborhood of $S^1$ then this definition agrees with the usual complex derivative. 

We now define a pair of central extensions of $\Diff(S^1)$ by $\Z/2$, which are the groups of  Ramond and Neveu-Schwarz spin diffeomorphisms. They are given as subgroups of $C^\infty(S^1)^\times \rtimes \Diff(S^1)$ by
\begin{align*}
\Diff^{NS}(S^1) &:= \{ (\psi,\gamma) \in C^\infty(S^1)^\times \rtimes \Diff(S^1) : \psi^2 = (\gamma^{-1})^\prime\},\\
\Diff^{R}(S^1) &:= \{ (\psi,\gamma) \in C^\infty(S^1)^\times \rtimes \Diff(S^1) : \psi^2 = \epsilon(\gamma)\abs{(\gamma^{-1})^\prime}\}.
\end{align*}

In the following, let $\sigma \in \{NS, R\}$.  One can see that the $\Diff^\sigma(S^1)$ are non-isomorphic central extensions, since $\Diff_+^R(S^1)$ is a split extension of $\Diff_+(S^1)$ and $\Diff_+^{NS}(S^1)$ is not.

Define the spinor representations $u_\sigma:\Diff^\sigma(S^1) \to \U(H)$ by
$$
u_\sigma(\psi, \gamma)f = \psi \cdot (f \circ \gamma^{-1}).
$$

\begin{Proposition}\label{propDiffsQuantized}
For $\sigma \in \{NS,R\}$, $u_\sigma(\Diff^\sigma_+(S^1)) \subset \U_{res}(H,p)$. If $\Diff^\sigma_+(S^1)$ is given the $C^3$ topology then the embedding of $\Diff^\sigma_+(S^1)$ into $\U_{res}(H,p)$ is continuous.
\end{Proposition}
\begin{proof}
It is clear that if $\sigma_n \to \sigma$ in the $C^3$ topology then $u_{\sigma_n} \to u_\sigma$ in the strong operator topology.
It remains to show that $u_\sigma(\Diff^\sigma_+(S^1)) \subset \U_{res}(H,p)$, and that $\norm{[u_{\sigma_n} - u_\sigma, p]}_2 \to 0$.

The first assertion is proven in \cite[Prop. 5.3]{Se81}. 
To see the second, observe that 
$$
\norm{[u_{\sigma_n} - u_\sigma, p]}_2 \le \norm{[u_{\sigma_n \circ \sigma^{-1}}, p]}_2 + \norm{(u_\sigma - u_{\sigma_n})[u_{\sigma},p]}_2.
$$
Since $\sigma_n \circ \sigma^{-1} \to \operatorname{id}$ in the $C^3$ topology, one may apply the estimate from the proof of \cite[Prop. 5.3]{Se81} to see that $\norm{[u_{\sigma_n \circ \sigma^{-1}}, p]}_2 \to 0$.
On the other hand, from \cite[Prop. 5.3]{Se81} one can also see that $[u_\sigma, p]$ is trace class, and thus can be factored as a product of Hilbert-Schmidt operators, say $[u_\sigma, p] = xy$. Then $(u_\sigma - u_{\sigma_n})x \to 0$ in operator norm, and so 
$$
(u_\sigma - u_{\sigma_n})[u_\sigma, p] = (u_\sigma - u_{\sigma_n})xy \to 0
$$
in Hilbert-Schmidt norm.
\end{proof}

In light of Proposition \ref{propDiffsQuantized}, one has a pair of projective representations $U_\sigma: \Diff_+^\sigma(S^1) \to \U(\F_{H,p})$ by composing $u_\sigma$ with the basic representation of $\U_{res}(H,p)$.

\begin{Corollary}\label{corDiffsEven}
The representations $U_\sigma$ are strongly continuous, and $U_\sigma(\psi,\gamma)$ is even for all $(\psi,\gamma) \in \Diff_+^\sigma(S^1)$.
\end{Corollary}
\begin{proof}
Combining Proposition \ref{propDiffsQuantized} with the continuity of the basic representation (Theorem \ref{thmBasicRepContinuous}) shows that $U_\sigma$ is strongly continuous. By Proposition \ref{propBasicRepHomogeneous} each $U_\sigma(\psi,\gamma)$ is homogeneous. Any $(\psi, \gamma) \in \Diff_+^\sigma(S^1)$ can be connected via a path to $(1, id)$ or $(-1, id)$, and since $U_\sigma(1, id)$ and $U_\sigma(-1, id)$ are even, so is $U_\sigma(\psi,\gamma)$ for every $(\psi, \gamma) \in \Diff_+(S^1)$.
\end{proof}

\begin{Remark}
If $(\psi,\gamma) \in \Diff^\sigma_-(S^1)$, then $u_\sigma(\psi,\gamma) pu_\sigma(\psi,\gamma)^* - (\Id-p)$ is Hilbert-Schmidt, and consequently one can define projective unitaries $U_\sigma(\psi,\gamma): \F_{H,p} \to \F_{H,p}^*$ which are compatible with the action of orientation preserving spin diffeomorphisms on $\F_{H,p}$ and $\F_{H,p}^*$.
\end{Remark}

Let $r_\theta \in \Diff_+(S^1)$ be counterclockwise rotation by $\theta$. Since $u_R(1,r_\theta)$ and $u_{NS}(e^{-i \theta/2}, r_\theta)$ commute with $p$, we obtain a pair of one parameter (honest) unitary groups acting on $\F_{H,p}$, namely
$$
\Rot_R(\theta) := U_R(1,r_\theta), \qquad \Rot_{NS}(\theta) := U_{NS}(e^{-i \theta/2}, r_\theta).
$$
By Stone's theorem we can find self-adjoint operators $L_0^R$ and $L_0^{NS}$ such that
$$
\Rot_R(\theta) = e^{2\pi i\theta L_0^R}, \qquad \Rot_{NS}(\theta) = e^{2\pi i\theta L_0^{NS}}.
$$
The generators of these one parameter groups are positive operators, which can be verified by diagonalization. 

\begin{Proposition}\label{propGeneratorRelations} Let $S$ be a finite subset of $\Z$, and suppose that $S = \{n_1, \ldots, n_p, m_1, \ldots, m_q\}$, where
$$
n_1 < \cdots < n_p < 0 \le m_1 < \cdots < m_q.
$$
Then the vectors
$$
\xi_S = \pi_{p}(a(z^{n_p}) \cdots a(z^{n_1}) a(z^{m_1})^* \cdots a(z^{m_q})^*)\Omega_p,
$$
form an orthonormal basis for $\F_{H,p}$ consisting of eigenvectors for $L_0^R$ and $L_0^{NS}$. Their eigenvalues are given by
$$
L_0^R \xi_S = \left(\sum_{i=1}^p -n_i + \sum_{i=1}^q m_i \right)\xi_S
$$
and
$$
L_0^{NS} \xi_S = \left(\sum_{i=1}^p -(n_i+\tfrac12) + \sum_{i=1}^q (m_i+\tfrac12) \right)\xi_S.
$$
In particular, for all $n \in \Z$ we have
\begin{equation}\label{eqnLCommRel}
[L_0^{NS}, \pi_{p}(a(z^n))] = -(n+\tfrac12)\pi_p(a(z^n)),\quad [L_0^R, \pi_{p}(a(z^n))] = -n\pi_{p}(a(z^n))
\end{equation}
and
\begin{equation}\label{eqnLCommRelStar}
[L_0^{NS}, \pi_{p}(a(z^n))^*] = (n+\tfrac12)\pi_p(a(z^n))^*,\quad [L_0^R, \pi_{p}(a(z^n))^*] = n\pi_{p}(a(z^n))^*.
\end{equation}
\end{Proposition}

\subsection{Spin structures}
\subsubsection{Spin structures on Riemann surfaces}
Let $\Sigma$ be a compact Riemann surface with boundary. The complex structure on $\Sigma$ induces an almost complex structure $J$. That is, $J$ is a smooth family of endomorphisms $J_p$ of the tangent spaces $T_p\Sigma$ such that $J_p^2 = -\Id$ for all $p \in \Sigma$. In any local holomorphic coordinate $z = x + iy$, one has 
$$
J \frac{\partial}{\partial x} = \frac{\partial}{\partial y}, \quad J \frac{\partial}{\partial x} = -\frac{\partial}{\partial y}.
$$
Set $T\Sigma_\C = T\Sigma \otimes_\R \C$, and let $T^{(1,0)}\Sigma$ and $T^{(0,1)}\Sigma$ be the bundles of eigenspaces of $J$ for $i$ and $-i$, respectively.

With respect to a local holomorphic coordinate $z:U \to \C$, we have sections
\begin{align*}
 \frac{\partial}{\partial z} := \frac12\left(\frac{\partial}{\partial x} - i \frac{\partial}{\partial y}\right), \quad \mbox{ and } \;\quad
  \frac{\partial}{\partial \overline{z}} := \frac12\left(\frac{\partial}{\partial x} + i \frac{\partial}{\partial y}\right) 
\end{align*}
of $T^{(1,0)}U$ and $T^{(0,1)}U$, respectively.

We give $T^{(1,0)}\Sigma$ and $T^{(0,1)}\Sigma$ the complex structure $J$. 
For $T^{(1,0)}\Sigma$ this coincides with the complex structure inherited from $T\Sigma_\C$, but on $T^{(0,1)}\Sigma$ the complex structure is conjugate to the inherited one. 
The bundles $T^{(1,0)}\Sigma$ and $T^{(0,1)}\Sigma$ are called the {\it holomorphic} and {\it antiholomorphic tangent bundles}, respectively.

Define the {\it holomorphic cotangent bundle} (or {\it canonical bundle}) $K_\Sigma$ by 
$$
K_\Sigma = (T^{(1,0)}\Sigma)^*.
$$
If $(z,U)$ is a local holomorphic coordinate, a trivialization of $K_U$ is given by the section $dz = dx + idy$. 
We also have a trivialization of $(T^{(0,1)}U)^*$ given by $d\overline{z} = dx - idy$. 
If $u \in C^\infty(\Sigma)$, define a section $\partial u$ of $K_\Sigma$ in local holomorphic coordinates by
$$
\partial u = \frac{\partial u}{\partial z} dz.
$$
Similarly define a section $\overline{\partial}u$ of $(T^{(0,1)}\Sigma)^*$ by
$$
\overline{\partial}u = \frac{\partial u}{\partial \overline{z}} d\overline{z}.
$$
The Dolbeault operators $\partial$ and $\overline{\partial}$ are related to the de Rahm differential by $d = \partial + \overline{\partial}$.

\begin{Definition}A {\it spin structure} on $\Sigma$ is a holomorphic line bundle $L$ over $\Sigma$ along with a holomorphic isomorphism $\Phi:L \otimes L \to K_\Sigma$ (that acts identically on the base space).
\end{Definition}

We will refer to a Riemann surface along with a choice of spin structure as a {\it spin Riemann surface}.

\begin{Remark}
This definition of a spin structure is particular to Riemann surfaces. The equivalence of this definition with the standard one is established in \cite[Sec. 3]{At71}.
\end{Remark}

If $L_1$ and $L_2$ are spin structures on $\Sigma_1$ and $\Sigma_2$, then an isomorphism of spin structures $L_1 \to L_2$ is a holomorphic isomorphism of bundles $B:L \to L^\prime$ such that the diagram
$$
\begin{CD}
L_1 \otimes L_1 @>{B \otimes B}>> L_2 \otimes L_2\\
@V\Phi_1 VV @VV\Phi_2 V\\
K_{\Sigma_1} @<{B|_{\Sigma_1}}^*<< K_{\Sigma_2}
\end{CD}
$$
commutes, where $B|_\Sigma^*$ is the pullback.

\begin{Example}\label{exSpinDisk}
Let $\D$ be the closed unit disk in $\C$. Then the (Neveu-Schwarz) spin disk $(\D, NS)$ is given by the following spin structure. We take $L = \D \times \C$.  The spin structure $\Phi:L \otimes L \to K_{\D}$ acts on sections $f \otimes g \in C^\infty(\D) \otimes C^\infty(\D)$ by $\Phi_* (f \otimes g) = fg dz$. Up to isomorphism, this is the only spin structure on $\D$.
\end{Example}

\begin{Example}\label{exSpinAnnuli}
Let $\A_r$ denote the closed annulus 
$$
\A_r = \{z \in \C : r \le \abs{z} \le 1\}.
$$
We define two spin structures on $\A_r$, called the Neveu-Schwarz and Ramond spin structures. Both are given by the trivial bundle $L = \A_r \times \C$. For $\sigma \in \{NS, R\}$ the spin structure $\Phi_{\sigma}$ acts on sections $f \otimes g \in C^\infty(\A_r) \otimes C^\infty(\A_r)$ by
$$
(\Phi_\sigma)_*(f \otimes g) = \left\{ \begin{array}{cl} 
f(z)g(z)\frac{dz}{i}, & \sigma = NS \vspace{.05in}\\ 
f(z)g(z)\frac{dz}{iz}, & \sigma = R
\end{array}\right.
$$
We refer to these spin surfaces as the spin annuli $(\A_r, \sigma)$.
\end{Example}

\begin{Example}\label{exSpinPants}
Let $w \in \D$ and $r_1,r_2 \in (0,1)$, and assume they satisfy $r_1 + r_2 < \abs{w} < 1 - r_1$. Define the pair of pants
$$
\bbP_{w,r_1,r_2} = \D \setminus \left( (r_1 \interior{\D} + w) \cup r_2 \interior{\D}\right),
$$
where $\interior{\D}$ is the open unit disk.
We define a pair of spin surfaces $(\bbP_{w, r_1, r_2}, \sigma)$ for $\sigma \in \{NS, R\}$ as in Example \ref{exSpinAnnuli}. That is, we let $L = \bbP_{w, r_1, r_2} \times \C$ and define spin structures $\Phi_\sigma$ which act on sections by
$$
(\Phi_\sigma)_*(f \otimes g) = \left\{ \begin{array}{cl} 
f(z)g(z)\frac{dz}{i}, & \sigma = NS\\ 
f(z)g(z)\frac{dz}{iz}, & \sigma = R
\end{array}\right.
$$
\end{Example}

\subsubsection{Spin structures on circles}

Let $Y$ be a smooth, closed 1-manifold. 

\begin{Definition}
A {\it spin structure} on $Y$ is a smooth, complex line bundle $L$ and an isomorphism of complex line bundles $\phi:L \otimes L \to T^*Y_\C$, where $T^*Y_\C = T^*Y \otimes_\R \C$.
\end{Definition}

We will refer to the triple $(Y,L,\phi)$ as a {\it (smooth, closed) spin $1$-manifold}.

\begin{Remark}
One could alternatively define a spin structure on $Y$ via real line bundles and an isomorphism to the real cotangent bundle $T^*Y$, and these definitions are equivalent since the real structure on $T^*Y_\C$ induces a real structure on $L$. We have chosen the definition given above because it makes the relationship with spin structures on surfaces more transparent.
\end{Remark}

\begin{Proposition}\label{propRestrictedSpinStructure}
There is a natural identification $K_\Sigma|_\Gamma \cong T^*\Gamma_\C$. Thus if $\Sigma$ is a compact Riemann surface with boundary $\Gamma$ and $(L,\Phi)$ is a spin structure on $\Sigma$, then $(L|_\Gamma, \Phi|_\Gamma)$ naturally becomes a spin structure on $\Gamma$.
\end{Proposition}
\begin{proof}
First, observe that there is a natural $\R$-linear isomorphism $T\Sigma \to T^{(1,0)}\Sigma$. Indeed, $T\Sigma$ sits naturally as a real linear subspace of $T\Sigma_\C$, and since $T\Sigma \cap T^{(0,1)}\Sigma = \{0\}$, the projection of $T\Sigma_\C$ onto $T^{(1,0)}\Sigma$, with respect to the decomposition $T^{(1,0)}\Sigma \oplus T^{(0,1)}\Sigma$, is injective on $T\Sigma$. By dimension counting, this projection induces the desired $\R$-linear isomorphism $T\Sigma \cong T^{(1,0)}\Sigma$.

Now $T\Gamma$ gives a $1$-real-dimensional subbundle of $T\Sigma|_\Gamma$, and transporting along the isomorphism constructed above gives a $1$-real-dimensional subbundle of $T^{(1,0)}\Sigma|_\Gamma$. 

All that remains is to note that if $W$ is a complexification of $V$, then $W^*$ is naturally a complexification of $V^*$, by embedding $V^*$ in $W^*$ as linear functionals taking real values on $V$.
\end{proof}

A morphism of spin structures $(Y_1, L_1) \to (Y_2, L_2)$ is a smooth bundle map $\beta:L_1 \to L_2$ such that
\begin{equation}\label{eqnCircleSpinIso}
\begin{CD}
L_1 \otimes L_1 @>{\beta \otimes \beta}>> L_2 \otimes L_2\\
@V\phi_1 VV @VV\phi_2 V\\
{T^*Y_1}_\C @<{\beta|_{Y_1}}^*<< {T^*Y_2}_\C
\end{CD}
\end{equation}
commutes. Note that ${\beta|_{Y_1}}^*$ is a real linear bundle map $T^*Y_2 \to T^*Y_1$, and thus induces a unique complex linear map bundle map between the complexifications.

%

\begin{Example}
We define a pair of spin structures on $S^1$, called the the {\it Neveu-Schwarz} and {\it Ramond} spin structures. Both are given by the trivial bundle $L = S^1 \times \C$. For $\sigma \in \{NS, R\}$, the spin structure $\phi_\sigma$ is given on sections $f \otimes g \in C^\infty(S^1) \otimes C^\infty(S^1)$ by 
\begin{equation}\label{eqCirclSpinDefinitions}
(\phi_\sigma)_*(f \otimes g) = \left\{ \begin{array}{cl} f(z)g(z)\frac{dz}{i} & \sigma = NS \vspace{0.05in}\\
f(z)g(z) \frac{dz}{iz} & \sigma = R \end{array}\right.
\end{equation}
We denote these spin circles by $(S^1, NS)$ and $(S^1, R)$.
\end{Example}

\begin{Example}
The restriction of the spin disk $(\D, NS)$ to the boundary circle is isomorphic to $(S^1, NS)$. For $\sigma \in \{NS, R\}$, the restriction of the spin annulus $(\A_r, \sigma)$ to either boundary component is isomorphic to $(S^1, \sigma)$. The restriction of $(\bbP_{w,r_1,r_2}, \sigma)$ to the boundary circles $S^1$ and $r_2S^1$ is isomorphic to $(S^1, \sigma)$, but the restriction to $r_1S^1 + w$ is isomorphic to $(S^1, NS)$ in either case.
\end{Example}

For $\sigma \in \{NS,R\}$, let $\Aut(S^1,\sigma)$ denote the group of spin structure automorphisms of the spin circle $(S^1, \sigma)$. Note that these automorphisms are not required to act identically on the base space.
\begin{Proposition}\label{propSpinDiffsAreSpinAuts}
$\Aut(S^1,\sigma)$ is naturally isomorphic to $\Diff^\sigma(S^1)$. Under this isomorphism, diffeomorphisms $(\psi, \gamma) \in \Diff^\sigma(S^1)$ act on sections of the spin bundle via the spin representation $u_\sigma$.
\end{Proposition}
\begin{proof}
Let $L = S^1 \times \C$ and let $\beta:L \to L$ be an automorphism of $\Aut(S^1, \sigma)$. It suffices to show that there exists a $(\psi, \gamma) \in \Diff^\sigma(S^1)$ such that $\beta_*f = u_\sigma(\psi,\gamma) f$ for all sections $f$ of $L$.

Let $\gamma = \beta|_{S^1} \in \Diff(S^1)$ and let $K = (T^*S^1)_\C$. By definition, the diagram
\begin{equation*}
\begin{CD}
L \otimes L @>{\beta \otimes \beta}>> L \otimes L\\
@V\phi_\sigma VV @VV\phi_\sigma V\\
K @<{\gamma}^*<< K
\end{CD}
\end{equation*}
commutes. 

Since $\beta:L \to L$ is a bundle isomorphism, it acts on sections by
\begin{equation*}
\beta_*f = \psi(z)f(\gamma^{-1}(z))
\end{equation*}
for some $\psi \in C^\infty(S^1)^\times$. 

Assume first that $\sigma = NS$, and let $f \otimes g$ be a section of $L \otimes L$. By definition we have
\begin{equation}\label{eqActionOfSpinAuts1}
(\phi_{NS})_*(f \otimes g) = -i f(z)g(z) dz.
\end{equation}
Following the commutative diagram the other way around, we get
\begin{equation}\label{eqActionOfSpinAuts2}
\gamma^*(\phi_{NS})_*(\beta \otimes \beta)_* (f \otimes g) = -if(z)g(z)\psi(\gamma(z))^2\gamma^\prime(z) dz.
\end{equation}

Since the diagram commutes, \eqref{eqActionOfSpinAuts1} and \eqref{eqActionOfSpinAuts2} coincide for all $f$ and $g$, and so we must have $\psi(\gamma(z))^{-2} = \gamma^\prime(z)$ for all $z \in S^1$. That is, $\psi^2 = (\gamma^{-1})^\prime$. We now identify $\beta$ with $(\psi,\gamma) \in \Diff^{NS}(S^1)$, and $\beta$ acts on sections by $u_{NS}(\psi,\gamma)$ as was to be shown.

The case $\sigma = R$ is similar, except in this case the commutativity of the diagram is equivalent to the condition
$$
\psi(\gamma(z))^{-2} = \frac{z\gamma^\prime(z)}{\gamma(z)}.
$$
The right-hand side is equal to $\epsilon(\gamma) \abs{\gamma^\prime(z)}$, where $\epsilon(\gamma) = \pm 1$ if $\gamma \in \Diff_{\pm}(S^1)$. Hence $\psi^2 = \epsilon(\gamma)\abs{(\gamma^{-1})^\prime}$ and we have $(\psi,\gamma) \in \Diff^R(S^1)$. We now identify $\beta$ with $(\psi,\gamma)$, and $\beta$ acts on sections by $u_{R}(\psi,\gamma)$.
\end{proof}

The automorphism corresponding to $(-1,\operatorname{id}) \in \Diff^\sigma(S^1)$ is called the {\it spin involution}.

\begin{Proposition}
The Neveu-Schwarz and Ramond spin structures on $S^1$ are not isomorphic, and every spin structure on $S^1$ is isomorphic to $(S^1, NS)$ or $(S^1, R)$.
\end{Proposition}
\begin{proof}
Let $(L,\phi)$ be a spin structure on $S^1$. For every $\gamma \in \Diff(S^1)$, $(L,\phi)$ has an automorphism that acts on $S^1$ by $\gamma$. Hence it suffices to classify spin structures on $S^1$ up to isomorphisms that act identically on the base space.

Since every complex line bundle on $S^1$ is trivializable, we may assume $L=S^1 \times \C$, in which case $\phi$ is characterized by the non-vanishing section $\omega:=\phi_*(1 \otimes 1)$ of $K_\C|_{S^1}$, where $1$ is the constant function. 
If $\omega_1$ and $\omega_2$ correspond to a pair of spin structures, then base space preserving isomorphisms between these spin structures correspond one-to-one with non-vanishing smooth functions $h \in C^\infty(S^1)^\times$ such that $\omega_1 = h^2 \omega_2$. 
Thus the isomorphism classes of spin structures on $S^1$ are a torsor for $C^\infty(S^1)^\times / (C^\infty(S^1)^\times)^2 \cong \Z/2$. Since $z^{-1}$ is not a square of a smooth function, the spin structures defined by $\omega_1 := \frac{dz}{i}$ and $\omega_2 := \frac{dz}{iz}$ are not isomorphic, and form a complete set of representatives of isomorphism classes.
\end{proof}

\subsubsection{Conjugate spin structures}\label{subsubsecConjugateSpinStructures}
Let $\Sigma$ be a Riemann surface, and let $L$ be a complex line bundle over $\Sigma$. 
We denote by $\overline{\Sigma}$ the Riemann surface obtained by taking the conjugate complex structure on $\Sigma$, and by $\overline{L}$ the line bundle obtained by taking the conjugate complex structure on each fiber of $L$. 
If $L$ has a holomorphic structure, then $\overline{L}$ has a natural holomorphic structure over $\overline{\Sigma}$. 
As real bundles, we have $L_\R = \overline{L}_\R$, and a smooth section of $L_\R$ is a holomorphic section of $L$ if and only if it is a holomorphic section of $\overline{L}$.

Observe that $\overline{T^{(0,1)}\Sigma}=T^{(1,0)}\overline{\Sigma}$. Complex conjugation on $T\Sigma_\C$ exchanges $T^{(1,0)}\Sigma$ and $T^{(0,1)}\Sigma$, and thus induces a holomorphic isomorphism $\overline{T^{(1,0)}\Sigma} \to T^{(1,0)}\overline{\Sigma}$. Dualizing, we get a holomorphic isomorphism $\overline{K_\Sigma} \cong K_{\overline{\Sigma}}$.

Now given a spin structure $\Phi:L \otimes L \to K_{\Sigma}$, there is a natural conjugate spin structure $\overline{\Phi}: \overline{L} \otimes \overline{L} \to K_{\overline{\Sigma}}$ given by
$$
\overline{\Phi}= \overline{L} \otimes \overline{L} \overset{\Phi}{\longrightarrow} \overline{K_\Sigma} \overset{\sim}{\longrightarrow} K_{\overline{\Sigma}}.
$$

Similarly, if $(L, \phi)$ is a spin structure on 1-manifold $Y$, we can define a conjugate spin structure by allowing $\phi$ to act on the conjugate vector spaces. The conjugate spin structure $(\overline{L}, \overline{\phi})$ is given by
$$
\overline{\phi} = \overline{L} \otimes \overline{L} \overset{\phi}{\longrightarrow} \overline{T^*Y_\C} \overset{\sim}{\longrightarrow} T^*Y_\C,
$$	
where the second arrow is complex conjugation.

\begin{Proposition}
Let $(\Sigma,L,\Phi)$ be a spin Riemann surface. Then $\overline{\Phi|_\Gamma} = \overline{\Phi}|_\Gamma$.
\end{Proposition}
\begin{proof}
Recall from Proposition \ref{propRestrictedSpinStructure} that we chose an isomorphism $T^{(1,0)}\Sigma|_\Gamma\to T\Gamma_\C$ so that the diagram
$$
\begin{CD}
\overline{T^{(1,0)}\Sigma|_\Gamma} @>{}>> T^{(1,0)}\overline{\Sigma}|_\Gamma\\
 @VVV @VVV\\
\overline{T\Gamma_\C} @>>> T\Gamma_\C
\end{CD}
$$
commutes, where the top arrow is the isomorphism induced by complex conjugation on $T\Sigma_\C$ and the bottom arrow is complex conjugation on $T\Gamma_\C$.

The above diagram induces a diagram of isomorphisms of dual spaces
$$
\begin{CD}
\overline{L|_\Gamma} \otimes \overline{L|_\Gamma} @>\Phi|_\Gamma>> \overline{K_\Sigma} @>{}>> K_{\overline{\Sigma}}\\
@. @VVV @VVV\\
@. \overline{T^*\Gamma_\C} @>>> T^*\Gamma_\C
\end{CD}
$$
The two paths around this diagram are $\overline{\Phi|_\Gamma}$ and $\overline{\Phi}|_\Gamma$

\end{proof}

\subsubsection{Conformal welding}
One of the fundamental operations in Segal CFT is that of gluing two Riemann surfaces along boundary circles. More generally, we will consider the operation of {\it sewing} a Riemann surface along a pair of boundary circles, which may lie on the same connected component. One wants the (topologically) sewn surface to again be a Riemann surface, and so one must construct a complex structure. It turns out that if the sewing map is a diffeomorphism, then the sewn surface has a natural complex structure.

If $\Sigma$ is a Riemann surface with boundary, a holomorphic function on $\Sigma$ is defined to be a smooth function on $\Sigma$ that is holomorphic in the interior. That is, we require that the function extend continuously to $\partial\Sigma$, and that the restriction to $\partial\Sigma$ be a smooth function.

\begin{Theorem}[Conformal welding]
\label{thmWelding}
Let $\Sigma$ be a Riemann surface, and $C_1$ and $C_2$ be distinct connected components of $\partial\Sigma$, and let $\gamma:C_1 \to C_2$ be an orientation reversing diffeomorphism. Then the topological manifold $\hat \Sigma$ obtained by sewing $C_1$ to $C_2$ along $\gamma$ has a unique complex structure such that the holomorphic functions on $\hat \Sigma$ are naturally in one-to-one correspondence with holomorphic functions $F$ on $\Sigma$ such that $F|_{C_2}\circ \gamma = F|_{C_1}$.
\end{Theorem}
A survery of conformal welding is given in \cite{Sharon2006}.


More generally, we are interested in the conformal welding of spin Riemann surfaces.

\begin{Theorem}
\label{thmWeldingSpinStructures}
Let $(L,\Phi)$ be a spin structure on a Riemann surface $\Sigma$, and let $C_1$ and $C_2$ be distinct boundary components of $\Sigma$.  Suppose that $\beta:L|_{C_1} \to L|_{C_2}$ is an isomorphism of spin structures, and that $\gamma:=\beta|_{C_1}$ is orientation reversing. Then the topological bundle $\hat L$ over $\hat \Sigma$ obtained by sewing along $\beta$ is naturally a spin structure, and the holomorphic sections of $\hat L$ are naturally in one-to-one correspondence with holomorphic sections $F$ of $L$ such that $\beta^* F|_{C_2} = F|_{C_1}$.
\end{Theorem}
\begin{proof}
As remarked in \cite[Sec. 3]{At71}, spin structures on $\Sigma$ are in one-to-one correspondence with topological line bundles $L$ along with {\it continuous} isomorphisms $\Phi: L \otimes L \to K_\Sigma$, as such a $\Phi$ gives $L$ a natural complex structure making $\Phi$ holomorphic. 

Now observe that the projection $\Sigma \to \hat \Sigma$ induces a continuous isomorphism of the topologically sewn bundle $K_\Sigma / \gamma$ with $K_{\hat \Sigma}$. We thus get a continuous isomorphism
$$
\hat \Phi = \hat L \otimes \hat L \longrightarrow K_\Sigma / \gamma \longrightarrow K_{\hat \Sigma}. 
$$
By the above discussion, the complex structure on $K_{\hat \Sigma}$ gives $\hat L$ the structure of a holomorphic bundle, for which $\hat \Phi$ is holomorphic. The holomorphic sections of $\hat L$ are precisely those continuous sections which are holomorphic away from the circle along which $\Sigma$ was sewn.
\end{proof}

One application of Theorem \ref{thmWeldingSpinStructures} is that one can easily embed a compact spin Riemann surface with boundary in an open spin Riemann surface

\begin{Corollary}\label{corEmbeddingInOpenSpinSurface}
Let $(\Sigma, L, \Phi)$ be a compact, connected Riemann surface with non-empty boundary $\Gamma$. Then $(\Sigma, L, \Phi)$ can be embedded in an open spin Riemann surface $(\tilde \Sigma, \tilde L, \tilde \Phi)$.
\end{Corollary}
\begin{proof}
The restriction of $L$ to each connected component of $\Gamma$ is isomorphic to some spin circle $(S^1,\sigma)$ for $\sigma \in \{NS, R\}$. Thus one can embed $\Sigma$ in a new spin Riemann surface $\Sigma^\prime$ by welding a spin annulus $(\A_r,\sigma)$ to each boundary component via Theorem \ref{thmWeldingSpinStructures}. The desired $\tilde \Sigma$ is any sufficiently small neighborhood of $\Sigma$ in $\Sigma^\prime$.
\end{proof}

One of the advantages of embedding a spin Riemann surface with boundary in an open spin Riemann surface is that we may apply the following result on triviality of holomorphic vector bundles.

\begin{Theorem}\label{thmTrivialityOfVectorBundles}
Every holomorphic vector bundle over an open Riemann surface is holomorphically trivializable.
\end{Theorem}

See \cite[\S30]{Fo81} for an extended discussion of Theorem \ref{thmTrivialityOfVectorBundles}.


\section{Spin Riemann surfaces and their Hardy spaces}
\label{secRiggedSpinRiemannSurfaces}

\subsection{Notation, definitions, and examples}

The following notational conventions will be used throughout the remainder of the paper. Let $(\Sigma, L, \Phi)$ be a spin Riemann surface. Let $\Gamma$ be the boundary of $\Sigma$, and let $\pi_0(\Gamma)$ be the set of connected components of $\Gamma$. Let $\beta := (\beta_j)_{j \in \pi_0(\Gamma)}$ be a trivialization of the spin structure $L|_\Gamma$. That is, we have a function $\sigma:\pi_0(\Gamma) \to \{NS,R\}$ and isomorphisms of spin structures
$$
\beta_j:(S^1, \sigma(j)) \to L|_j.
$$
Note that $\sigma$ is uniquely determined by the spin structure on $\Sigma$. 

For $j \in \pi_0(\Gamma)$, let $\gamma_j$ be the isomorphism of 1-manifolds $\beta_j|_{S^1}:S^1 \to j$. 
Riemann surfaces have natural orientations given by the complex structure, and we give $\Gamma$ the orientation induced by restriction.
Now the family $\gamma_j$ induces a partition of the boundary $\Gamma = \Gamma^0 \sqcup \Gamma^1$ into closed connected submanifolds by declaring that $j \subset \Gamma^1$ if and only if $\gamma_j$ is orientation preserving, where $S^1$ is given the standard counterclockwise orientation. 
For a fixed partition $\Gamma = \Gamma^1 \sqcup \Gamma^0$ the collection of compatible boundary trivializations $\beta$ is a torsor for the group $\prod_{j \in \pi_0(\Gamma)} \Diff_+^{\sigma(j)}(S^1)$ by Proposition \ref{propSpinDiffsAreSpinAuts}.

\begin{Definition}
A {\it spin Riemann surface with boundary parametrization} is a quadruple $(\Sigma, L, \Phi, \beta)$ as above. 
That is, $(\Sigma, L, \Phi)$ is a spin Riemann surface, and $\beta_j:(S^1, \sigma(j)) \to L|_j$ is an isomorphism of spin structures.
 We denote by $\cR$ the collection of such $(\Sigma, L, \Phi, \beta)$ with the additional property that $\Sigma$ has no closed components. 
 An isomorphism of spin Riemann surfaces with boundary parametrizations $(\Sigma_1, L_1, \Phi_1, \beta_1) \to (\Sigma_2, L_2, \Phi_2, \beta_2)$ is an isomorphism of spin structures $B:L_1 \to L_2$ such that $B \circ \beta_1 = \beta_2$.
\end{Definition}

\begin{Example}\label{exRiggedSpinDisk}
The spin disk $(\D, NS)$ defined in Example \ref{exSpinDisk} has a boundary trivialization given by the identity map $S^1 \times \C \to S^1 \times \C$.
\end{Example}

\begin{Example}\label{exRiggedSpinAnnuli}
The spin annuli $(\A_r, \sigma)$ defined in Example \ref{exSpinAnnuli} have families of standard boundary trivializations. When $\sigma = R$, this family is parametrized by $q \in rS^1$ and the isomorphisms $\beta_j:S^1 \times \C \to j \times \C$, for $j \in \pi_0(\Gamma)$, are given by
$$
\beta_{j} (z, \alpha) = \left\{\begin{array}{cl}(z, \alpha) & j = S^1\\ (qz, \alpha) & j = rS^1\end{array}\right.
$$
We refer to this spin Riemann surface with boundary parametrization as $(\A_q, R)$. 

When $\sigma = NS$, the standard boundary trivializations depend on $q \in rS^1$ as well as a square root $q^{1/2}$ of $q$. We then define
$$
\beta_{j} (z, \alpha) = \left\{\begin{array}{cl}(z, \alpha) & j = S^1\\ (qz, q^{-1/2}\alpha) & j = rS^1\end{array}\right.
$$
We refer to this spin Riemann surface with boundary parametrization as $(\A_{q,q^{1/2}}, NS)$, or by abuse of notation simply as $(\A_q, NS)$, leaving implicit the choice of $q^{1/2}$.
\end{Example}

\begin{Example}\label{exRiggedSpinPants}
Let $w \in \D$ and $r_1, r_2 \in (0,1)$, and suppose that $r_1 + r_2 < \abs{w} < 1- r_1$, so that we have spin pairs of pants $(\bbP_{w,r_1,r_2}, NS)$ as in Example \ref{exSpinPants}. We define boundary trivializations
$$
\beta_{j} (z, \alpha) = \left\{
\begin{array}{cl}(z, \alpha) & j = S^1\\ 
(q_1z + w, q_1^{-1/2}\alpha) & j = r_1S^1 + w\\
(q_2z, q_2^{-1/2}\alpha) & j = r_2S^1
\end{array}\right.
$$
We refer to this spin Riemann surface with boundary parametrization as $(\bbP_{w,q_1,q_1^{1/2},q_2,q_2^{1/2}}, NS)$, or by abuse of notation as $(\bbP_{w,q_1,q_2}, NS)$, leaving the dependence on the choice of square roots implicit. The moduli space of parametrized standard Neveu-Schwarz spin pairs of pants is
$$
\M_{NS} = \{(w, q_1, q_1^{1/2}, q_2, q_2^{1/2}) \in \C^5 : 0<\abs{q_1} + \abs{q_2} < \abs{w} < 1 - \abs{q_1} \}.
$$
\end{Example}

Let $X=(\Sigma, L, \Phi, \beta) \in \cR$, and let $\Gamma$ be the boundary of $\Sigma$. Define the pre-quantized boundary Hilbert space $H_\Gamma$ by
$$
H_\Gamma = \bigoplus_{j \in \pi_0(\Gamma)} L^2(S^1)
$$
and similarly for $i \in \{0,1\}$ let
$$
H_{\Gamma}^i = \bigoplus_{j \in \pi_0(\Gamma^i)} L^2(S^1).
$$
Note that while $H_\Gamma$ only depends on the manifold $\Gamma$, the partition $\Gamma = \Gamma^0 \sqcup \Gamma^1$, and thus the decomposition $H_\Gamma = H_{\Gamma}^0 \oplus H_{\Gamma}^1$, depend on the spin structure $L$ and the boundary trivialization $\beta$.

Let $X \in \cR$ and denote by $\O(\Sigma;L)$ the collection of sections of $L$ which are holomorphic on the interior of $\Sigma$ and restrict to smooth sections of $L|_\Gamma$.  

\begin{Definition}
The Hardy space $H^2(X) \subset H_\Gamma$ is 
defined by
$$
H^2(X) = \cl \{ \beta^*F|_{\Gamma} : F \in \O(\Sigma;L)\}.
$$
\end{Definition}

\begin{Remark}\label{rmkClassicalHardy}
Elements of the closed subspace $H^2(X)$ have an explicit description in terms of holomorphic sections on the interior of $\Sigma$ with $L^2$ boundary values. The equivalence of these two descriptions is given in the planar case in \cite[\S6]{Bell92}, and the same proof goes through in the case of Riemann surfaces. We will not use this description of the Hardy space.
\end{Remark}

\begin{Proposition}\label{propHardyIsoInvariant}
Let $X_1,X_2 \in \cR$ and suppose that $X_1$ and $X_2$ are isomorphic as spin Riemann surfaces with boundary parametrizations. Then $H^2(X_1) = H^2(X_2)$. 
\end{Proposition}
\begin{proof}
Let $B:X_1 \to X_2$ be an isomorphism. That is, $B$ is an isomorphism of the spin structures of $X_1$ and $X_2$ such that $B \circ \beta_1 = \beta_2$. Then $\O(\Sigma_1; L_1) = B^* \O(\Sigma_2; L_2)$, and thus
\begin{align*}
\{ \beta_1^*F|_{\Gamma_1} : F \in \O(\Sigma_1;L_1)\} &= \{ \beta_1^*B^* F|_{\Gamma_2} : F \in \O(\Sigma_2;L_2)\}\\
&= \{ \beta_2^* F|_{\Gamma_2} : F \in \O(\Sigma_2;L_2)\}.
\end{align*}
\end{proof}

\subsection{Operations on spin Riemann surfaces}

We now introduce several operations on spin Riemann surfaces with boundary parametrizations, starting with the most straightforward, disjoint union.

\subsubsection{Disjoint union}

\begin{Definition}
Given $X = (\Sigma, L, \Phi, \beta) \in \cR$ and $X^\prime = (\Sigma^\prime, L^\prime, \Phi^\prime, \beta^\prime)$, we define the disjoint union 
$$
X \sqcup X^\prime := (\Sigma \sqcup \Sigma^\prime, L \sqcup L^\prime, \Phi \sqcup \Phi^\prime, \beta \sqcup \beta^\prime) \in \cR
$$
in the obvious way.
\end{Definition}

\begin{Proposition}\label{propMonoidalHardy}
Let $X_1,X_2 \in \cR$. Then $H^2(X_1 \sqcup X_2) = H^2(X_1) \oplus H^2(X_2)$.
\end{Proposition}
\begin{proof}
This is immediate from the definitions.
\end{proof}
\subsubsection{Reparametrization}

In Proposition \ref{propSpinDiffsAreSpinAuts}, we identified spin structure automorphisms $\phi:(S^1, \sigma) \to (S^1, \sigma)$ with $(\psi,\gamma) \in \Diff^\sigma(S^1)$ in such a way that 
\begin{equation}\label{eqRelAutsAndDiffs}
\phi_*f = u_\sigma(\psi,\gamma)f, \quad f \in C^\infty(S^1).
\end{equation}
Now given $X=(\Sigma, L, \Phi, \beta) \in \cR$ and 
$$
(\psi,\gamma):= \prod_{j \in \pi_0(\Gamma)} (\psi_j,\gamma_j) \in \prod_{j \in \pi_0(\Gamma)} \Diff_+^{\sigma(j)}(S^1)
$$
we define the action $(\psi,\gamma) \cdot \beta$ by
$$
((\psi,\gamma) \cdot \beta)_j = \beta_j \circ \phi_j^{-1},
$$
where $\phi_j$ is the spin structure automorphism of $(S^1, \sigma(j))$ associated to $(\psi_j,\gamma_j)$ as in \eqref{eqRelAutsAndDiffs}.

\begin{Definition}
Let $X \in \cR$ and let $(\psi,\gamma) \in \prod_{j \in \pi_0(\Gamma)} \Diff_+^{\sigma(j)}(S^1)$. Then the reparametrization of $X$ by $(\psi, \gamma)$ is given by
$$
(\psi, \gamma) \cdot X := (\Sigma, L, \Phi, (\psi, \gamma) \cdot \beta).
$$
\end{Definition}

\begin{Proposition}\label{propReparametrizedHardy}
For $X \in \cR$ and $(\psi,\gamma) \in \prod_{j \in \pi_0(\Gamma)} \Diff_+^{\sigma(j)}(S^1)$, we have
$$
H^2((\psi,\gamma) \cdot X) = \left(\bigoplus_{j \in \pi_0(\Gamma)} u_{\sigma(j)}(\psi_j,\gamma_j) \right) H^2(X).
$$
\end{Proposition}
\begin{proof}
Let $\phi_j$ be the automorphism of $(S^1, \sigma(j))$ corresponding to $(\psi_j,\gamma_j)$ as in \eqref{eqRelAutsAndDiffs} and Proposition \ref{propSpinDiffsAreSpinAuts}, and let $\phi = \prod_{j \in \pi_0(\Gamma)} \phi_j$. From the definition of the Hardy space we have
$$
H^2((\psi,\gamma) \cdot X) = (\phi^{-1})^* H^2(X)
= \phi_* H^2(X).
$$
But $\phi_* H^2(X)$ coincides with the desired expression for $H^2((\psi,\gamma) \cdot X)$ by construction.
\end{proof}


\subsubsection{Conjugation}\label{subsecConjugation}

To formulate the unitarity condition for a Segal CFT, we need a notion of complex conjugation on $\cR$. 
The involution sends a spin Riemann surface $(\Sigma, L, \Phi)$ to the conjugate spin Riemann surface $(\overline{\Sigma}, \overline{L}, \overline{\Phi})$, as defined in Section \ref{subsubsecConjugateSpinStructures}. 
It remains to define the involution $\beta \mapsto \overline{\beta}$ on boundary trivializations $\beta:\prod_{j \in \pi_0(\Gamma)} (S^1, \sigma(j)) \to L|_\Gamma$. 

Let $L = S^1 \times \C$, and for $\sigma \in \{NS, R\}$ let $\rho_\sigma: L \to \overline{L}$ be the bundle isomorphism characterized by
\begin{equation}\label{eqCirclConjugationOnSections}
{\rho_{NS}}_*f(z) = \overline{zf(z)}, \quad {\rho_{R}}_*f(z) = \overline{f(z)}.
\end{equation}

\begin{Caution}\label{ctnConjugates}
Fiberwise, the bundle maps $\rho_\sigma$ give complex linear maps $\C \mapsto \overline{\C}$. 
The reader is cautioned that our notation does not distinguish between elements of $\C$ and $\overline{\C}$ (or, more generally, between elements of $V$ and $\overline{V}$ when $V$ is a complex vector space). 
For example, we write the natural conjugate linear map $V\to \overline{V}$ by $v \mapsto v$.
The notation $\alpha \mapsto \overline{\alpha}$ is used exclusively for complex conjugation, which in the definition of $\rho_\sigma$ we think of as a complex linear map $\C \to \overline{\C}$.

Moreover, whenever we write a map $V \to \overline{V}$, we think of this as being the given map $V \to V$, composed with the (transparant) real isomorphism $V \to \overline{V}$. For example, if we define a map $V \to \overline{V}$ by $v \mapsto iv$, the complex structure is that of $V$, not $\overline{V}$. If $x:V \to W$, we use the same symbol $x$ to refer to the induced map $\overline{V} \to \overline{W}$. Thankfully, once we establish Proposition \ref{propConjugateHardy} we will no longer need these considerations.
\end{Caution}

\begin{Proposition}\label{propRhoIsIso}
$\rho_\sigma:(S^1, \sigma) \to \overline{(S^1, \sigma)}$ is an isomorphism of spin structures.
\end{Proposition}
\begin{proof}
To check that $\rho_\sigma$ is an isomorphism of spin structures, we must verify that the following diagram commutes
$$
\begin{CD}
L \otimes L @>{\rho_\sigma \otimes \rho_\sigma}>> \overline{L} \otimes \overline{L}\\
@V\phi_\sigma VV @VV \phi_\sigma V\\
(T^*S^1)_\C @>{c}>> \overline{(T^*S^1)_\C}
\end{CD}
$$
where the map $c:(T^*S^1)_\C \to \overline{(T^*S^1)_\C}$ is complex conjugation with respect to the real subbundle $T^*S^1$.

Since $d\theta = \frac{dz}{i z}$ is a real section of ${T^*S^1}_\C$, we have
$$
c_*f(z) \frac{dz}{iz} = \overline{f(z)} \frac{dz}{iz}.
$$
Note that as described in Caution \ref{ctnConjugates}, the complex multiplication $\overline{f(z)} \frac{dz}{iz}$ takes place in ${T^*S^1}_\C $ and not $\overline{{T^*S^1}_\C }$.

Recall that if $L = S^1 \times \C$ and $f \otimes g$ is a section of $L \otimes L$, then the action of the spin structures $\phi_{NS}$ and $\phi_{R}$ on $S^1$ are given by
$$
{\phi_{NS}}_*(f \otimes g)(z) = z \,f(z)g(z)\frac{dz}{iz}, \quad {\phi_{R}}_*(f \otimes g)(z) = f(z)g(z)\frac{dz}{iz},
$$
and using the convention of Caution \ref{ctnConjugates} the action of $(\phi_\sigma)_*$ on sections of the conjugate bundle is given by the same formula.

We can check that
\begin{align*}
c_*(\phi_{NS})_*(\rho_{NS}^{\otimes 2})_*(f \otimes g)(z) &=c_* \overline{zf(z)g(z)}\frac{dz}{iz}\\
&= f(z)g(z) \frac{dz}{i}\\
&= {\phi_{NS}}_* (f \otimes g)(z).
\end{align*}
The argument when $\sigma = R$ is similar.
\end{proof}

\begin{Definition}
If $X=(\Sigma, L, \Phi, \beta) \in \cR$, the conjugate $\overline{X}$ is given by $\overline{X} = (\overline{\Sigma}, \overline{L}, \overline{\Phi}, \overline{\beta})$, where 
$$
\overline{\beta}_j = (S^1, \sigma(j)) \overset{\rho_{\sigma(j)}}{\longrightarrow} \overline{(S^1, \sigma(j))} \overset{\beta_j}{\longrightarrow} \overline{L|_j}.
$$
\end{Definition} 
Note that $X \mapsto \overline{X}$ reverses the orientation of $\Sigma$, and so exchanges $\Gamma^0$ and $\Gamma^1$.

The relationship between the Hardy spaces $H^2(X)$ and $H^2(\overline{X})$ is given by the following proposition.

\begin{Proposition}\label{propConjugateHardy}
Let $X = (\Sigma, L, \Phi, \beta) \in \cR$. Then
$$
H^2(\overline{X}) = \left\{ \overline{M_z^{NS} f} : f \in H^2(X) \right\}
$$
where $M_z^{NS} \in \U(H_\Gamma)$ is given by multiplication by the function $z$ on copies of $L^2(S^1)$ indexed by $j \in \pi_0(\Gamma)$ with $\sigma(j) = NS$, and the identity on copies of $L^2(S^1)$ indexed by $j$ with $\sigma(j) = R$.
\end{Proposition}
\begin{proof}
Let $F \in \O(\Sigma;L)$, and note that $F$ is also a holomorphic section of the conjugate bundle $\overline{L}$ over $\overline{\Sigma}$. Then by construction 
$$
\overline{\beta_j}^*F = \rho_{\sigma(j)}^* \beta_j^* F = \overline{M_z^{NS} \beta_j^*F}.
$$
by the defintion of $\rho_{\sigma(j)}$ in \eqref{eqCirclConjugationOnSections}. The desired result now follows from the definition of the Hardy space.
\end{proof}


\subsubsection{Sewing}
\label{subsecSewing}

Let $X = (\Sigma, L, \Phi, \beta) \in \cR$, let $j^0 \in \Gamma^0$ and $j^1 \in \Gamma^1$, and assume that $\sigma(j^0) = \sigma(j^1)$. 
Then 
$$
\beta_{j^1} \circ \beta_{j^0}^{-1} : L|_{j^0} \to L|_{j^1}
$$
is an isomorphism of spin structures that is orientation reversing on the base space. 
Sewing $L|_{j^0}$ and $L_{j^1}$ via this isomorphism yields a spin Riemann surface $(\hat \Sigma, \hat L, \hat \Phi)$ by conformal welding (Theorem \ref{thmWeldingSpinStructures}). We set $\hat \beta_j = \beta_j$ for $j \in \pi_0(\hat \Gamma) \subset \pi_0(\Gamma)$, where $\hat \Gamma$ is the boundary of $\hat \Sigma$.

\begin{Definition}
Let $\cR_*$ be the collection of triples $(X,j^0,j^1)$, where $X \in \cR$ and $j^i \in \pi_0(\Gamma^i)$, such that $\sigma(j^0) = \sigma(j^1)$ and the sewn surface $\hat \Sigma$ has no closed components. We call such a $(X, j^0, j^1)$ a {\it marked spin Riemann surface with boundary parametrization}. 
\end{Definition}

\begin{Definition}
Given $(X, j^0, j^1) \in \cR_*$ we define the sewn spin Riemann surface $\hat X := (\hat \Sigma, \hat L, \hat \Phi, \hat \beta) \in \cR$.
\end{Definition}

We now observe basic properties relating sewing, conjugation and the Hardy space.

\begin{Proposition}\label{propSewnHardy}
Let $(X, j^0, j^1) \in \cR_*$. The subspace of $H^2(\hat X)$ consisting of $(f_j)_{j \in \pi_0(\hat \Gamma)} \in H^2(\hat X)$ which satisfy
\begin{itemize}
\item $f_j \in C^\infty(S^1)$ for all $j \in \pi_0(\hat \Gamma)$,
\item there exists a $f_{j^0} = f_{j^1} \in C^\infty(S^1)$ such that $(f_j)_{j \in \pi_0(\Gamma)} \in H^2(X)$.
\end{itemize}
is dense in $H^2(\hat X)$.
\end{Proposition}
\begin{proof}
This follows immediately from the definition of the Hardy space, and the characterization of sections of $\hat L$ given in Theorem \ref{thmWeldingSpinStructures}.
\end{proof}

\begin{Proposition}\label{propConjugationAndSewingCommute}
Let $(X, j_0, j_1) \in \cR_*$. Then $(\overline{X}, j^1, j^0) \in \cR_*$ and $H^2(\hat {\overline{X}}) = H^2(\overline{\hat X})$.
\end{Proposition}
\begin{proof}
Recall that by definition $\overline{\beta_j} = \beta_j \circ \rho_{\sigma(j)}$ where $\rho_\sigma:(S^1, \sigma) \to \overline{(S^1, \sigma)}$ is a fixed isomorphism of spin circles constructed in Section \ref{subsecConjugation}. Since $\sigma(j_0) = \sigma(j_1)$, we have 
$$
\overline{\beta_{j^1}} \circ \overline{\beta_{j^0}}^{-1} = \beta_{j^1} \circ \rho_{\sigma(j^1)} \circ \rho_{\sigma(j^0)}^{-1} \circ \beta_{j^0}^{-1} = \beta_{j^1} \circ \beta_{j^0}^{-1}. 
$$

Let $\alpha = \beta_{j^1} \circ \beta_{j^0}^{-1}$. Recalling that a section of the holomorphic bundle $L \to \Sigma$ is holomorphic if and only if the corresponding section of $\overline{L} \to \overline{\Sigma}$ is, we see by Theorem \ref{thmWeldingSpinStructures} that holomorphic sections of $\hat{\overline{L}}$ and $\overline{\hat L}$ both correspond to holomorphic sections $F$ of $\overline{L}$ such that $F|_{j^1} \circ \alpha = F|_{j^0}$. The desired result immediately follows.
\end{proof}

Proposition \ref{propSewnHardy} gives the expected relation between $H^2(\hat X)$ and $H^2(X)$,  describing the compatibility of the Hardy space construction with the sewing of spin Riemann surfaces. In Section \ref{secSegalCFT} we will also require the analogous compatibility relation between $H^2(\hat X)^\perp$ and $H^2(X)^\perp$, where the orthogonal complements are taken in $H_{\hat \Gamma}$ and $H_{\Gamma}$, respectively. This precise statement of the compatibility relation is given below as Lemma \ref{lemSewnHardyPerp}.

The compatibliity for orthogonal complements is not a consequence of Proposition \ref{propSewnHardy}. The proof of Lemma \ref{lemSewnHardyPerp} requires the formula for $H^2(\hat X)$ given in Theorem \ref{thmHardyPerp} using the Cauchy transform.

\begin{Lemma}\label{lemSewnHardyPerp}
Let $(X, j^0, j^1) \in \cR_*$. The subspace of $H^2(\hat X)^\perp \subset H_{\hat \Gamma}$ consisting of $(f_j)_{j \in \pi_0(\hat \Gamma)} \in H^2(\hat X)^\perp$ which satisfy
\begin{itemize}
\item $f_j \in C^\infty(S^1)$ for all $j \in \pi_0(\hat \Gamma)$,
\item there exist $f_{j^0} = -f_{j^1} \in C^\infty(S^1)$ such that $(f_j)_{j \in \pi_0(\Gamma)} \in H^2(X)^\perp$.
\end{itemize}
is dense in $H^2(\hat X)$.
\end{Lemma}
\begin{proof}
By Theorem \ref{thmHardyPerp}, we have 
\begin{equation}\label{eqSewnHardyPerp1}
H^2(X)^\perp = M_{\pm} H^2(\overline{X}),
\end{equation}
where $M_{\pm}$ is given by multiplication by $1$ on copies of $L^2(S^1)$ indexed by $j \in \pi_0(\Gamma^1)$ and multiplication by $-1$ on copies of $L^2(S^1)$ indexed by $j \in \pi_0(\Gamma^0)$. Combining this with Proposition \ref{propConjugationAndSewingCommute}, we have 
\begin{equation}\label{eqSewnHardyPerp2}
H^2(\hat X)^\perp = M_{\pm} H^2(\hat{\overline{X}}).
\end{equation}
Applying Proposition \ref{propSewnHardy} to $\overline{X}$ completes the proof.
\end{proof}


\section{The free fermion Segal CFT}
\label{secSegalCFT}

\subsection{Definition of the free fermion Segal CFT}

We continue to use the notation introduced at the beginning of Section \ref{secRiggedSpinRiemannSurfaces}.

The free fermion Segal CFT assigns to the circle a Hilbert space $\F$, and to a spin Riemann surface with boundary parametrization $X = (\Sigma, L, \Phi, \beta) \in \cR$ a one-dimensional space of trace class maps of unordered tensor products
$$
E(X) \subset \B_1\left( \bigotimes_{j \in \pi_0(\Gamma^0)} \F, \bigotimes_{j \in \pi_0(\Gamma^1)} \F\right).
$$
We will characterize the operators $E(X)$ in terms of certain commutation relations derived from the Hardy space $H^2(X)$, which we now describe.

Let $H^0$ and $H^1$ be Hilbert spaces, and let $p_i \in \mathcal{P}(H^i)$. From this data we construct the Fock spaces $\F_{H^i,p_i}$, which are super Hilbert spaces carrying representations $\pi_{p_i}$ of $\CAR(H^i)$, as described in Section \ref{secFABackground}.

\begin{Definition}\label{defCommutationRelations}
Given  a closed subspace $K \subset H^1 \oplus H^0$, we say that a homogeneous bounded operator $T:\F_{H^0,p_0} \to \F_{H^1,p_1}$ satisfies the {\it $K$ commutation relations} if 
\begin{equation}\label{eqImageCommRels}
\pi_{p_1}(a(f^1))T = (-1)^{p(T)} T\pi_{p_0}(a(f^0))
\end{equation}
for all $(f^1, f^0) \in K$, and
\begin{equation}\label{eqPerpCommRels}
\pi_{p_1}(a(g^1))^*T = -(-1)^{p(T)} T\pi_{p_0}(a(g^0))^*
\end{equation}
for all $(g^1, g^0) \in K^\perp$. We have written elements of $H^1 \oplus H^0$ as $(f^1,f^0)$ with respect to the given direct sum decomposition. 
For non-homogeneous operators $T$, we extend the $K$ commutation relations by linearity, so that an operator satisfies the $K$ commutation relations if and only if its even and odd parts do.
\end{Definition}


We now fix notation for the free fermion Segal CFT. 

\begin{Notation}\label{ntnFockStuff}
Let $H=L^2(S^1)$, and let $p \in \P(H)$ be the projection onto the classical Hardy space
$$
pH = \cl \operatorname{span} \{z^n : n \ge 0\}.
$$

Given $X = (\Sigma, L, \Phi, \beta) \in \cR$, we set
$$
H^i_{\Gamma} = \bigoplus_{j \in \pi_0(\Gamma^i)} H,
$$
and $H_\Gamma = H_\Gamma^1 \oplus H_\Gamma^0$.
Define $p_i \in \P(H_{\Gamma}^i)$ by 
\begin{equation}\label{eqnBoundaryProjectionsRecall}
p_i = \bigoplus_{j \in \pi_0(\Gamma^i)} p.
\end{equation}
Let $\F_\Gamma^i = \F_{H^i_{\Gamma},p_i}$. 
\end{Notation}

\begin{Remark}\label{rmkUnorderedTensorProducts}
There is a natural isomorphism between $\F_\Gamma^i$ and the unordered tensor product
$$
\bigotimes_{j \in \pi_0(H_\Gamma^i)} \F_{H,p}
$$
via Proposition \ref{propFockSumToTensor}. In light of this, we identify bounded maps of unordered tensor products
$$
\bigotimes_{j \in \pi_0(H_\Gamma^0)} \F_{H,p} \to \bigotimes_{j \in \pi_0(H_\Gamma^1)} \F_{H,p}
$$
with elements of $\B(\F_\Gamma^0, \F_\Gamma^1)$.
\end{Remark}

\begin{Definition}[The free fermion]
The free fermion Segal CFT assigns to a spin Riemann surface with boundary parametrization $X \in \cR$ the space of all trace class maps $T \in \B_1(\F_\Gamma^0,\F_\Gamma^1)$ which satisfy the $H^2(X)$ commutation relations. We denote this space by $E(X)$.
\end{Definition}

The following theorem, one of the main theorems of the paper, summarizes the most important properties of the free fermion Segal CFT.
\begin{Theorem}\label{thmCFTProperties} Let $X=(\Sigma,L,\Phi,\beta) \in \cR$.
\begin{enumerate}
\item \label{cftPropExistence} (Existence) $E(X)$ is one-dimensional, and its elements are homogeneous and trace class.
\item \label{cftPropNondegeneracy} (Non-degeneracy) If every connected component of $\Sigma$ has an outgoing boundary component, then non-zero elements of $E(X)$ are injective. If every connected component of $\Sigma$ has an incoming boundary component, then non-zero elements of $E(X)$ have dense image.
\item \label{cftPropMonoidal} (Monoidal) If $Y \in \cR$, then $E(X \sqcup Y) = E(X) \grotimes E(Y)$.
\item \label{cftPropSewing} (Sewing) If $(X,j^0,j^1) \in \cR_*$, then the partial supertrace $\tr_{j^0j^1}^s$ induces an isomorphism $E(X) \to E(\hat X)$.
\item \label{cftPropReparametrization} (Reparametrization) If $(\psi_j,\gamma_j) \in \prod_{j \in \pi_0(\Gamma)} \Diff^{\sigma(j)}_+(S^1)$, then
$$
E((\psi_j,\gamma_j) \cdot X) = \left( \bigotimes_{j \in \pi_0(\Gamma^1)} U_{\sigma(j)}(\psi_j,\gamma_j) \right)E(X)\left( \bigotimes_{j \in \pi_0(\Gamma^0)} U_{\sigma(j)}(\psi_j,\gamma_j)^* \right)
$$
where $U_\sigma: \Diff_+^\sigma(S^1) \to \U(\F_{H,p})$ are the spin representations (see Section \ref{subsecDiffReps}).
\item \label{cftPropUnitarity} (Unitarity) $E(\overline{X}) = E(X)^*$, where $E(X)^*$ denotes taking the adjoint elementwise.
\end{enumerate}
\end{Theorem}

As a result of the monoidal and sewing properties, we obtain the usual relationship between gluing of surfaces and composition of operators. 
\begin{Corollary}\label{corOperadicComposition} Let $X,Y \in \cR$, and let $S \subset \pi_0(\Gamma_Y^0)$ and $T \subset \pi_0(\Gamma_X^1)$. Suppose we have a bijection $s: S \to T$ such that $\sigma(s(j))=\sigma(j)$ for all $j \in S$. Let $Z$ be the spin Riemann surface obtained by sewing boundary components of $X$ and $Y$ along $s$, and suppose that $Z$ has no closed components. Then elements of $E(Z)$ are compositions of elements of $E(Y)$ and $E(X)$. More explicitly, we have
$$
E(Z) = \{(y \grotimes \Id_{T^c})(x \grotimes \Id_{S^c}) : x \in E(X), y \in E(Y)\}
$$
where the composition is that of morphisms of unordered tensor products. Here $\Id_{T^c}$ is given by
$$
\Id_{T^c} := \bigotimes_{j \in \pi_0(\Gamma_X^1) \setminus T} \Id_\F
$$
and similarly for $\Id_{S^c}$.
\end{Corollary}
\begin{proof}
By Property \eqref{cftPropMonoidal} of Theorem \ref{thmCFTProperties}, $E(X \sqcup Y) = E(X) \grotimes E(Y)$. Repeatedly applying Property \eqref{cftPropSewing} yields 
$$
E(Z) = \left(\prod_{j \in S} \tr^s_{j,s(j)} \right) E(X) \grotimes E(Y).
$$
By Proposition \ref{propIteratedPartialTrace} the iterated partial supertrace is given by taking the partial supertrace over $\bigotimes_{j \in S} \F_{H,p}$ (identified with the corresponding factors of the codomain via $s$). By Proposition \ref{propSupertraceComposition}, this partial supertrace corresponds to composition of operators, which gives the desired formula for $E(Z)$.
\end{proof}

\subsection{Verification of properties}\label{subsecCFTProof}

In each subsection below, we will establish one of the numbered results from Theorem \ref{thmCFTProperties}.
The technique we will use is to first establish a corresponding property for the Hardy space $H^2(X)$, and show that the property of the CFT is a consequence. 
We continue to use the notation of Notation \ref{ntnFockStuff}.

\subsubsection{Existence/uniqueness}

The main tool for establishing $\dim E(X) = 1$ is the Segal equivalence criterion (Theorem \ref{thmSegalEq}), of which the following is essentially a restatement.

\begin{Lemma}\label{lemExistenceCriterion}
Let $H^0$ and $H^1$ be Hilbert spaces, and let $p_i \in \P(H^i)$. Let $K$ be a closed subspace of $H^1 \oplus H^0$ and let $q_K$ be the corresponding projection. Then the following are equivalent.
\begin{enumerate}
\item \label{itmDiffHS} $(p_1 \oplus (\Id - p_0)) - q_K$ is Hilbert-Schmidt.
\item \label{itmExistsCommRels} There exists a non-zero Hilbert-Schmidt operator $T \in \B_2(\F_{H^0,p_0}, \F_{H^1, p_1})$ which satisfies the $K$ commutation relations (Definition \ref{defCommutationRelations}).
\end{enumerate}
If the above conditions are satisfied, then the operator $T$ is homogeneous and any other Hilbert-Schmidt operator satisfying the $K$ commutation relations is a scalar multiple of $T$. If $(p_1 \oplus (\Id - p_0)) - q_K$ is trace class, then so is $T$.
\end{Lemma}
\begin{proof}
First assume condition \eqref{itmDiffHS} holds. Let $r_0 := (\Id- 2p_0) \in \U(H^0)$ be reflection across $\Id - p_0$, and set $r := \Id \oplus r_0 \in \B(H^1 \oplus H^0)$. Since $[r_0, p_0] = 0$, the modified projection
$
q := r q_K r
$
 also satisfies condition \eqref{itmDiffHS}. Thus by Theorem \ref{thmSegalEq} there exists a non-zero $\tilde \Omega_{q} \in \F_{H \oplus K, p_1 \oplus (\Id-p_0)}$ satisfying the vacuum equations for $q$ (Definition \ref{defVacuumEquations}).
By Proposition \ref{propVacuumHomogeneous}, $\tilde \Omega_q$ is homogeneous. 
Identifying this Fock space with $\F_{H^1,p^1} \otimes \F_{H^0,\Id-p_0}$ as in Proposition \ref{propFockSumToTensor}, these vacuum equations read
$$
(\pi_{p_1}(a(f^1)) \grotimes \Id +\Id \grotimes \pi_{\Id-p_0}(a(f^0))) \tilde \Omega_{q} = 0
$$
for all $(f^1,f^0) \in \operatorname{Im}(q)$ and
$$
(\pi_{p_1}(a(g^1))^* \grotimes \Id + \Id \grotimes\pi_{1-p_0}(a(g^0))^*)\tilde \Omega_{q} = 0
$$
for all $(g^1, g^0) \in \operatorname{Im}(q)^\perp$.

Let $\Phi:\F_{H^0, (\Id - p_0)} \to \F_{H^0,p_0}^*$ be the unitary defined in Section \ref{subsecFockSpace}. By Proposition \ref{propPhiConjugation}, we have
\begin{equation}\label{eqnVacCondOne}
(\pi_{p_1}(a(f^1))\grotimes \Id - \Id \grotimes \overline{\pi_{p_0}(a(r_0f^0))^*}d ) (\Id \otimes \Phi)\tilde\Omega_{q} = 0
\end{equation}
and
\begin{equation}\label{eqnVacCondTwo}
(\pi_{p_1}(a(g^1))^* \grotimes \Id +\Id \grotimes \overline{\pi_{p_0}(a(r_0g^0))}d) (\Id \otimes \Phi)\tilde\Omega_{q} = 0
\end{equation}
where $d$ is the grading involution.

Let $\mu:\F_{H^1,p_1} \otimes \F_{H^0,p_0}^* \to \B_2(\F_{H^0,p_0}, \F_{H^1,p_1})$ be the natural isomorphism, and let $T_q = \mu((\Id \otimes \Phi)\tilde \Omega_q)$.
Since $\tilde \Omega_q$ is homogeneous, so is $T_q$.
Applying Proposition \ref{propMuBimodularity} to Equation \eqref{eqnVacCondOne} gives
\begin{equation}\label{eqnPreCommRel1}
\pi_{p_1}(a(f^1))T_q = d\, T_qd \;  \pi_p(a(r_0 f^0)) = (-1)^{p(T_q)} T_q \pi_p(a(r_0 f^0))
\end{equation}
for all $(f^1, f^0) \in \operatorname{Im}(q)$. By construction, $(f^1, f^0) \in \operatorname{Im}(q)$ if and only if $(f^1, r_0f^0) \in K$, and so $T_q$ satisfies the first half of the $K$ commutation relations, equation \eqref{eqImageCommRels}.

Similarly, if $(g^1, g^0) \in \operatorname{Im}(q)^\perp$, then applying Proposition \ref{propMuBimodularity} to equation \eqref{eqnVacCondTwo} yields
\begin{equation}\label{eqnPreCommRel2}
\pi_{p_1}(a(g^1))^*T_q = -(-1)^{p(T_q)}T_q \pi_{p_0}(a(r_0 g^0))^*
\end{equation}
whenever $(g^1, g^0) \in \operatorname{Im}(q)^\perp$.  Hence $T_q$ satisfies the second half of the $K$ commutation relations, equation \eqref{eqPerpCommRels}.
This completes the proof that \eqref{itmDiffHS} implies \eqref{itmExistsCommRels}.

In fact, the proof shows that the grading preserving map $\F_{H^1, p_1} \otimes \F_{H^0,\Id-p_0} \to \B_2(H_0,H_1)$ given by $\xi \mapsto \mu((\Id \grotimes \Phi)\xi)$ induces an isomorphism between the space of vectors satisfying the $q$ commutation and the space of Hilbert-Schmidt maps satisfying the $K$ commutation relations. 
By Theorem \ref{thmSegalEq}, the space of vectors satisfying the $q$ commutation relations has dimension zero or one, with dimension one exactly when \eqref{itmDiffHS} is satisfied. Thus \eqref{itmDiffHS} holds if and only if \eqref{itmExistsCommRels} holds.

It remains to show that if $(p_1 \oplus (\Id - p_0) )- q$ is trace class, then $T_q = \mu((\Id \otimes \Phi)\tilde \Omega_q)$ is trace class. 
From the explicit formula for $\tilde \Omega_q$ in, e.g., \cite[Thm. 10.6]{Th92} or \cite[\S3]{Wa98}, there exist unit vectors $f_k, g_k, h_j \in H^1 \oplus H^0$ such that
\begin{equation}\label{eqnProductFormulaForVacuum}
\tilde \Omega_q = y\prod_{k=1}^\infty (\Id + \lambda_k x_k) (\Omega \otimes \Omega),
\end{equation}
where 
$$
x_k = \pi_{p_1 \oplus (\Id-p_0)}(a(f_k)a(g_k)^*), \qquad
y = \pi_{p_1 \oplus (\Id-p_0)}(a(h_1) \cdots a(h_n) a(h_{n+1})^* \cdots a(h_m)^*)
$$
and $\lambda_k \in \R_{\ge 0}$ are distinct eigenvalues of $\abs{(p_1 \oplus (\Id - p_0)) - q}$.

If $f=(f^1,f^0) \in H^1 \oplus H^0$, then $\pi_{p_1 \oplus (\Id - p_0)}(a(f)) = \pi_{p_1}(a(f^1)) \grotimes \Id + \Id \grotimes \pi_{\Id-p_0}(a(f^0))$.
Thus if $\norm{f} \le 1$ and $\xi \in \F_{H^1, p_1} \otimes \F_{H^0, \Id - p_0}$ is a linear combination of at most $C$ simple tensors, each with norm at most $\alpha$, then $\pi_{p_1 \oplus (\Id - p_0)}(a(f))\xi$ is a linear combination of at most $2C$ simple tensors, each with norm at most $\alpha$.

Hence, expanding the product \eqref{eqnProductFormulaForVacuum} for $\tilde \Omega_q$, we can write $\tilde \Omega_q \in \F_{H^1, p_1} \otimes \F_{H^0,\Id-p_0}$ as a sum of vectors $\xi_S$ indexed by finite subsets $S \subset \Z_{\ge 1}$, such that $\xi_S$ is a sum of at most $2^{2\abs{S}+m}$ simple tensors, each with norm at most $\sum_{k \in S} \lambda_k$.

If $\xi \in \F_{H^1, p_1} \otimes \F_{H^0,\Id-p_0}$ is a simple tensor, then so is $(\Id \grotimes \Phi)\xi$, and $\norm{\mu(\Id \grotimes \Phi)\xi}_1 = \norm{\xi}$. 
Hence
$$
\norm{T_q}_1 = \norm{\mu(\Id \grotimes \Phi)\tilde \Omega_q}_1 \le \sum_{S} \norm{\xi_S} \le 2^m \sum_{S} 4^{\abs{S}} \sum_{k \in S} \lambda_k = 2^m \prod_{k=1}^\infty (1+4\lambda_k).
$$
The last term is finite because $\sum \lambda_k \le \norm{p_1 \oplus (\Id - p_0) - q}_1$, and so $T_q$ is trace class.
\end{proof}

Establishing that condition \eqref{itmDiffHS} of Lemma \ref{lemExistenceCriterion} holds for the Hardy spaces $H^2(X) \subset H_\Gamma$ is one of the main results of Section \ref{secCauchyTransform}, which allows us to establish the existence property for $E(X)$.

\begin{Theorem}\label{thmExistence}
If $X \in \cR$, then $\dim E(X) = 1$ and the elements of $E(X)$ are homogeneous and trace class.
\end{Theorem}
\begin{proof}
By Theorem \ref{thmProjectionDifferenceTraceClass}, condition \eqref{itmDiffHS} of Lemma \ref{lemExistenceCriterion} holds for $H^i = H_\Gamma^i$, with $p_i$ as in \eqref{eqnBoundaryProjectionsRecall}, and $K = H^2(X)$. Moreover, from the same theorem, $(p_1 \oplus \Id - p_0) - q_K$ is trace class. Thus the conclusion follows immediately from Lemma \ref{lemExistenceCriterion}.
\end{proof}

\subsubsection{Non-degeneracy}

Before establishing the non-degeneracy property of the CFT (Theorem \ref{thmCFTProperties} \eqref{cftPropNondegeneracy}), we need the corresponding property of the Hardy space.

\begin{Proposition}\label{propHardyDenseComponents}
Let $X \in \cR$, and let $S \subset \pi_0(\Gamma)$. Let $H_\Gamma = \bigoplus_{j \in \pi_0(\Gamma)} L^2(S^1)$, and let $p_S$ be the projection of $H_\Gamma$ onto the copies of $L^2(S^1)$ indexed by $S$. If each connected component of $\Sigma$ has a boundary component not contained in $S$, then $p_S H^2(X)$ and $p_S H^2(X)^\perp$ are dense in $\bigoplus_S L^2(S^1)$.
\end{Proposition}
\begin{proof}
In light of Proposition \ref{propMonoidalHardy}, we may assume without loss of generality that $\Sigma$ is connected.
By Corollary \ref{corEmbeddingInOpenSpinSurface}, we may assume that $(\Sigma, L, \Phi)$ is embedded in an open spin Riemann surface $(\tilde \Sigma, \tilde L, \tilde \Phi)$.
By Theorem \ref{thmTrivialityOfVectorBundles}, we may assume that $L$ is the trivial $\C$-bundle. 
By Bishop's approximation theorem \cite[Cor. 2]{Bishop58}, every continuous function on $\bigsqcup_{j \in S} j$ can be uniformly approximated by holomorphic functions on $\tilde \Sigma$, and thus $p_S H^2(X)$ is dense in $\bigoplus_S L^2(S^1)$.

By Theorem \ref{thmHardyPerp}, $H^2(X)^\perp = M_{\pm} H^2(\overline{X})$, where $M_{\pm}$ is multiplication by $1$ and $-1$ on copies of $L^2(S^1)$ indexed by outgoing and incoming boundary componenents, respectively. Thus the density of $p_S H^2(X)^\perp$ follows from that of $p_S H^2(\overline{X})$.
\end{proof}

And now non-degeneracy of the CFT follows from Proposition \ref{propHardyDenseComponents}.

\begin{Proposition}\label{propCFTInjectiveDenseImage}
Let $X=(\Sigma,L,\Phi,\beta) \in \cR$.
\begin{enumerate}
\item \label{itmCFTInjective} If every connected component of $\Sigma$ has an outgoing boundary component, then non-zero elements of $E(X)$ are injective. 
\item \label{itmCFTDenseImage} If every connected component of $\Sigma$ has an incoming boundary component, then non-zero elements of $E(X)$ have dense image.
\end{enumerate}
\end{Proposition}
\begin{proof}
Assume first that every connected componenet of $\Sigma$ has an outgoing boundary component. 
Let $T \in E(X)$.
That is, $T \in \B_1(\F_{H_{\Gamma}^0,p_0}, \F_{H_{\Gamma}^1,p_1})$ and satisfies the commutation relations for $H^2(X) \subset H_{\Gamma}^1 \oplus H_\Gamma^0$.
We will show that $\ker T$ is invariant under $\CAR(H_\Gamma^0)$, and since $\CAR(H_\Gamma^0)$ acts irreducibly on $\F_{H_{\Gamma}^0,p_0}$ this will imply the desired result.

Applying Proposition \ref{propHardyDenseComponents} with $S = \pi_0(\Gamma^0)$, we get that the projection of $H^2(X)$ onto $H_\Gamma^0$ has dense image. 
Call this subspace $K$. 
By definition, for every $f^0 \in K$, there exists an $f^1 \in H_\Gamma^1$ such that $(f^1, f^0) \in H^2(X)$. 

Now let $\xi \in \ker T$. 
Since $T$ satisfies the $H^2(X)$ commutation relations, we have 
$$
Ta(f^0)\xi = (-1)^{p(T)}a(f^1)T\xi = 0
$$
for every $f^0 \in K$. 
Since $K$ is dense in $H_{\Gamma}^0$, $\ker T$ is invariant under $a(f)$ for all $f \in H_\Gamma^0$.
A similar argument, using the projection of $H^2(X)^\perp$ onto incoming boundary componenets, shows that $\ker T$ is invariant under $a(f)^*$ for all $f \in H_\Gamma^0$, which completes the proof of item \eqref{itmCFTInjective}.

The proof of item \eqref{itmCFTDenseImage} is similar, or alternatively \eqref{itmCFTDenseImage} follows from \eqref{itmCFTInjective} and the unitarity property Proposition \ref{propCFTUnitarity}.
\end{proof}

\subsubsection{Monoidal property}

\begin{Proposition}\label{propCFTMonoidal}
If $X,Y \in \cR$, then $E(X \sqcup Y) = E(X) \grotimes E(Y)$.
\end{Proposition}
\begin{proof}
By Proposition \ref{propMonoidalHardy}, we have $H^2(X \sqcup Y) = H^2(X) \oplus H^2(Y)$. It is now a simple exercise to check that if $T_1$ satisfies the $L_1$ commutation relations and $T_2$ satisfies the $L_2$ commutation relations, then $T_1 \grotimes T_2$ satisfies the $L_1 \oplus L_2$ commutation relations. This gives us an inclusion $E(X) \grotimes E(Y) \subseteq E(X \sqcup Y)$, but since both spaces are 1-dimensional by Theorem \ref{thmExistence}, this is an equality.
\end{proof}

\subsubsection{Reparametrization}

We saw in Proposition \ref{propReparametrizedHardy} that reparametrizing the boundary components of $X \in \cR$ acted on $H^2(X)$ by unitary operators coming from the spin representations $u_\sigma$ of $\Diff_+^\sigma(S^1)$ (see Section \ref{subsecDiffReps}). The following proposition describes the corresponding action on maps satisfying the $H^2(X)$ commutation relations.

\begin{Proposition}\label{propUnitaryChangeCommRels}
Let $H^0$ and $H^1$ be Hilbert spaces, and let $K \subset H^1 \oplus H^0$ be a closed subspace. Let $p_i \in \P(H_i)$ and let $u_i \in \U_{res}(H^i,p_i)$. Let $u_i \mapsto U_i$ denote the basic representation (see Section \ref{subsecFockSpace}). If $T \in \B(\F_{H^0,p_0}, \F_{H^1,p_1})$ satisfies the $K$ commutation relations, then $d^{p(U_1) + p(U_2)} U_1TU_0^*$ satisfies the $(u_1 \oplus u_0)K$ commutation relations. 
\end{Proposition}
\begin{proof}
Let $(u_1 f^1,u_0f^0) \in (u_1 \oplus u_0)K$. Then we have
\begin{align*}
\pi_{p_1}(a(u_1f_1))d^{p(U_1)+p(U_2)} U_1T U_0^* &=(-1)^{p(U_1)+p(U_2)}d^{p(U_1)+p(U_2)}U_1\pi_{p_1}(a(f_1)) T U_0^*\\
&= (-1)^{p(U_1)+p(U_2)+p(T)}d^{p(U_1)+p(U_2)}U_1T U_0^*\pi_{p_0}(a(u_0f_0)).
\end{align*}
Thus  $d^{p(U_1)+p(U_2)}U_1T U_0^*$ satisfies the first half of the $(u_1 \oplus u_0)K$ commutation relations. The relations for $(u_1g^1, u_0g^0) \in (u_1 \oplus u_0)K^\perp$ are similar.
\end{proof}

In our case, the spaces $H^i$ will be given as a direct sum 
$$
H^i = \bigoplus_{j \in \pi_0(\Gamma^i)} H.
$$
Thus we also need to know how the basic representation on $\F_{H^1}$ relates to the basic representation on $\bigotimes_j\F_{H}$ under the isomorphism of Proposition \ref{propFockSumToTensor}.

\begin{Proposition}\label{propFockSumBasicRep}
Let $H_1$ and $H_2$ be Hilbert spaces, with $p_i \in \P(H_i)$. Suppose $u_i \in \U_{res}(H_i,p_i)$, and $U_i \in \U(\F_{H_i,p_i})$ is the image of $u_i$ under the basic representation. 
Let $U \in \U_{res}(H_1 \oplus H_2, p_1 \oplus p_2)$ be the image of $u_1 \oplus u_2$ under the basic representation. 
Then, up to a scalar multiple, the isomorphism $\F_{H_1 \oplus H_2,p_1 \oplus p_2} \cong \F_{H_1, p_1} \otimes \F_{H_2, p_2}$ identifies $U$ with $U_1 d^{p(U_2)} \grotimes U_2 d^{p(U_1)}$.
\end{Proposition}
\begin{proof}
It suffices to check that $U_1 d^{p(U_2)} \grotimes U_2 d^{p(U_1)}$ implements the Bogoliubov automorphism corresponding to $u_1 \oplus u_2$ in the representation of $\CAR(H_1 \oplus H_2)$ on $\F_{H_1} \otimes \F_{H_2}$ (given by Equation \eqref{eqTensorCARReps}). This computation is straightforward.
\end{proof}

We can now prove the reparametrization property for the CFT.

\begin{Proposition}\label{propCFTReparametrization}
If $(\psi_j,\gamma_j) \in \prod_{j \in \pi_0(\Gamma)} \Diff^{\sigma(j)}_+(S^1)$, then
$$
E((\psi_j,\gamma_j) \cdot X) = \left( \bigotimes_{j \in \pi_0(\Gamma^1)} U_{\sigma(j)}(\psi_j,\gamma_j) \right)E(X)\left( \bigotimes_{j \in \pi_0(\Gamma^0)} U_{\sigma(j)}(\psi_j,\gamma_j)^* \right)
$$
\end{Proposition}
\begin{proof}
By Proposition \ref{propReparametrizedHardy}, 
$$
H^2((\psi_j,\gamma_j) \cdot X) = \left(\bigoplus_{j \in \pi_0(\Gamma)} u_{\sigma(j)}(\psi_j,\gamma_j) \right) H^2(X).
$$
Let $U_i$ be the image of $\bigoplus_{j \in \pi_0(\Gamma^i)} u_{\sigma(j)}(\psi_j,\gamma_j)$ under the basic representation on $\F_{\Gamma}^i$. 
By Proposition \ref{propUnitaryChangeCommRels} and the fact that the $U_i$ are even (Corollary \ref{corDiffsEven}), we have $E((\psi_j, \gamma_j) \cdot X) = U_1 E(X) U_0^*$. The desired result now follows from Proposition \ref{propFockSumBasicRep}.
\end{proof}

\subsubsection{Unitarity}

As with the other properties of the CFT, to establish unitarity we first need to understand what happens at the level of Hardy spaces.

\begin{Proposition}\label{propAdjointCommRels}
Let $K \subset H^1 \oplus H^0$ be a closed subspace, and let $p_i \in \P(H^i)$. Then $T:\F_{H^0,p_0} \to\F_{H^1,p_1}$ satisfies the $K$ commutation relations if and only if $T^*$ satisfies the commutation relations for $M_{\pm}K^\perp \subset H^0 \oplus H^1$, where $M_{\pm} = \Id_{H^0} \oplus -\Id_{H^1}$.
\end{Proposition}
\begin{proof}
It suffices to prove that $T^*$ satisfies the $M_{\pm} K^\perp$ commutation relations, since the converse is equivalent. The statement for $T^*$ follows immediately from taking adjoints in the definition of the $K$ commutation relations (Definition \ref{defCommutationRelations}).
\end{proof}

Unitarity now follows as an easy consequence of the formula for $H^2(X)^\perp$ calculated in Section \ref{secCauchyTransform}.

\begin{Proposition}\label{propCFTUnitarity}
$E(\overline{X}) = E(X)^*$
\end{Proposition}
\begin{proof}
By Theorem \ref{thmHardyPerp}, we have $H^2(X)^\perp = M_{\pm} H^2(\overline{X})$. Thus by Proposition \ref{propAdjointCommRels}, adjoints of elements of $E(X)$ lie in $E(\overline{X})$, and vice versa.
\end{proof}

\subsubsection{Sewing}\label{subsubsecSewing}
Suppose $(X, j^0, j^1)  \in \cR_*$, and let $\hat X$ be the result of sewing $X$ along $j^0$ and $j^1$ (see Section \ref{subsecSewing}). 
Recall that by the definition of $\cR_*$, $\hat X$ has no closed components.

The partial supertrace $\tr_{j^0j^1}^s$ gives a map 
$$
\tr_{j^0j^1}^s:\B_1(\F_\Gamma^0, \F_\Gamma^1) \to \B_1(\F_{\hat \Gamma}^0, \F_{\hat \Gamma}^1),
$$
where $\hat \Gamma = \partial \hat \Sigma$.
\begin{Theorem}\label{thmCFTSewing}
Let $(X, j^0, j^1) \in \cR_*$ and let $\hat X \in \cR$ be the result of sewing $j^0$ to $j^1$. Then $\tr^s_{j^0j^1}$ induces an isomorphism $E(X) \to E(\hat X)$.
\end{Theorem}
\begin{proof}
We first show that $\tr^s_{j^0j^1}(E(X)) \subset E(\hat X)$. That is, for $T \in E(X)$ we show that $\tr^s_{j^0j^1}(T)$ satisfies the $H^2(\hat X)$ commutation relations. 

Fix $f = (f^1 , f^0) \in H^2(\hat X) \subseteq H_{\hat \Gamma}^1 \oplus H_{\hat \Gamma}^0$ and $g = (g^1, g^0) \in H^2(\hat X)^\perp$.
We write
$$
(f^1,f^0) = (f_j) \in \bigoplus_{j \in \pi_0(\hat \Gamma)} L^2(S^1),
$$
and similarly for $(g^1,g^0)$.

We must show that 
$$
a(f^1)\tr^s_{j^0j^1}(T) = (-1)^{p(\tr^s_{j^0j^1}(T))}\tr^s_{j^0j^1}(T)a(f^0)
$$
and that 
$$
a(g^1)^*\tr^s_{j^0j^1}(T) = -(-1)^{p(\tr^s_{j^0j^1}(T))}\tr^s_{j^0j^1}(T)a(g^0)^*.
$$
It suffices to verify these identities for $(f^1, f^0)$ lying in a dense subspace of $H^2(\hat X)$, and for $(g^1,g^0)$ lying in a dense subspace of $H^2(\hat X)^\perp$. Hence by Proposition \ref{propSewnHardy} we may assume without loss of generality that there exists a $h = (h^0, h^1) \in H^2(X)$ such that $h_j = f_j$ for $j \ne j^i$, and $h_{j^1} = h_{j^0}$.

To reduce notational complexity, we will simply write $a(f)$ instead of $\pi_{p_i}(a(f))$ for the action of $\CAR(H_\Gamma^i)$ on $\F_{\Gamma}^i$.

We embed $H^i_{\hat \Gamma}$ as a subspace of $H^i_{\Gamma}$ by the natural inclusion coming from $\pi_0(\hat \Gamma) \subset \pi_0(\Gamma)$. 
We then have $h^i = f^i + h_{j^i}$, with respect to the decomposition $H_\Gamma^i = H_{\hat \Gamma}^i \oplus L^2(S^1)$. 
By Proposition \ref{propFockSumToTensor}, this implies that 
$$
a(h^i) = a(f^i) \grotimes \Id_{j^i} + \Id_{\pi_0(\Gamma^i) \setminus j^i} \grotimes a(h_{j^i}).
$$
Using the partial supertrace properties from Proposition \ref{propSupertraceProperties}, we now have have
\begin{align*}
a(f^1) \tr^s_{j^0j^1}(T) &=\tr^s_{j^0j^1}\left((a(f^1) \grotimes \Id_{j^1} ) T\right)\\
&=\tr^s_{j^0j^1}\left(\left(a(h^1)  - \Id_{\pi_0(\Gamma^1) \setminus j^1} \grotimes a(h_{j^1})\right) T\right)\\
&= (-1)^{p(T)}\tr^s_{j^0j^1}\left( T \left(a(h^0) - \Id_{\pi_0(\Gamma^0)\setminus j^0} \grotimes a(h_{j^1})\right)\right)\\
&= (-1)^{p(T)}\tr^s_{j^0j^1}\left(T\left(a(f^0) \grotimes \Id_{j^0}\right)\right)\\
&= (-1)^{p(T)}\tr^s_{j^0j^1}(T) a(f^0).
\end{align*}
Hence $\tr^s_{j^0j^1}(T)$ satisfies the first $H^2(\hat X)$ commutation relations  \eqref{eqImageCommRels}.

The same proof establishes the corresponding relations for $(g^1,g^0) \in H^2(\hat X)^\perp$. By Lemma \ref{lemSewnHardyPerp}, we may assume without loss of generality that there exists $(k^1,k^0) \in H^2(X)^\perp$ such that $k_j = g_j$ for $j \ne j^i$, and $k_{j^1} = -k_{j^0}$. The same computation as above now yields 
$$
a(g^1)^*\tr^s_{j^0j^1}(T) = -(-1)^{p(T)} \tr^s_{j^0j^1}(T)a(g^0)^*.
$$
We conclude that $\tr^s_{j^0j^1}(T) \in E( \hat X)$.

To complete the proof, we must show that $\tr_{j^0j^1}^s: E(X) \to E(\hat X)$ is an isomorphism. Since both spaces are one-dimensional, it suffices to prove that if $\tr^s_{j^0j^1}(T) = 0$ then $T = 0$. 

Assume first that $j^0$ and $j^1$ lie on the same connected component of $\Sigma$, and suppose that $\tr^s_{j^0j^1}(T) = 0$.

By the monoidal property, we may assume without loss of generality that $\Sigma$ is connected.

Suppose that $(h^1, h^0) \in H^2(X)$, write $h^i = f^i + h^i_{j^i}$ with respect to the decomposition $H^i_\Gamma = H_{\hat \Gamma}^i \oplus L^2(S^1)$. 
Calculating as above, we have
\begin{align*}
0 =& a(f^1) \tr^s_{j^0j^1}(T) \\
=& \tr^s_{j^0j^1}\left((a(f^1) \grotimes \Id_{j^1})T\right) \\
=& \tr^s_{j^0j^1}\left((a(h^1) - (\Id_{\pi_0(\hat \Gamma^1)} \grotimes a(h^1_{j^1}))T\right) \\
=& (-1)^{p(T)} \tr^s_{j^0j^1}\left(Ta(h^0)\right) - \tr^s_{j^0j^1}\left((\Id_{\pi_0(\hat \Gamma^1)} \grotimes a(h^1_{j^1}))T\right) \\
=& (-1)^{p(T)} \tr^s_{j^0j^1}\left(T(\Id \grotimes a(h_{j^0}^0))\right) + (-1)^{p(T)} \tr^s_{j^0j^1}(T)a(f^0)- \\
&-\tr^s_{j^0j^1}\left((\Id_{\pi_0(\hat \Gamma^1)} \grotimes a(h^1_{j^1}))T\right) \\
=&(-1)^{p(T)} \tr^s_{j^0j^1}\left(T(\Id_{\pi_0(\hat \Gamma^0)} \grotimes a(h_{j^0}^0))\right) - \tr^s_{j^0j^1}\left((\Id_{\pi_0(\hat \Gamma^1)} \grotimes a(h^1_{j^1}))T\right)  \numberthis \label{eqnProvingSewingIso}
\end{align*}

Since $(X, j^0, j^1) \in \cR_*$, the connected component of $\Sigma$ containing $j^0$ and $j^1$ has at least one more boundary component, and so the projection of $H^2(X)$ onto $\bigoplus_{j \in \{j^0, j^1\}} L^2(S^1)$ has dense image by Proposition \ref{propHardyDenseComponents}. Thus given any $f \in L^2(S^1)$ we may take a sequence $(h^{1,n}, h^{0,n}) \in H^2(X)$ with $h^{0,n}_{j^0} \to 0$ and $h^{1,n}_{j^1} \to f$. 
Hence 
\begin{equation}\label{eqnProvingSewingIsoLimits}
T(\Id_{\pi_0(\hat \Gamma^0)} \grotimes a(h^{0,n}_{j^0})) \to 0, 
\quad \mbox{ and } \quad
(\Id_{\pi_0(\hat \Gamma^1)} \grotimes a(h^1_{j^1}))T \to (\Id_{\pi_0(\hat \Gamma^1)} \grotimes  a(f))T
\end{equation}
in the trace norm.
We can apply the result of the calculation \eqref{eqnProvingSewingIso} to $(h^{1,n},h^{0,n})$, and by \eqref{eqnProvingSewingIsoLimits} and the continuity of the partial trace, we have
$$
\tr^s_{j^0j^1}\left((\Id_{\pi_0(\hat \Gamma^1)} \grotimes a(f))T\right) = 0.
$$

Applying this argument repeatedly using elements of $H^2(X)$ and $H^2(X)^\perp$ yields
\begin{equation}\label{eqnStillProvingSewingIso}
\tr^s_{j^0j^1}\left((\Id_{\pi_0(\hat \Gamma^1)} \grotimes x)T\right) = 0
\end{equation}
when $x$ is an arbitrary word in $a(f)$'s and $a(g)^*$'s. 

Now for arbitrary $y \in \B(\F_{\hat \Gamma}^1, \F_{\hat \Gamma}^0)$, by Proposition \ref{propSupertraceProperties} we have
\begin{equation}\label{eqnStillStillProvingSewingIso}
0 = y \tr^s_{j^0j^1}\left((\Id_{\pi_0(\hat \Gamma^1)} \grotimes x)T\right) = \tr^s_{j^0j^1}\left((y \grotimes x)T\right).
\end{equation}

Let $\mathcal{A}$ be the the linear span of operators $y \grotimes x$ with $x$ and $y$ as above.  
Since $\CAR(L^2(S^1))$ acts irreducibly on $\F_{L^2(S^1),p}$, $\mathcal{A}$ is dense in $\B(\F_{\Gamma}^1, \F_{\Gamma}^0)$ in the strong operator topology.
A standard argument using the Kaplansky density theorem shows that every element of $\B(\F_{\Gamma}^1, \F_{\Gamma}^0)$ is a limit of a sequence in $\mathcal{A}$.

If $S_n$ is a sequence of operators on a Hilbert space converging strongly, and $T$ is trace class, then $S_nT \to ST$ in the trace norm.
Hence by the continuity of the partial supertrace, we have $\tr^s_{j^0j^1}(ST) = 0$ for all $S \in \B(\F_{\Gamma}^1, \F_{\Gamma}^0)$. 
In particular, 
$$
\tr(T^*T) = \tr^s_{\F_{\hat \Gamma}^0}(\tr^s_{j^0j^1}(d_{\F_\Gamma^0} T^*T)) = 0.
$$
It follows that $T = 0$, which completes the proof of injectivity in the case where $j^1$ and $j^0$ lie on the same connected component of $\Sigma$.

Now consider when $X = X_0 \sqcup X_1$, with $j^i$ a boundary component of the surface underlying $X_i$. Since $(X, j^1, j^0) \in \cR_*$, either $X_0$ or $X_1$ has a boundary component which is neither $j^1$ nor $j^0$. If it is $X_1$ that has the additional boundary component, then we may use the same argument as above, and may even take $h_{j^0}^{0,n} = 0$ for all $n$. On the other hand, if $X_0$ has the additional boundary component, then we must take $h^{1,n}_{j^1} = 0$, and choose $h_{j^1}^{1,n} \to f$. 
The rest of the argument is the same.
\end{proof}

\section{From Segal CFT to vertex operators}
\label{secVertexOperators}

The main result of this section is Theorem \ref{thmPantsAreVertexOperators}, in which we identify the value of the CFT on standard pairs of pants $(\bbP_{w,q_1,q_2},NS)$ with fields from the free fermion vertex operator algebra. We fix the notation $H = L^2(S^1)$, $p \in \P(H)$ is the projection onto the classical Hardy space $H^2(\D)$, and $\F = \F_{H,p}$. We will drop the notation $\pi_p$ for the representation of $\CAR(H)$ on $\F$, and simply write $a(f)$.

\subsection{Warmup: Disks and annuli}

Let $(\D, NS)$ be the standard spin disk with its standard parametrization, descibed in Example \ref{exRiggedSpinDisk}, and for $q \in \D$ and $\sigma \in\{NS, R\}$ let $(\A_q, \sigma)$ be the standard spin annuli described in Example \ref{exRiggedSpinAnnuli}.

\begin{Proposition}\label{propEvalCFTDisk}
$E(\D, NS) = \C \Omega_p \in \F$
\end{Proposition}
\begin{proof}
The boundary parametrization of $(\D, NS)$ is the identity, so $H^2(\D, NS)$ is just the classical Hardy space $H^2(\D)$. Hence the $H^2(\D, NS)$ commutation relations coincide with the $p$ vacuum equations (Definition \ref{defVacuumEquations}) which characterize $\Omega_p$ up to scale .
\end{proof}

\begin{Proposition}\label{propEvalCFTAnnuli}
$E(\A_q, \sigma) = \C q^{L_0^\sigma}$, where both sides of the equation are understood as depending on a fixed choice of $q^{1/2}$ when $\sigma = NS$.
\end{Proposition}
\begin{proof}
Since $q^{L_0^\sigma}$ is trace class (see \cite[\S5.1]{Kac98}), it suffices in both cases to show that $q^{L_0^\sigma}$ satisfies the $H^2(\A_q, \sigma)$ commutation relations. We first consider $\sigma = NS$. Recall that the boundary parametrizations are given by
$$
\beta_{j} (z, \alpha) = \left\{\begin{array}{cl}(z, \alpha) & j = S^1\\ (qz, q^{-1/2}\alpha) & j = rS^1\end{array}\right.
$$
and thus
\begin{align*}
H^2(\A_q, NS) = \cl \operatorname{span} \{(z^n, q^{n+1/2} z^n) : n \in \Z\} \subset H^1 \oplus H^0 := H \oplus H
\end{align*}
and
$$
H^2(\A_q, NS)^\perp = \cl \operatorname{span} \{(z^n, -\overline{q}^{-(n+1/2)} z^n) : n \in \Z\}.
$$
Hence an operator $T \in \B(\F)$ satisfies the $H^2(\A_q,NS)$ commutation relations if and only if
$$
a(z^n)T = q^{n+1/2}Ta(z^n)
$$
and
$$
a(z^n)^*T = q^{-(n+1/2)}Ta(z^n)^*
$$
for all $n \in \Z$.
These equations are satisfied by $q^{L_0^{NS}}$ by \eqref{eqnLCommRel} and \eqref{eqnLCommRelStar}.

Similarly, one has $H^2(\A_q,R) = \cl \operatorname{span} \{(q^{n} z^n, z^n) : n \in \Z\}$, which corresponds to the commutation relations
$$
a(z^n)T = q^{n}Ta(z^n)
$$
and
$$
a(z^n)^*T = q^{-n}Ta(z^n)^*
$$
for all $n \in \Z$.
The operator $q^{L_0^R}$ satisfies these equations by  \eqref{eqnLCommRel} and \eqref{eqnLCommRelStar}.
\end{proof}

\subsection{Vertex operators}

Recall (Example \ref{exRiggedSpinPants}) that the the moduli space of standard Neveu-Schwarz spin pairs of pants with standard boundary parametrizations is
$$
\M_{NS} = \{(w, q_1, q_1^{1/2}, q_2, q_2^{1/2}) \in (\C^\times)^5 : 0<\abs{q_1} + \abs{q_2} < \abs{w} < 1 - \abs{q_1} \}.
$$
Coresponding to a point $x \in \M_{NS}$, we have a spin pair of pants $\bbP_x$, given as a manifold by 
$$
\D \setminus ((q_1 \interior{\D} + w) \cup q_2 \interior{\D})
$$
with spin structure inherited from $\D$. 
The boundary trivializations are
$$
\beta_{j} (z, \alpha) = \left\{
\begin{array}{cl}(z, \alpha) & j = S^1\\ 
(q_1z + w, q_1^{-1/2}\alpha) & j = q_1S^1 + w\\
(q_2z, q_2^{-1/2}\alpha) & j = q_2S^1
\end{array}\right.
$$

We will now show that $E(\bbP_x)$ can be described by the free fermion vertex operator algebra.
We will not give an introduction to vertex operators (see, e.g., \cite{Kac98, Was11}). 
Instead, we will introduce just the necessary objects and properties, with references to the literature.
The free fermion vertex operator algebra is introduced in \cite[\S5.1]{Kac98} under the name ``charged free fermions.''

Let $\F^0 \subset \F$ be the dense subspace spanned algebraically by vectors
$$
a(z^{n_p})^* \cdots a(z^{n_1})^* a(z^{m_1}) \cdots a(z^{m_q})\Omega.
$$
Let $\End(\F^0)$ denote the space of linear (not necessarily bounded) linear endomorphisms of $\F^0$, and let $\End(\F^0)[[z^{\pm 1}]]$ denote the space of formal distributions with coefficients in $\End(\F^0)$. That is, an element of $\End(\F^0)[[z^{\pm 1}]]$ is a formal sum
$$
\sum_{n \in \Z} \xi_n z^{-n-1}
$$
where $\xi_n \in \End(\F^0)$ and $z$ is a formal variable. 

The vertex operator algebra structure on $\F^0$ gives a state-field correspondence 
$$
Y:\F^0 \to \End(\F^0)[[z^{\pm 1}]].
$$
This is commonly written
$$
Y(\xi,z) = \sum_{n \in \Z} \xi_n z^{-n-1}
$$
for $\xi \in \F^0$. The endomorphisms $\xi_n$ are called the {\it modes} of $\xi$ (or of $Y(\xi,z)$).

The vacuum state is assigned to the identity field. That is, $Y(\Omega, z) = \Id$ or more formally
$$
\Omega_n = \delta_{n+1,0} \Id.
$$
The generating fields are those assigned to the states $a(z^0)^*\Omega$ and $a(z^{-1})\Omega$, where we have written $z^0$ for the constant function $z \mapsto 1$. The generating fields are given by
\begin{equation}\label{eqnGeneratingFieldNoStar}
Y(a(z^{-1})\Omega, z) = \sum_{n \in \Z} a(z^n) z^{-n-1}.
\end{equation}
and
\begin{equation}\label{eqnGeneratingFieldStar}
Y(a(z^0)^*\Omega, z) = \sum_{n \in \Z} a(z^{-n-1})^* z^{-n-1}
\end{equation}

The modes of the generating fields extend to bounded operators on $\F$, which is {\it not} a general feature of modes of vertex operators.

The modes of other fields can be reconstructed from the Borcherds product formula (given as \cite[Eqn. 4.8.3]{Kac98} with $m=0$, and in \cite[\S 5]{Was11}):
\begin{Theorem}
Suppose $\xi, \eta \in \F^0$, and write $Y(\xi,z) = \sum_{n \in \Z} \xi_n z^{-n-1}$ and $Y(\eta, z) = \sum_{n \in \Z} \eta_n z^{-n-1}$ for the fields in the free fermion vertex operator algebra. Then the modes of $Y(\eta_n\xi,z)$ are given by the following formula:
\begin{equation}\label{eqnBPF}
(\eta_n\xi)_m = \sum_{j \ge 0} (-1)^j \binom{n}{j} \left(\eta_{n-j} \xi_{m+j} - (-1)^{p(\eta)p(\xi)+n}
\xi_{m+n-j}\eta_{j}\right).
\end{equation}
for homogeneous $\xi$ and $\eta$, and extended linearly in general. 
\end{Theorem}
The modes of any field will satisfy $\xi_n\eta = 0$ for $n$ sufficiently large (depending on $\xi$ and $\eta$), so the sum on the right-hand side of \eqref{eqnBPF} is finite when applied to any fixed vector in $\F^0$.

With this description of the vertex operators in hand, we can prove the main theorem of the section.
\begin{Theorem}\label{thmPantsAreVertexOperators}
Let $x = (w,q_1,q_1^{1/2},q_2,q_2^{1/2}) \in \M_{NS}$. For every $\xi \in \F^0$ and $n \in \Z$, the map $\xi_n q_2^{L_0^{NS}}$ extends to a bounded operator on $\F$. $E(\bbP_x)$ is spanned by the map $T:\F \otimes \F \to \F$ given on $\F^0 \otimes \F^0$ by
\begin{equation}\label{eqnPantsFormula}
T(\xi \otimes \eta) = Y(q_1^{L_0^{NS}}\xi, w) q_2^{L_0^{NS}} \eta = \sum_{n \in \Z} (q_1^{L_0^{NS}}\xi)_n  q_2^{L_0^{NS}}w^{-n-1} \eta.
\end{equation}
We have ordered the input cicles so that the one centered at $w$ comes first. 

For every fixed $\xi \in \F^0$, the sum in \eqref{eqnPantsFormula} converges absolutely in operator norm as a function of $\eta$, uniformly on compact subsets of $\M_{NS}$. 
\end{Theorem}
\begin{proof}

To simplify notation, we will write $L_0$ instead of $L_0^{NS}$ throughout the proof. It suffices to prove the theorem for $\xi$ of the form
\begin{equation}\label{eqnXiBasis}
\xi =  a(z^{n_p}) \cdots a(z^{n_1}) a(z^{m_1})^* \cdots a(z^{m_q})^* \Omega,
\end{equation}
where $n_i \in \Z_{< 0}$ and $m_i \in \Z_{\ge 0}$. Since $q_1^{L_0}$ is invertible as a map $\F^0 \to \F^0$, we will instead prove
\begin{equation}\label{eqnVertexSimplifiedForm}
T(q_1^{-L_0} \xi \otimes \eta) = Y(\xi,w)q_2^{L_0}\eta=  \sum_{n \in \Z} \xi_n  q_2^{L_0}w^{-n-1} \eta
\end{equation}
for some $T \in E(\bbP_x)$, all $\xi$ as in \eqref{eqnXiBasis}, and all $\eta \in \F^0$, with the stated convergence properties.

Let $T$ be a nonzero element of $E(\bbP_x)$. By Corollary \ref{corOperadicComposition} and the calculation of $E(\D)$ and $E(\A_{q_2})$ (Propositions \ref{propEvalCFTDisk} and \ref{propEvalCFTAnnuli}), the map $\eta \mapsto T(\Omega \otimes \eta)$ lies in $\C q_2^{L_0}$. 
Rescale $T$ so that 
\begin{equation}\label{eqnTOmegaEta}
T(\Omega \otimes \eta) = q_2^{L_0}\eta.
\end{equation}
In particular, note that $T(\Omega \otimes \Omega) = \Omega$. Since $T$ is homogeneous by Theorem \ref{thmExistence}, we can conclude that $T$ is even.

We now establish \eqref{eqnVertexSimplifiedForm} by induction on the length of the word in $a(z^n)$ and $a(z^m)^*$'s in \eqref{eqnXiBasis}.

Since $\Omega_n = \delta_{n+1,0}\Id$, equation \eqref{eqnVertexSimplifiedForm} holds when $\xi = \Omega$ by \eqref{eqnTOmegaEta}. The convergence properties are trivial, as the sum only has one term.

Now assume that \eqref{eqnVertexSimplifiedForm} holds for $\xi$,  with the sum converging absolutely in operator norm as a function of $\eta$, uniformly on compact subsets of $\M_{NS}$. We will show that the same holds with $a(z^n) \xi$ and $a(z^{-n-1})^* \xi$ in place of $\xi$, for all $n \in \Z$. 

We first consider $a(z^n) \xi$. From the holomorphic function $(z-w)^n \in \O(\bbP_x)$, we have
\begin{equation}\label{eqnVOAHardyElement}
((z-w)^n, q_1^{n+\frac12} z^n, q_2^{\frac12} (q_2 z- w)^n) \in H^2(\bbP_x, NS).
\end{equation}
Here we have ordered the boundary circles with $S^1$ first, $q_1S^1 + w$ second, and $q_2 S^1$ third. By the definition of $E(\bbP_x, NS)$, $T$ satisfies the commutation relation
$$
a((z-w)^n)T = T(a(q_1^{n+\frac12} z^n ) \grotimes \Id) + T(\Id \grotimes a(q_2^{1/2} (q_2 z- w)^n)).
$$
Hence
\begin{align}\label{eqnVOAProofCommRel}
T(q_1^{-L_0} a(z^n)\xi \otimes \eta) &= T(a(q_1^{n+1/2} z^n) q_1^{-L_0} \xi \otimes \eta) \nonumber\\
&= a((z-w)^n)T(q_1^{-L_0}\xi \otimes \eta) - T(d_\F \otimes a(q_2^{\frac12} (q_2 z - w)^n )) (q_1^{-L_0}\xi \otimes \eta) \nonumber\\
&= a((z-w)^n)T(q_1^{-L_0}\xi \otimes \eta) - (-1)^{p(\xi)}T(q_1^{-L_0}\xi \otimes a(q_2^{\frac12} (q_2 z - w)^n ) \eta ).
\end{align}

We treat the two summands in \eqref{eqnVOAProofCommRel} separately. Since $(z-w)^n$ appears as a function of $z \in S^1$, we can expand it as a power series converging uniformly on compact subsets of $\abs{w} < 1$. Combining this with the inductive hypothesis for $\xi$, we compute

\begin{align}\label{eqnVOACompFirstTermPartOne}
a((z-w)^n)T({q_1}^{-L_0}\xi \otimes \eta) &=\sum_{j \ge 0} (-1)^j \binom{n}{j} a(z^{n-j}) T(q_1^{-L_0} \xi \otimes \eta) w^j \nonumber\\
&=\sum_{j \ge 0}\sum_{m \in \Z} (-1)^j \binom{n}{j} a(z^{n-j}) \xi_m q_2^{L_0} w^{j-m-1} \eta \\
:&= \sum_{j \ge 0}\sum_{m \in \Z} S_{j,m}\eta. \nonumber
\end{align}
Observe that every $S_{j,m}$ is a bounded operator, and since $\norm{a(z^{n-j})} = \norm{z^{n-j}}_{L^2(S^1)} = 1$, we have 
$$
\sum_{j\ge 0} \sum_{m \in \Z} \norm{S_{j,m}} \le \left(\sum_{j \ge 0} \binom{n}{j} \abs{w}^j \right) \left(\sum_{m \in \Z} \norm{\xi_m q_2^{L_0}} \abs{w}^{-m-1} \right).
$$
The sum indexed by $j$ on the right-hand side converges uniformly on compact subsets of $\abs{w} < 1$. 
The sum indexed by $m$ converges by the inductive hypothesis, uniform on compact subsets of $\M_{NS}$.
Hence $\sum_{j \ge 0} \sum_{m \in \Z} S_{j,m}$ is absolutely summable, uniformly on compact subsets of $\M_{NS}$.

We now reindex the sum \eqref{eqnVOACompFirstTermPartOne} in $m$ and exhange the order of summation to get
\begin{align}\label{eqnVOACompFirstTermPartTwo}
\sum_{j \ge 0}\sum_{m \in \Z} S_{j,m}\eta &=\sum_{m \in \Z} \left(\sum_{j \ge 0} (-1)^j \binom{n}{j} a(z^{n-j}) \xi_{m+j}q_2^{L_0}\right) {w}^{-m-1} \eta  \nonumber\\
:&= \sum_{m \in \Z} \tilde S_{m} {w}^{-m-1} \eta,
\end{align}
where $\tilde S_m$ is a bounded operator and the sum \eqref{eqnVOACompFirstTermPartTwo} converges uniformly absolutely in operator norm on compact subsets of $\M_{NS}$.

We now treat the second summand of \eqref{eqnVOAProofCommRel} similarly to how we treated the first, expanding $(q_2 z - w)^n$ as a power series in $q_2/w$. We have
\begin{align}\label{eqnVOACompSecondTermPartOne}
T({{q_1}}^{-L_0}\xi \otimes a(q_2^{\frac12} (q_2 z - w)^n ) \eta ) &= \sum_{j \ge 0}(-1)^j\binom{n}{j} (-1)^n T({q_1}^{-L_0}\xi \otimes a(z^j ) \eta) {q_2}^{j+\frac12}{w}^{n-j}\nonumber\\
&= \sum_{j \ge 0}\sum_{m \in \Z} (-1)^j \binom{n}{j} (-1)^n \xi_m {q_2}^{L_0} a(z^j) {q_2}^{j+\frac12}{w}^{n-m-j-1}\eta\\
:&= \sum_{j \ge 0}\sum_{m \in \Z} U_{j,m} \eta \nonumber.
\end{align}

The operators $U_{j,m}$ are bounded, with 
\begin{equation}\label{eqnUjmNorm}
\sum_{j \ge 0} \sum_{m \in \Z} \norm{U_{j,m}} \le \abs{q_2}^{\tfrac12} \abs{w}^n \left( \sum_{j \ge 0} \binom{n}{j} \abs{\frac{q_2}{w}}^j \right) \left(\sum_{m \in \Z} \norm{\xi_m q_2^{L_0}} \abs{w}^{-m-1}\right).
\end{equation}
By the inductive hypothesis, the right-hand side of \eqref{eqnUjmNorm} converges uniformly on compact subsets of $\M_{NS}$. 
Hence the same summability holds for $\norm{U_{j,m}}$. 
We now rewrite $U_{j,m}$ using the commutation relation \eqref{eqnLCommRel} for $a(z^j)$ and ${q_2}^{L_0}$, along with reindexing and interchanging the sums, to get

\begin{align}\label{eqnVOACompSecondTermPartTwo}
\sum_{j \ge 0}\sum_{m \in \Z} U_{j,m} \eta  &= \sum_{j \ge 0}\sum_{m \in \Z} (-1)^j\binom{n}{j} (-1)^n \xi_m  a(z^j)^*q_2^{L_0}{w}^{n-j-m-1} \eta \nonumber\\
&= \sum_{m \in \Z} \left(\sum_{j \ge 0} (-1)^j\binom{n}{j} (-1)^n \xi_{m+n-j}  a(z^j)q_2^{L_0}\right){w}^{-m-1} \eta\\
:&= \sum_{m \in \Z} \tilde U_m {w}^{-m-1} \eta \label{eqnSumTildeUm}.
\end{align}
Observe that each $\tilde U_m$ is a bounded operator, and the sum \eqref{eqnSumTildeUm} converges uniformly absolutely in operator norm (as a function of $\eta$) on compact subsets of $\M_{NS}$.

From the formula for the generating field \eqref{eqnGeneratingFieldStar}, we see that $(a(z^{-1})\Omega)_n = a(z^n)$. Hence the Borcherds product formula \eqref{eqnBPF} asserts that
\begin{align}\label{eqnVOACompAreWeThereYet}
(a(z^n)\xi)_m &=\sum_{j \ge 0} (-1)^j \binom{n}{j} \left(a(z^{n-j}) \xi_{m+j} - (-1)^{p(\xi)+n}
\xi_{m+n-j}a(z^j)\right) .
\end{align}
Comparing \eqref{eqnVOACompAreWeThereYet} with the definitions of $\tilde S_m$ \eqref{eqnVOACompFirstTermPartTwo} and $\tilde U_m$ \eqref{eqnVOACompSecondTermPartTwo} yields
\begin{equation}\label{eqnTildeTermsGiveModes}
\tilde S_m + (-1)^{p(\xi)}\tilde U_m = (a(z^{n}) \xi)_m {q_2}^{L_0}.
\end{equation}
Plugging the results of the computations \eqref{eqnVOACompFirstTermPartOne} through \eqref{eqnVOACompSecondTermPartTwo} into \eqref{eqnVOAProofCommRel}, and then applying \eqref{eqnTildeTermsGiveModes}, yields
\begin{align}
T({q_1}^{-L_0} a(z^n)\xi \otimes \eta) &= \sum_{m \in \Z}(\tilde S_m + (-1)^{p(\xi)} \tilde U_m) {w}^{-m-1} \eta \nonumber\\
&= \sum_{m \in \Z} (a(z^{n}) \xi)_m {q_2}^{L_0} {w}^{-m-1} \eta, \label{eqnVertexCompResult}
\end{align}
which establishes \eqref{eqnVertexSimplifiedForm} for $a(z^n)\xi$. The required convergence property of the sum \eqref{eqnVertexCompResult} follows from the corresponding convergence properties of $\sum \tilde S_m {w}^{-m-1}$ and $\sum \tilde U_m {w}^{-m-1}$ that we previously established.

To complete the proof, we must establish \eqref{eqnVertexSimplifiedForm} with $a(z^{-n-1})^*\xi$ in place of $\xi$. This is nearly identical to the computation above for $a(z^n)\xi$, so we will only sketch the argument. By Theorem \ref{thmHardyPerp}, $H^2(\bbP_x, NS)^\perp = \overline{M_{\pm z} H^2(\bbP_x,NS)}$, where $M_{\pm z}$ is multiplication by the function $z$ on outgoing boundary components and multiplication by $-z$ on incoming boundary components. 

We saw in \eqref{eqnVOAHardyElement} that 
$$
((z-w)^n, q_1^{n+\frac12} z^n, q_2^{\frac12} (q_2 z- w)^n) \in H^2(\bbP_x, NS),
$$
and so 
$$
(z^{-1}(z^{-1}-\overline{w})^n, \; -\overline{q_1}^{n+\frac12} z^{-n-1}, \;\overline{q_2}^{\frac12} z^{-1}(\overline{q_2}z^{-1} - \overline{w})^n) \in H^2(\bbP_{x}, {NS})^\perp.
$$
By the definition of $E(\bbP_x,NS)$, we have
$$
a((z^{-1}-\overline{w})^n)^*T = T(a(\overline{q_1}^{n+\frac12} z^{-n-1})^* \grotimes \Id) + T(\Id \grotimes a(\overline{q_2}^{\frac12} z^{-1}(\overline{q_2}z^{-1} - \overline{w})^n)^*)
$$
and thus
$$
T({q_1}^{-L_0} a(z^{-n-1})^*\xi \otimes \eta) = a((z^{-1}-\overline{w})^n)^*T({q_1}^{-L_0}\xi  \otimes \eta) - (-1)^{p(\xi)} T({q_1}^{-L_0}\xi \otimes a(\overline{q_2}^{\frac12} z^{-1}(\overline{q_2}z^{-1} - \overline{w})^n)^*\eta).
$$
We can now establish the desired formula for the left-hand side by expanding $(z^{-1} - \overline{w})^n$ in the domain $\abs{w} < 1$, expanding $(\overline{q_2}z^{-1} - \overline{w})^n$ in the domain $\abs{q_2} < \abs{w}$, and applying the inductive hypothesis, just as before.
\end{proof}

\section{The Cauchy transform for Riemann surfaces}
\label{secCauchyTransform}

\subsection{Main theorems}
When establishing the properties of the free fermion Segal CFT in Section \ref{subsecCFTProof}, we deferred the proof of two key analytic properties of the Hardy space $H^2(X)$.
In order to prove the sewing property, we needed a formula for $H^2(X)^\perp$:

\begin{Theorem}\label{thmHardyPerp}
Let $X=(\Sigma, L, \Phi, \beta) \in \cR$ be a spin Riemann surface with boundary parametrization. Let $H_\Gamma = \bigoplus_{j \in \pi_0(\Gamma)} L^2(S^1)$ and let $H^2(X) \subset H_\Gamma$ be the Hardy space. Then 
\begin{equation}\label{eqHardyPerp}
H^2(X)^\perp = M_{\pm} \overline{M_z^{NS} H^2(X)} = M_{\pm} H^2(\overline{X}).
\end{equation}

Here $M_{\pm}$ is multiplication by $1$ on copies of $L^2(S^1)$ indexed by outgoing boundary components, and multiplication by $-1$ on copies of $L^2(S^1)$ indexed by incoming boundary components, and $M_z^{NS}$ is multiplication by the function $z$ on copies of $L^2(S^1)$ indexed by $j$ for which $L|_j$ is Neveu-Schwarz, and the identity on other boundary components.
\end{Theorem}

In order to establish non-triviality of the spaces $E(X)$, we required the following theorem.

\begin{Theorem}\label{thmProjectionDifferenceTraceClass}
Let $X=(\Sigma, L, \Phi, \beta) \in \cR$ be a spin Riemann surface with boundary parametrization. Let $H_\Gamma = \bigoplus_{j \in \pi_0(\Gamma)} L^2(S^1)$ and let $H^2(X) \subset H_\Gamma$ be the Hardy space. Let $q_X \in \P(H_\Gamma)$ be the projection onto $H^2(X)$, and let 
$$
p_\Gamma = \bigoplus_{j \in \pi_0(\Gamma^1)} p \oplus \bigoplus_{j \in \pi_0(\Gamma^0)} \Id - p,
$$
where $p \in \P(L^2(S^1))$ is the projection onto $H^2(\D)$. Then $q_X - p_\Gamma$ is trace class.
\end{Theorem}

The main tool for establishing Theorems \ref{thmHardyPerp} and \ref{thmProjectionDifferenceTraceClass} will be a generalization of the Cauchy transform to Riemann surfaces. A treatment of these theorems when $\Sigma$ is a planar domain appears in the book of Bell \cite[\S1-5]{Bell92}. We will follow Bell's treatment, making adjustments for the non-planar case when needed and reducing to the planar case when possible.

The author would like to thank Antony Wassermann for suggesting the reference \cite{Bell92}, and for explaining the role of the Cauchy transform in proving Theorem \ref{thmProjectionDifferenceTraceClass} in the planar case.

\subsection{The Cauchy transform}

\subsubsection{Definitions}

Let $\Sigma$ be a compact Riemann surface with no closed components, and let $\Gamma = \partial \Sigma$. 
By welding annuli onto each component of $\Gamma$ as in Theorem \ref{thmWelding}, we may assume that $\Sigma$ is embedded in an open Riemann surface $\tilde \Sigma$.

By \cite{GuNa67}, there exists a locally injective holomorphic map $\rho: \tilde \Sigma \to \C$.  
By \cite{Scheinberg78}, there exists a meromorphic  function $q(s,t):\tilde \Sigma \times \tilde \Sigma \to \C$ which is holomorphic except on the diagonal $s=t$, and such that $q(s,t) - (\rho(s) - \rho(t))^{-1}$ is holomorphic on $U \times U$ for any open $U$ on which $\rho$ is injective. We can assume that $q(s,t) = -q(t,s)$ by replacing $q$ with $\frac12 q(s,t) - \frac12 q(t,s).$  Let $\omega_t(s) = q(s,t)d\rho(s)$.

We call $q$ a Cauchy kernel on $\tilde \Sigma$, which is justified by the following Cauchy integral formula.

\begin{Proposition}[{\cite[Prop. 7.1]{Scheinberg78}}]
Let $U$ be an open set in $\tilde \Sigma$  with $\overline{U}$ compact, and with a piecewise $C^1$ oriented boundary $\partial U$.
If $u \in C^1(\overline{U})$, then for every $t \in U$,
$$
u(t) = \frac{1}{2\pi i} \int_{\partial U} u \omega_t - \frac{1}{2\pi i} \int_U \overline{\partial}u \wedge  \omega_t.
$$
\end{Proposition}

We denote by $C^\infty(\interior{\Sigma})$ and $\O(\interior{\Sigma})$ the smooth (resp. holomorphic) functions on the interior of $\Sigma$.
We will write $C^\infty(\Sigma)$ for the subspace of $C^\infty(\mathring{\Sigma})$  consisting of functions which extend to smooth functions on the boundary, and $\O(\Sigma)$ for the subspace of $C^\infty(\Sigma)$ consisting of functions which are holomorphic in the interior.

\begin{Definition}
If $u \in C^\infty(\Gamma)$, then define its Cauchy transform $\CC u \in \O(\mathring{\Sigma})$ by
$$
(\CC u)(t) = \frac{1}{2\pi i} \int_{\Gamma} u\omega_t.
$$
\end{Definition}
This definition has appeared many places in the literature, with early examples including \cite{Scheinberg78,Ga77,Bo87}.

\subsubsection{Basic properties}

Note that $\CC$ depends on the choice of $\rho$ and $q$, so we will regard these as fixed.  We will now show that $\CC u \in \O(\Sigma)$, but first we need the following version of \cite[Thm. 2.2]{Bell92}.

\begin{Theorem}
Suppose $v \in C^\infty(\Sigma)$.  Then the function $u$ defined by
$$
u(t) = \frac{1}{2 \pi i} \int_{\Sigma} v\omega_t \wedge d\rhobar
$$
for $t \in {\Sigma}$ satisfies $\overline{\partial}u = v d\rhobar$ and $u \in C^\infty(\Sigma)$.
\end{Theorem}
\begin{proof}
We first check that the integral defining $u$ makes sense. Fix $t_0 \in \Sigma$, and let $V$ be a neighborhood of $t_0$ in $\Sigma$ on which $\rho$ is injective. Let $z_0 = \rho(t_0)$, and let $\tau = (\rho|_V)^{-1}$. For $z\in \rho(V)$ we have an identity of $1$-forms on $\rho(V)$
$$
\tau^* \omega_{\tau(z)} = \frac{dw}{w-z} + f(z,w)dw,
$$
where $f$ is holomorphic and $w$ is the standard global parameter for $\C$. 
We then have
\begin{align*}
u(\tau(z)) &= \frac{1}{2 \pi i} \int_{\Sigma  \setminus V} v\omega_{\tau(z)} \wedge d\rhobar +  \frac{1}{2\pi i} \int_{\rho(V)} \frac{v(\tau(w)) dw \wedge d\overline{w}}{w-z} \;+\\  
&\quad\quad+ \; \frac{1}{2\pi i} \int_{\rho(V)} v(\tau(w))f(z,w) dw \wedge d\overline{w}\\
&:= u_1(z) + u_2(z) + u_3(z).
\end{align*}
Both $u_1$ and $u_3$ are clearly smooth in a neighborhood of $z_0$. From \cite[Thm 2.2]{Bell92}, $u_2$ is well-defined and $u_2 \in C^\infty(\rho(V))$. Thus $u$ is smooth in a neighborhood of $t_0$, and since $t_0$ was arbitrary $u \in C^\infty(\Sigma)$.

Differentiating under the integral, we see that 
$$
\frac{\partial}{\partial \overline{z}} u \circ \tau = \frac{\partial}{\partial \overline{z}} u_2 = v \circ \tau
$$
by  \cite[Thm 2.2]{Bell92}. Pulling back by $\rho$ gives $\overline{\partial} u = v d\overline{\rho}$ on $V$, and since $z_0$ was arbitrary, the equality holds on all of $\Sigma$.
\end{proof}

As a corollary, we can show that $\CC u$ extends smoothly to the boundary.
\begin{Proposition}
The Cauchy transform maps $C^\infty(\Gamma)$ into $\O(\Sigma)$.
\end{Proposition}
\begin{proof}
Let $u \in C^\infty(\Gamma)$ and let $\tilde u$ be a function in $C^\infty(\Sigma)$ which is equal to $u$ on $\Gamma$.  The Cauchy integral formula says
$$
\tilde u(t) = (\CC u)(t) - \frac{1}{2\pi i} \int_{\Sigma} \overline{\partial}\tilde u \wedge  \omega_t.
$$
We can write $\overline{\partial}\tilde u = v d\overline{\rho}$ for some $v \in C^\infty(\Sigma)$, so by the preceding theorem, the integral term is in $C^\infty(\Sigma)$.  Hence $\CC u \in C^\infty(\Sigma)$ as well.
\end{proof}

By restriction, we can consider $\CC $ as a map from $C^\infty(\Gamma)$ into itself. The Cauchy integral formula says that $\CC $ is idempotent. 

We will need the following technical results, which are a generalization of \cite[Lem. 2.3 and Thm 3.4]{Bell92}.
\begin{Proposition}\label{propBell23}
Suppose that $v \in C^\infty(\Sigma)$.  Then there exists a function $\Phi \in C^\infty(\Sigma)$ which vanishes on $\Gamma$ and satisfies $\overline{\partial} \Phi|_\Gamma = \overline{\partial}v|_\Gamma$.
\end{Proposition}
\begin{proof}
We may choose annular neighborhoods $U_j$ in $\Sigma$ of each boundary component $j$, and holomorphically identify these with annuli in $\C$. Thus by the planar version of the proposition \cite[Lem. 2.3]{Bell92}, there exist smooth functions on each $U_j$ with the desired property. Since the conclusion only depends on a neighborhood of $\Gamma$, we can extend these functions to $\Sigma$ via smooth cutoff functions with support in the $U_j$ and which are identically 1 in a neighborhood of $\Gamma$.
\end{proof}

\begin{Proposition}\label{propBoundaryValFormula}
Suppose that $u \in C^\infty(\Gamma)$.  Then there is a $\Psi \in C^\infty(\Sigma)$ with $\overline{\partial}\Psi|_\Gamma = 0$ such that the boundary values of $\CC u$ are expressed by
$$
(\CC u)(t) = u(t) + \frac{1}{2\pi i} \int_{{\Sigma}} \overline{\partial}\Psi \wedge \omega_t,
$$
for all $t \in \Gamma$.  The 2-form $(\overline{\partial}\Psi \wedge \omega_t)(s)$ extends continuously to $(s,t) \in \Sigma \times \Gamma$.
\end{Proposition}
\begin{proof}
Let $\tilde u$ be an element of $C^\infty(\Sigma)$ with boundary values $u$.  Let $\Phi \in C^\infty(\Sigma)$ be a function from Proposition \ref{propBell23} that vanishes on $\Gamma$ such that $\overline{\partial}\Phi|_\Gamma = \overline{\partial}\tilde u|_\Gamma$. 
Let 
$
\Psi = \tilde u - \Phi.
$
Applying the Cauchy integral formula to $\Psi$ yields
$$
\Psi(t) = (\CC u)(t) - \frac{1}{2\pi i} \int_{\Sigma} \overline{\partial}\Psi \wedge \omega_t.
$$
Since $\Psi = u$ on the boundary, we have established the desired boundary value formula for $\CC u$.

The 2-form $\overline{\partial}\Psi \wedge \omega_t $ is clearly continuous at all points of $\Sigma \times \Gamma$ not of the form $(t_0,t_0)$ with $t_0 \in \Gamma$. 
Fix a neighborhood $V$ of $t_0$ on which $\rho$ is injective and set $z = \rho(t)$ and $\tau = \rho|_V^{-1}$. We have
$$
\tau^*(\overline{\partial}\Psi \wedge \omega_t )(w) = \left(\frac{\partial_{\overline{w}}(\Psi \circ \tau)(w)}{w-z} + \mbox{smooth}\right)d\overline{w} \wedge dw
$$
for $(w,z) \in \rho(V) \times \rho(V)$ with $w \ne z$. Since $\partial_{\overline{w}}(\Psi \circ \tau)$ is smooth and vanishes on $\rho(\Gamma \cap V)$, the above expression defines a continuous function on $\rho(V) \times \rho(V \cap \Gamma)$.
Pulling back by $\rho$, we see that $(s, t) \mapsto (\overline{\partial}\Psi \wedge \omega_t)(s)$ extends continuously to $\Sigma \times \Gamma$.
\end{proof}

We will now define the Hilbert transform for $C^\infty(\Gamma)$, and relate it to the Cauchy transform. 
If $t_0 \in \Gamma$, let $V$ be a neighborhood of $t_0$ in $\Sigma$ on which $\rho$ is injective, and let
$$
\Gamma_\epsilon = (\Gamma \setminus V) \cup \{ t \in V : \abs{\rho(t) - \rho(t_0)} \ge \epsilon\}.
$$ 
Observe that for a different choice of $V$, the resulting sets $\Gamma_\epsilon$ coincide for sufficiently small $\epsilon$.
Define the Hilbert transform $\H u$ for $u \in C^\infty(\Gamma)$ by
$$
(\H u)(t_0) = \PV \frac{1}{2\pi i} \int_\Gamma u\omega_{t_0} := \lim_{\epsilon \downarrow 0} \frac{1}{2\pi i}\int_{\Gamma_\epsilon} u\omega_{t_0}.
$$

We will now establish the Plemelj formula relating the Cauchy and Hilbert transforms, as in \cite[\S 5]{Bell92}.
\begin{Lemma}\label{lemPlemelj} The limit defining $(\H u)(t_0)$ exists and
$$
(\CC u)(t_0) = \frac{1}{2} u(t_0) + (\H u)(t_0).
$$
\end{Lemma}
\begin{proof}
We first prove the theorem in the case where $u$ is a constant function.  Let 
$$
C_\epsilon = \{t \in V : \abs{\rho(t) - \rho(t_0)} = \epsilon\},
$$
oriented so that $\Gamma_\epsilon \cup C_\epsilon$ is an oriented curve for sufficiently small $\epsilon$ (i.e. so that $C_\epsilon$ is oriented negatively around $t_0$). 
We give $\rho(C_\epsilon)$ the opposite of the orientation coming from $C_\epsilon$, so that it is oriented counterclockwise about $\rho(t_0)$. 
Let $\tau = \rho|_V^{-1}$. Using the holomorphicity of $u(s)\omega_{t_0}(s)$ away from $s=t_0$ and the fractional residue formula, we compute
\begin{align*}
\lim_{\epsilon \downarrow 0} \frac{1}{2\pi i}\int_{\Gamma_\epsilon} u(s)\omega_{t_0}(s) &=
\lim_{\epsilon \downarrow 0} \frac{1}{2\pi i}\int_{\Gamma_\epsilon} u(s) q(s,t_0) d\rho(s)\\
&=-\lim_{\epsilon \downarrow 0} \frac{1}{2\pi i}\int_{C_\epsilon} u(s) q(s,t_0) d\rho(s)\\
&= \lim_{\epsilon \downarrow 0} \frac{1}{2\pi i}\int_{\rho(C_\epsilon)} \frac{u(\tau(w))}{w-\rho(t_0)} dw\\
&= \frac{1}{2} u(t_0).
\end{align*}

We now return to arbitrary $u \in C^\infty(\Gamma)$, but we assume without loss of generality that $u(t_0) = 0$.  
Hence the integrand in the Hilbert and Cauchy transforms $u\omega_{t_0}$ is continuous at $t_0$, and thus on $\Sigma$.  
In this case $(\H u)(t_0)$ is given by the ordinary integral
$$
(\H u)(t_0) = \frac{1}{2\pi i} \int_\Gamma u\omega_{t_0},
$$
and the same for $(\CC u)(t_0)$.
\end{proof}

\subsubsection{Adjoint of the Cauchy transform}

Define a bilinear form $[\cdot, \cdot]$ on $C^\infty(\Gamma)$ by 
$$
[u,v] = \frac{1}{2\pi i}\int_\Gamma uv d\rho.
$$

\begin{Lemma}\label{lemCauchyAdjoint}
For $u,v \in C^\infty(\Gamma)$, we have $[\CC u, v] = [u, (\Id - \CC)v]$.
\end{Lemma}
\begin{proof}
By Proposition \ref{propBoundaryValFormula}, for $t \in \Gamma$ we have $\CC  u = u + I$ where
$$
I(t) =\frac{1}{2\pi i} \int_{{\Sigma}} \overline{\partial}\Psi \wedge \omega_{t}
$$
and $\Psi$ is as in Proposition \ref{propBoundaryValFormula}.
By Proposition \ref{propBoundaryValFormula}, the integrand in the definition of $I$ is continuous, and so we may apply Fubini's theorem to compute
\begin{align}
\int_\Gamma I(t) v(t) d\rho(t) &= \int_\Gamma \left(\frac{1}{2\pi i} \int_{{\Sigma}} \overline{\partial}\Psi(s) \wedge \omega_{t}(s) \right) v(t) d\rho(t)\nonumber\\
&= \int_{{\Sigma}} \left(\frac{1}{2\pi i} \int_\Gamma -q(s,t)v(t) d\rho(t) \right) d\rho(s) \wedge \overline{\partial}\Psi(s)\nonumber\\
&= \int_{{\Sigma}} \left(\frac{1}{2\pi i} \int_\Gamma q(t,s)v(t) d\rho(t) \right) d\rho(s) \wedge \overline{\partial}\Psi(s)\nonumber\\
&=\int_{{\Sigma}}  (\CC v)(s) \, d\rho(s) \wedge \overline{\partial}\Psi(s) \label{eqnIntAgainstI}.
\end{align}
Recall that $\Psi|_\Gamma = u|_\Gamma$. Since $\rho$ and $\CC v$ are holomorphic, 
\begin{equation*}
d(\Psi(\CC v)  d\rho) = - (\CC v) \, d\rho \wedge \overline{\partial}\Psi
\end{equation*}
and we may apply Stokes' theorem to obtain
\begin{align}\label{eqnIntAgainstICont}
\int_{{\Sigma}}  (\CC v) d\rho \wedge \overline{\partial}\Psi &= -\int_\Gamma \Psi (\CC v) d\rho 
= -\int_\Gamma u(\CC v) d\rho.
\end{align}
Combining \eqref{eqnIntAgainstI} and \eqref{eqnIntAgainstICont}, we get
$$
\int_\Gamma I(t) v(t) d\rho(t) = -\int_\Gamma u(\CC v) d\rho.
$$
Hence
$$
\int_\Gamma (\CC u) v d\rho = \int_\Gamma (u + I) v d\rho = \int_\Gamma u(v - \CC v) d\rho,
$$
which was to be shown.
\end{proof}

Let $\gamma: \bigsqcup_{j \in \pi_0(\Gamma)} S^1 \to \Gamma$ be family of diffeomorphisms.
Let $\Gamma^0$ be the subset of the boundary consisting of boundary components on which $\gamma$ is orientation reversing, and $\Gamma^1$ be the complement, on which $\gamma$ is orientation preserving. 

Let $H_\Gamma = \bigoplus_{j \in \pi_0(\Gamma)} L^2(S^1)$, and let $W_\Gamma = C^\infty(\sqcup_{j \in \pi_0(\Gamma)} S^1)  \subset H_\Gamma$. Define the Hardy space 
$$
H^2(\Sigma, \gamma) = \cl \{ \gamma^*F : F \in \O(\Sigma) \} \subseteq H_\Gamma.
$$

Using the parameterization $\gamma: \bigsqcup S^1 \to \Gamma$, we may identify $C^\infty(\Gamma)$ with $W_\Gamma$.
Thus the Cauchy transform $\CC \in \End(C^\infty(\Gamma))$ induces a linear map $C \in \End(W_\Gamma)$ by
$$
Cu = \gamma^*\CC(u \circ \gamma^{-1}).
$$

Let $r \in W_\Gamma$ be given by $\gamma^*d\rho = rdz$.
Define the formal adjoint $C^* \in \End(W_\Gamma)$ by 
\begin{equation}\label{eqnCauchyFormalAdjoint}
(C^* v)(z) := v(z) - \overline{\pm z r(z) C(M_{\pm \overline{z}} r^{-1}\overline{v})(z)},
\end{equation}
where $M_{\pm \overline{z}}$ is the operator on $H_\Gamma$ given by multiplication by the function $\overline{z}$ on direct summands indexed by $j \in \pi_0(\Gamma^1)$, and multiplication by $-\overline{z}$ on the complement. 
We think of $C$ and $C^*$ as unbounded operators on $H_\Gamma$ (although the adjoint of $C$ will turn out to actually be an extension of $C^*$, since we will see that $C$ is bounded).

\begin{Proposition}\label{propFormalAdjoint}
Let $u, v \in W_\Gamma \subset H_\Gamma$. 
Then
$
\ip{C u, v} = \ip{u, C^* v}.
$
\end{Proposition}
\begin{proof}
Let $\tilde u,\tilde v \in C^\infty(\Gamma)$ be given by $\tilde u = u \circ \gamma^{-1}$ and $\tilde v = v \circ \gamma^{-1}$. 
Then we have
\begin{align*}
\ip{u, M_{\pm \overline{z}}\overline{r v}} &= \frac{1}{2\pi i}\int_{\bigsqcup S^1} \pm u(z) v(z) r(z) dz\\
&=\frac{1}{2\pi i} \int_\Gamma \tilde u(t) \tilde v(t) d\rho(t)\\
&= [\tilde u, \tilde v].
\end{align*}
By Lemma \ref{lemCauchyAdjoint}, $[\CC\tilde u, v] = [\tilde u, (\Id - \CC) \tilde v]$.
Hence
$$
\ip{C u, M_{\pm \overline{z}}\overline{r v}} = \ip{u,M_{\pm \overline{z}}\overline{r (\Id -C)v }},
$$
which was to be shown.
\end{proof}

We now establish the Kerzman-Stein formula
\begin{equation}\label{eqKS}
q_\Sigma(\Id  + A) = C
\end{equation}
where $A = C - C^*$ and $q_\Sigma$ is the orthogonal projection of $H_\Gamma$ onto $H^2(\Sigma,\gamma)$. 
For now, we regard \eqref{eqKS} as an identity of endomorphisms of $W_\Gamma$.
Soon, however, we will show that $A$ is trace class, and thus $C$ extends to a bounded operator on $H_\Gamma$, and \eqref{eqKS} gives an equality of operators on $H_\Gamma$.

\begin{Proposition}\label{propKS}
If $u \in W_\Gamma$, then
$$
q_\Sigma(I + A)u = C u.
$$
\end{Proposition}
\begin{proof}
For $v \in H^2(\Sigma,\gamma)$ we have
\begin{align*}
\ip{(\Id - C^*)u,v} &= \ip{u,v} - \ip{C^* u, v} = \ip{u,v} - \ip{u, C v} = 0.
\end{align*}
Thus $(\Id-C^*)u$ is orthogonal to any smooth function in $H^2(\Sigma,\gamma)$.  By construction, such functions are dense in $H^2(\Sigma,\gamma)$ so we have $q_\Sigma(\Id-C^*)u = 0$. We now have
$$
q_\Sigma(\Id + A)u = q_\Sigma C u = C u.
$$
\end{proof}

Our proof that $A$ is an integral operator with smooth kernel follows \cite[Ch. 4-5]{Bell92}.  
\begin{Theorem}\label{thmKSSmooth}
For $u \in W_\Gamma$, the operator $A = C - C^*$ is given by the formula
$$
(A u)(z) = \frac{1}{2\pi } \int_{\bigsqcup_{j \in\pi_0(\Gamma)} S^1} a(w,z)u(w) \frac{dw}{iw}
$$
for a smooth function $a: \bigsqcup S^1 \times \bigsqcup S^1 \to \C$. In particular, $A$ is trace class.
\end{Theorem}
\begin{proof}
Recall that for $u \in W_\Gamma$, the formal adjoint $C^*$ is given by the formula
$$
(C^*u)(z) = u(z) - \overline{\pm z r(z) C(M_{\pm \overline{z}} r^{-1}\overline{u})(z)},
$$
where $r(z) dz = \gamma^* d\rho$.

By definition,
$$
(Cu) = \CC(u \circ \gamma^{-1}) \circ \gamma.
$$
Thus we can apply Lemma \ref{lemPlemelj} to get
$$
(Au)  = \H(u \circ \gamma^{-1}) \circ \gamma + \overline{M_{\pm \overline{z}}r (\H(v \circ \gamma^{-1}) \circ \gamma)}
$$
where $v = M_{\pm \overline{z}} r^{-1} \overline{u}$. That is, for $z \in \bigsqcup_{j \in \pi_0(\Gamma)} S^1$ we have
\begin{align*}
(A u)(z) &= \frac{1}{2\pi i} \PV \int_\Gamma u(\gamma^{-1}(s))q(s,\gamma(z))d\rho(s) \\
&\quad - \overline{\frac{1}{2\pi i} \pm z r(z) \PV \int_\Gamma \pm \overline{\gamma^{-1}(s)}r(\gamma^{-1}(s))^{-1}\overline{u(\gamma^{-1}(s))} q(s,\gamma(z)) d\rho(s)}
\end{align*}
where the two $\pm$ are determined by whether the boundary near $s$ and $\gamma(z)$ is incoming or outgoing. It is clear that the kernel of $A$ is smooth in any neighborhood of $(s,t)$ when $s$ and $t$ lie on distinct components of $\Gamma$. Thus in order to simplify notation, we will assume that $\Gamma$ has a single outgoing connected component, and the general case is no different. When restricting to $s$ and $t$ on the same connected component, the signs $\pm$ cancel. 

Pulling the integral back to $S^1$, we get
\begin{align*}
(A u)(z) &= \frac{1}{2\pi} \PV \int_{S^1} w u(w)q(\gamma(w),\gamma(z))r(w)\frac{dw}{i w}\\
&\quad + \overline{\frac{1}{2\pi} z r(z) \PV \int_{S^1}  \overline{u(w)} q(\gamma(w),\gamma(z)) \frac{dw}{iw}}\\
&= \frac{1}{2\pi} \PV \int_{S^1} w r(w) q(\gamma(w),\gamma(z))u(w)\frac{dw}{iw} \\
&\quad + \frac{1}{2\pi}  \PV \int_{S^1}  \overline{z r(z)q(\gamma(w),\gamma(z))} u(w)  \frac{dw}{iw}\\
&= \frac{1}{2\pi} \PV \int_{S^1} a(w,z) u(w) \frac{dw}{iw}
\end{align*}
where
$$
a(w,z) = wr(w)q(\gamma(w),\gamma(z)) + \overline{zr(z)q(\gamma(w),\gamma(z))}.
$$
Clearly $a$ is smooth away from $w=z$, so we fix $z$ and consider when $w-z$ is small. In this scenario, we may write 
\begin{align}
a(w,z) &= \frac{wr(w)}{\rho(\gamma(w)) - \rho(\gamma(z))} + \frac{\overline{zr(z)}}{\overline{\rho(\gamma(w))} - \overline{\rho(\gamma(z))}} + \mbox{ smooth} \nonumber\\
&= \frac{w}{w-z} \left(\frac{r(w)(w-z)}{\rho(\gamma(w)) - \rho(\gamma(z))} - \frac{\overline{r(z)(w-z)}}{\overline{\rho(\gamma(w))} - \overline{\rho(\gamma(z))}}\right) + \mbox{ smooth}. \label{eqnKSKernelexpansion}
\end{align}
Since $\tfrac{\partial}{\partial w} \,\rho \circ \gamma = r$, we have that $(w-z)a(w,z)$ is a smooth function vanishing on the diagonal $(z,z)$. Hence $a(w,z)$ is itself a smooth function.
\end{proof}

\begin{Theorem}\label{thmWPerp}
Let $H^2(\Sigma, \gamma) \subset H_\Gamma$ be the Hardy space, and let $q_\Sigma \in \P(H_\Gamma)$ be the projection onto $H^2(\Sigma, \gamma)$. 
Then the Cauchy transform $C$ extends to a bounded operator on $H_\Gamma$ and $q_\Sigma - C$ is trace class.  
We have $H^2(\Sigma,\gamma)^\perp = \overline{r} M_{\pm \overline{z}} \overline{H^2(\Sigma,\gamma)}$, where $r$ satisfies $\gamma^*R = rdz$ for any non-vanishing holomorphic 1-form $R$. 
In particular, one may take $R = d\rho$. 
\end{Theorem}
\begin{proof}
The fact that $C$ is bounded follows immediately from Proposition \ref{propKS} and the fact that $A$ is bounded.  Rewriting the Kerzman-Stein formula as $q_\Sigma - C =  - q_\Sigma A$ we can see that $q_\Sigma - C$ is trace class.

Since $C$ is an idempotent with image $H^2(\Sigma,\gamma)$, we have that $\Id - C^*$ is an idempotent with image $H^2(\Sigma,\gamma)^\perp$. 
Since $C$ is bounded, the formula for the formal adjoint from Lemma \ref{lemCauchyAdjoint} indeed gives the adjoint. 
It follows that $H^2(\Sigma,\gamma)^\perp = \overline{r M_{\pm z} H^2(\Sigma,\gamma)}$, where $\gamma^*d\rho = r(z)dz$. Since $H^2(\Sigma,\gamma)$ is invariant under multiplication by $\gamma^*F$ for any $F \in \O(\Sigma)$, the formula for $H^2(\Sigma,\gamma)^\perp$ holds when $\gamma^*R = r(z) dz$ for any non-vanishing holomorphic 1-form R.
\end{proof}

Recall that the Cauchy transform $\CC$ for $C^\infty(\Gamma)$ depended on a choice of holomorphic immersion $\rho$ and Cauchy kernel $q$. 
The induced Cauchy transform $C \in \B(H_\Gamma)$ also depended on the boundary parametrization $\gamma$. 
However we will see that, modulo a trace class perturbation, $C$ does not actually depend on the choices of  $\rho$, $q$ and $\gamma$.
That $C$ is independent of $\rho$ and $q$ modulo trace class operators is a simple corollary of Theorem \ref{thmWPerp}.

\begin{Corollary}\label{corCauchyIndependent}
Suppose $C_1$ and $C_2$ are two Cauchy transforms for $H_\Gamma$ coming from different choices of $q$ and $\rho$. Then $C_1 - C_2$ is trace class.
\end{Corollary}
\begin{proof}
Note that $q_\Sigma$ only depends on $H^2(\Sigma, \gamma)$, and not on $\rho$ or $q$. Thus $C_1 - C_2$ is trace class by Theorem \ref{thmWPerp}.
\end{proof}
Let $p \in \P(L^2(S^1))$ be the projection onto the classical Hardy space $H^2(\D)$, and let 
$$
p_\Gamma = \bigoplus_{j \in \pi_0(\Gamma^1)} p \oplus \bigoplus_{j \in \pi_0(\Gamma^0)} \Id - p \in \P(H_\Gamma).
$$
Let $q_\Sigma \in \P(H_\Gamma)$ be the projection onto $H^2(\Sigma, \gamma)$.

We wish to show that $q_\Sigma - p_\Gamma$ is trace class.
We begin by showing that this property is independent of the choice of $\gamma$.
First, a simple observation relating idempotents and range projections.
\begin{Proposition}\label{propGLresPreservesGrres}
Let $K$ be a Hilbert space, and let $p$ be a projection on $K$. Let $c$ be an idempotent operator on $K$ with $c - p$ trace class, and let $q$ be the range projection of $c$. Then $q - p$ is trace class. 
\end{Proposition}
\begin{proof}
Since $c- p$ is trace class, so is $(c-p) - (c-p)^* = c - c^*$. We compute 
\begin{align*}
q - p &= cq - p\\
&= (c-c^*)q + (qc - p)^*\\ 
&= (c-c^*)q + (c-p)^*
\end{align*}
which is evidently trace class.
\end{proof}

\begin{Proposition}\label{propReparametrizeHardyHS}
Let $\Sigma$ be a compact Riemann surface,  let $\gamma$ be a family of boundary trivializations for $\Sigma$, and let $q_\Sigma \in \P(H_\Gamma)$ be the projection onto $H^2(\Sigma, \gamma)$. 
Let $\alpha: \bigsqcup_{j \in \pi_0(\Gamma)} S^1 \to \bigsqcup_{j \in \pi_0(\Gamma)} S^1$ be a family of orientation preserving diffeomorphisms, and let $q_\Sigma^\prime \in \P(H_\Gamma)$ be the projection onto $H^2(\Sigma, \gamma \circ \alpha^{-1})$. 
Then $q_\Sigma - p_\Gamma$ is trace class if and only if $q_\Sigma^\prime - p_\Gamma$ is trace class.
\end{Proposition}
\begin{proof}
Suppose that $q_\Sigma - p_\Gamma$ is trace class. Let $c_\alpha$ be the bounded operator on $H_\Gamma$ given by $f \mapsto f \circ \alpha^{-1}$. Observe that 
$$
H^2(\Sigma, \gamma\circ \alpha^{-1}) = c_\alpha H^2(\Sigma, \gamma).
$$
Thus $c_\alpha q_\Sigma c_\alpha^{-1}$ is an idempotent whose range projection is $q_\Sigma^\prime$. 
But $[c_\alpha, p_\Gamma]$ is trace class by \cite[Prop. 6.3.1 and Prop. 6.8.2]{PrSe86}, and so $c_\alpha p_\Gamma c_\alpha^{-1} - p_\Gamma$ is trace class as well. 
Since $c_\alpha q_\Sigma c_\alpha^{-1} - c_\alpha p_\Gamma c_{\alpha}^{-1}$ is trace class by assumption, we must also have that $c_\alpha q_\Sigma c_\alpha^{-1} - p_\Gamma$ is trace class.
By Proposition \ref{propGLresPreservesGrres} we can conclude that $q_\Sigma^\prime - p_\Gamma$ is trace class.
\end{proof}

\begin{Theorem}\label{thmHardyGrres}
Let $\Sigma$ be a compact Riemann surface,  let $\gamma$ be a family of boundary trivializations for $\Sigma$, and let $q_\Sigma$ be the projection of $H_\Gamma$ onto $H^2(\Sigma, \gamma)$. Then $q_\Sigma - p_\Gamma$ is trace class.
\end{Theorem}
\begin{proof}
Fix $\rho$ and $q$, and let $C$ be the corresponding Cauchy transform for $\Sigma$.
For $j \in \pi_0(\Gamma)$, let $p_j:H_\Gamma \to L^2(S^1)$ be the projection from $H_\Gamma$ onto the copy of $L^2(S^1)$ indexed by $j$.
We will show 
\begin{enumerate}
\item \label{itmDiagonalOutgoing} $p_j C p_j^* - p$ is trace class when $j \in \pi_0(\Gamma^1)$,
\item \label{itmDiagonalIncoming} $p_j C p_j^* - (\Id -p)$ is trace class when $j \in \pi_0(\Gamma^0)$,
\item \label{itmOffDiagonal} $p_j C p_k^*$ is trace class when $j,k \in \pi_0(\Gamma)$ and $j \ne k$.
\end{enumerate}

The statement of condition \eqref{itmOffDiagonal} is clear, since $p_j C p_k^*$ is an integral operator with smooth kernel.

We now consider condition \eqref{itmDiagonalOutgoing}. Let $j \in \pi_0(\Gamma^1)$, and let $K_j$ be a closed annulus in $\Sigma$ with one boundary component $j$. There is an annulus
$$
\A = \{ z \in \C : 1-\epsilon \le \abs{z} \le 1\} \subset \C
$$
such that we can find a biholomorphic map $g_j:\A \to K_j$. 
By Proposition \ref{propReparametrizeHardyHS}, the conclusion of the theorem is independent of the choice of $\gamma$, so we may assume without loss of generality that $\gamma_j = g_j|_{S^1}$.

There is a Cauchy transform $\CC_\A$ for $\A$ coming from the holomorphic immersion $\rho \circ g_j$ and Cauchy kernel 
$q(g_j(z),g_j(w))$.
Let $\Gamma_\A$ be the boundary of $\A$, and parametrize $\Gamma_\A$ via the identity map on the boundary component $S^1$, and arbitrarily on the other component.
Conjugating by these parametrizations, we get a Cauchy transform 
$$
C_A \in \B\left(\bigoplus_{\pi_0(\Gamma_\A)} L^2(S^1)\right)=:\B(H_{\Gamma_\A}).
$$
By construction, we have
$$
p_j C p_j^* = p_{S^1} C_\A p_{S^1}^*,
$$
where $p_{S^1} : H_{\Gamma_\A} \to L^2(S^1)$ is the projection onto the copy of $L^2(S^1)$ indexed by the boundary component $S^1$ of $\A$.

On the other hand, we have the standard Cauchy transform $C_{st}$ on $\A$ given by the standard Cauchy kernel $\frac{1}{w-z}$, and the same parametrizations used before to define $C_\A$.
By Corollary \ref{corCauchyIndependent}, $C_\A - C_{st}$ is trace class. 
Hence $p_j C p_j^* - p_{S^1} C_{st} p_{S^1}^*$ is trace class as well. 
But $p_{S^1}C_{st}p_{S^1}^*$ is just the projection onto the standard Hardy space $H^2(\D)$. 
Hence $p_j C p_j^* - p$ is trace class, as desired.

If $j \in \pi_0(\Gamma^0)$, we can establish \eqref{itmDiagonalIncoming} using essentially the same argument. The only modification is that we identify an annular neighborhood of $j$ with 
$$
\A^\prime = \{ z \in \C : 1 \le \abs{z} \le 1+\epsilon\}.
$$
\end{proof}

We now prove Theorem \ref{thmHardyPerp} and Theorem \ref{thmProjectionDifferenceTraceClass} by applying the preceding results to Hardy spaces coming from spin structures. We restate the theorems here for the convenience of the reader.

\begin{Theorem*}[Theorem \ref{thmHardyPerp}]
Let $X=(\Sigma, L, \Phi, \beta) \in \cR$ be a Riemann spin surface with boundary parametrization. 
Let $H_\Gamma = \bigoplus_{j \in \pi_0(\Gamma)} L^2(S^1)$ and let $H^2(X) \subset H_\Gamma$ be the Hardy space. 
Then 
\begin{equation}\label{eqHardyPerpRepeat}
H^2(X)^\perp = M_{\pm} \overline{M_z^{NS} H^2(X)} = M_{\pm} H^2(\overline{X}).
\end{equation}

Here $M_{\pm}$ is multiplication by $1$ on copies of $L^2(S^1)$ indexed by outgoing boundary components, and multiplication by $-1$ on copies of $L^2(S^1)$ indexed by incoming boundary components, and $M_z^{NS}$ is multiplication by the function $z$ on copies of $L^2(S^1)$ indexed by $j$ for which $L|_j$ is Neveu-Schwarz, and the identity on other boundary components.
\end{Theorem*}
\begin{proof}[Proof of Theorem \ref{thmHardyPerp}]
The second equality of \eqref{eqHardyPerpRepeat} is Proposition \ref{propConjugateHardy}, and so we only need to establish the first.

By Theorem \ref{thmTrivialityOfVectorBundles}, there exists a non-vanishing holomorphic section $F$ of $L$, and we denote the corresponding boundary values by $f := \beta^*F \in H^2(X)$. We then have 
\begin{equation}\label{eqnRelateTwoHardys}
H^2(X) = f H^2(\Sigma, \gamma) := \{f h : h \in H^2(\Sigma, \gamma)\},
\end{equation}
where $\gamma_j = \beta_j|_{S^1}$ is the boundary parametrization of $\Sigma$ given by $\beta$.
We then have 
$$
H^2(X)^\perp = (fH^2(\Sigma, \gamma))^\perp = \overline{f}^{-1} H^2(\Sigma, \gamma)^\perp.
$$
Applying Theorem \ref{thmWPerp} we get
\begin{equation}\label{eqIntermediateHardyPerp}
H^2(X)^\perp = M_{\pm} \overline{f^{-1} r z H^2(\Sigma, \gamma)}
\end{equation}
where $M_z$ is multiplication by $z$ on each copy of $L^2(S^1)$ in $H_\Gamma$, and $r$ is characterized by $r \; dz = \gamma^*R$ for any non-vanishing holomorphic section $R$ of $K_\Sigma$. 

In particular, we can take $R = i\Phi_*(F \otimes F)$.
Let $j \in \pi_0(\Gamma)$.
We will now show that 
\begin{equation}\label{eqSpinPullback}
\gamma_j^*R = \left\{
\begin{array}{ll}
f_j^2 dz \quad & \sigma(j) = NS\\
f_j^2 \frac{dz}{z} & \sigma(j) = R.
\end{array}
\right.
\end{equation}
Once we establish \eqref{eqSpinPullback}, then the desired result easily follows. Indeed, since $r\; dz = \gamma^* R$, we can rewrite \eqref{eqSpinPullback} as
\begin{equation}\label{eqCondensedSpinPullback}
r = M_z^{NS} z^{-1} f^2,
\end{equation}
Now \eqref{eqHardyPerpRepeat} follows from plugging \eqref{eqCondensedSpinPullback} into \eqref{eqIntermediateHardyPerp}.

We now turn to establishing \eqref{eqSpinPullback}. Recall that $\beta_j:(S^1, \sigma(j)) \to L|_j$ is an isomorphism of spin structures, and that $\gamma_j$ is the restriction of $\beta_j$ to the base space $S^1$. Since $\beta_j:(S^1,\sigma(j) \to (\Phi|_j, L|_j)$ is a spin isomorphism, by definition \eqref{eqnCircleSpinIso} we have
\begin{equation}\label{eqSpinIsoDefHardyProof}
\gamma_j^* \Phi_*(F \otimes F) = (\phi_{\sigma(j)})_*(\beta_j^*F \otimes \beta_j^* F) = (\phi_{\sigma(j)})_* (f_j \otimes f_j).
\end{equation}
Recall that we defined the spin structure $\phi_{\sigma(j)}$ in \eqref{eqCirclSpinDefinitions} so that
\begin{equation}\label{eqSpinIsoDefHardyProof2}
i(\phi_{\sigma(j)})_*(f_j \otimes f_j) = \left\{ \begin{array}{ll}
f_j^2 dz \quad & \sigma(j) = NS \vspace{0.05in}\\
f_j^2 \frac{dz}{z} & \sigma(j) = R.
\end{array}
\right..
\end{equation}
Combining \eqref{eqSpinIsoDefHardyProof} and \eqref{eqSpinIsoDefHardyProof2} yields \eqref{eqSpinPullback} and completes the proof.
\end{proof}

\begin{Theorem*}[Theorem \ref{thmProjectionDifferenceTraceClass}]
Let $X=(\Sigma, L, \Phi, \beta) \in \cR$ be a spin Riemann surface with boundary parametrization. 
Let $H_\Gamma = \bigoplus_{j \in \pi_0(\Gamma)} L^2(S^1)$ and let $H^2(X) \subset H_\Gamma$ be the Hardy space. 
Let $q_X \in \P(H_\Gamma)$ be the projection onto $H^2(X)$, and let 
$$
p_\Gamma = \bigoplus_{j \in \pi_0(\Gamma^1)} p \oplus \bigoplus_{j \in \pi_0(\Gamma^0)} \Id - p,
$$
where $p \in \P(L^2(S^1))$ is the projection onto $H^2(\D)$. Then $q_X - p_\Gamma$ is trace class.
\end{Theorem*}
\begin{proof}[Proof of Theorem \ref{thmProjectionDifferenceTraceClass}]
Let $\gamma_j = \beta_j|_{S^1}$ be the parametrization of $\Gamma$ induced by $\beta$. Let $q_\Sigma$ be the projection onto $H^2(\Sigma, \gamma)$. By Theorem \ref{thmHardyGrres}, $q_\Sigma - p_\Gamma$ is trace class, so we just need to show that $q_X - q_\Sigma$ is trace class. Let $F$ be a non-vanishing section $L$, and let $f = \beta^*F \in H^2(X)$ be the corresponding element of the Hardy space. Note that $f$ is a smooth function on $\bigsqcup S^1$.

We have $H^2(X) = f H^2(\Sigma, \gamma)$, so $f q_\Sigma f^{-1}$ is an idempotent whose range projection is $q_X$. By Proposition \ref{propGLresPreservesGrres}, to prove that $q_X - q_\Sigma$ is trace class, it suffices to prove that $[f, q_\Sigma]$ is trace class. This is done in \cite[Prop. 6.3.1]{PrSe86}.
\end{proof}

\bibliographystyle{alpha}
\bibliography{../ffbib.bib} 

\end{document}